\documentclass[a4paper,12pt]{amsart}
\pdfoutput=1

\usepackage[sc]{mathpazo}
\usepackage{stmaryrd}
\usepackage{xy}
\usepackage[T1]{fontenc}
\usepackage{subfig}
\usepackage{graphicx}
\usepackage{bm}
\usepackage{amsmath}
\usepackage{amssymb}
\usepackage[active]{srcltx}
\usepackage{hyperref}
\usepackage{enumerate}
\usepackage{bbm}

\usepackage{esint}

\def\ointccw{\ointctrclockwise}

\usepackage{color}
\definecolor{red}{rgb}{1,0,0}

\definecolor{gre}{rgb}{0,0.7,0}

\definecolor{blu}{rgb}{0,0,1}

\newtheorem{thm}{Theorem}
\newtheorem{assump}{Assumption}%
\newtheorem{lem}{Lemma}
\newtheorem{prop}{Proposition}
\theoremstyle{definition}

\theoremstyle{remark}
\theoremstyle{plain}

\numberwithin{equation}{section}

\def\CC{{\mathbb C}}
\def\DD{{\mathbb D}}
\def\EE{{\mathbb E}}

\def\GG{{\mathbb G}}

\def\KK{{\mathbb K}}

\def\NN{{\mathbb N}}

\def\RR{{\mathbb R}}
\def\TT{{\mathbb T}}

\def\WW{{\mathbb W}}
\def\ZZ{{\mathbb Z}}

\def\vecb{{\text{\boldmath$b$}}}

\def\vece{{\text{\boldmath$e$}}}

\def\veck{{\text{\boldmath$k$}}}

\def\vecm{{\text{\boldmath$m$}}}

\def\vecn{{\text{\boldmath$n$}}}

\def\vecp{{\text{\boldmath$p$}}}

\def\vecs{{\text{\boldmath$s$}}}

\def\vecu{{\text{\boldmath$u$}}}

\def\vecx{{\text{\boldmath$x$}}}

\def\vecy{{\text{\boldmath$y$}}}

\def\vecz{{\text{\boldmath$z$}}}

\def\vecalpha{{\text{\boldmath$\alpha$}}}
\def\vecalf{{\text{\boldmath$\alpha$}}}
\def\vecbeta{{\text{\boldmath$\beta$}}}

\def\veceta{{\text{\boldmath$\eta$}}}

\def\vectheta{{\text{\boldmath$\theta$}}}

\def\vecxi{{\text{\boldmath$\xi$}}}

\def\vecnull{{\text{\boldmath$0$}}}

\def\scrA{{\mathcal A}}
\def\scrB{{\mathcal B}}
\def\scrC{{\mathcal C}}

\def\scrI{{\mathcal I}}

\def\scrH{{\mathcal H}}
\def\scrJ{{\mathcal J}}
\def\scrK{{\mathcal K}}
\def\scrL{{\mathcal L}}

\def\scrP{{\mathcal P}}
\def\scrR{{\mathcal R}}
\def\scrS{{\mathcal S}}
\def\scrT{{\mathcal T}}

\def\scrW{{\mathcal W}}

\def\Re{\operatorname{Re}}
\def\Im{\operatorname{Im}}

\def\e{\mathrm{e}}
\def\i{\mathrm{i}}
\def\d{\mathrm{d}}
\def\off{\mathrm{off}}

\def\nc{\mathrm{nc}}

\def\diag{\operatorname{diag}}

\def\id{\operatorname{id}}

\def\LB{\mathrm{LB}}
\def\L{\operatorname{L{}}}

\def\Op{\operatorname{Op}}

\def\T{\operatorname{T{}}}

\def\tot{{\operatorname{tot}}}
\def\Tr{\operatorname{Tr}}

\def\vol{\operatorname{vol}}

\def\HiS{{\operatorname{HS}}}

\def\F{\underline{F}}
\def\Fura{\underrightarrow{F}}

\def\one{\mathbbm{1}}

\def\Green{G}

\title[Quantum transport in a crystal]{Quantum transport in a crystal with short-range interactions: The Boltzmann-Grad limit}

\author{Jory Griffin}
\address{Jory Griffin, Department of Mathematics,
University of Oklahoma,
Norman, OK 73019-3103, USA}
\email{\tt j.griffin@ou.edu}

\author{Jens Marklof}
\address{Jens Marklof, School of Mathematics, University of Bristol, Bristol BS8 1TW, U.K.} 
\email{\tt j.marklof@bristol.ac.uk}

\thanks{Research supported by EPSRC grant EP/S024948/1}

\date{\today}

\begin{document}

\begin{abstract}
We study the macroscopic transport properties of the quantum Lorentz gas in a crystal with short-range potentials, and show that in the Boltzmann-Grad limit the quantum dynamics converges to a random flight process which is not compatible with the linear Boltzmann equation. Our derivation relies on a hypothesis concerning the statistical distribution of lattice points in thin domains, which is closely related to the Berry-Tabor conjecture in quantum chaos.
\end{abstract}

\maketitle

\section{Introduction}

In 1905, Lorentz \cite{Lorentz05} introduced a kinetic model for electron transport in metals, which he argued should in the limit of low scatterer density be described by the linear Boltzmann equation. Although Lorentz' paper predates the discovery of quantum mechanics, the {\em Lorentz gas} has since served as a fundamental model for chaotic transport in both the classical and quantum setting, with applications to radiative transfer, neutron transport, semiconductor physics, and other models of transport in low-density matter. There has been significant progress in the derivation of the linear Boltzmann equation from first principles in the case of classical transport, starting from the pioneering works \cite{Gallavotti69,Spohn78,Boldrighini83} for random scatterer configurations, to the more recent derivation of new, generalised kinetic transport equations that highlight the limited validity of the Boltzmann equation for periodic \cite{Caglioti10,partII} and other aperiodic scatterer configurations \cite{MS2019}. 

In the quantum setting, the only complete derivation of the linear Boltzmann equation in the low-density limit is for random scatterer configurations \cite{EngErdos}, which followed analogous results in the weak-coupling limit \cite{Spohn77,ErdosYau}. The theory of quantum transport in periodic potentials on the other hand is a well developed theory of condensed matter physics. The general consensus is transport in periodic potentials without the presence of any disorder is {\em ballistic}, that is particles move almost freely, with minimal interaction with the scatterers; there is no diffusion. In this paper we propose that this picture changes in the low-density limit, where under suitable rescaling of space and time units the quantum dynamics is asymptotically described by a random flight process with strong scattering, similar to the setting of random potentials in the work of Eng and Erd\"os  \cite{EngErdos}. Our work is motivated by Castella's important studies \cite{Castella_WC,Castella_LD,Castella_LD2} of both the weak-coupling and low-density limits for periodic potentials in the case of zero Bloch vector. Castella shows that the weak-coupling limit gives rise to a linear Boltzmann equation with memory \cite{Castella_WC}. The low-density limit on the other hand diverges \cite{Castella_Plagne_LD}, and only the introduction of physically motivated off-diagonal damping terms leads to a limit, which for small damping is compatible with the linear Boltzmann equation. As we will show here, the case of random or generic Bloch vector does not diverge in the Boltzmann-Grad limit, without any requirement for damping, but the limit process differs significantly from that described by the linear Boltzmann equation. Our results complement our recent paper \cite{GM} which establishes convergence rigorously up to second order in perturbation theory. The aim of the present paper is thus to give a derivation of all higher order terms and identify the full random flight process, conditional on an assumption on the distribution of lattice points in a particular scaling limit. A rigorous verification of this hypotheses seems currently out of reach. 

We will here focus on the case when the particle wavelength $h$ is comparable to the potential range $r$, and much smaller than the fundamental cell of the lattice.  This choice of scaling means that a wave-packet will evolve semiclassically far away from the scatterers, but that any interaction with the potential is truly quantum. This is a scaling not traditionally discussed in homogenisation theory in which one usually assumes the characteristic wavelength is either much larger than the period (low-frequency homogenisation) or of the same or smaller order (high-frequency homogenisation); see for example \cite{Allaire05,BenoitGloria,BLP,Birman,Craster10,Gerard91,Gerard97,HMC16,Markowich94,Panati03}. Our scaling is also different from that leading to the classic point scatterer (or $s$-wave scatterer, Fermi pseudo-potential), where the potential scale $r$ is taken to zero with an appropriate renormalisation of the potential strength. In contrast to the setting of smooth finite-range potentials discussed in the present paper, periodic (and other) superpositions of point scatterers are exactly solvable \cite{Albeverio,AlbeverioGesztesyHoegh-Krohn,AlbeverioHoegh-Krohn,ExnerSeba,Grossmann, HH-KJ}.

Our set-up is as follows. We assume throughout that the space dimension $d$ is three or higher. Consider the Schr\"{o}dinger equation
\begin{equation}
\frac{\i h}{2 \pi} \partial_t \psi(t,\vecx) = H_{h,\lambda} \psi(t,\vecx)
\end{equation}
where 
\begin{equation}\label{Hamiltonian}
H_{h,\lambda}=-\frac{h^2}{8\pi^2}\; \Delta + \lambda \Op(V) ,
\end{equation}
$\Delta$ is the standard $d$-dimensional Laplacian and $\Op(V)$ denotes multiplication by the $\scrL$-periodic potential
\begin{equation}\label{Vdef}
V(\vecx) = \sum_{\vecb\in\scrL}  W(r^{-1} (\vecx-\vecb)),
\end{equation}
with $W$ the single-site potential, scaled by $r>0$, and $\scrL$ a full-rank Euclidean lattice in $\RR^d$. We re-scale space units by a constant factor so that the co-volume of $\scrL$ (i.e.\ the volume of its fundamental cell) is one. (One example to keep in mind is the cubic lattice $\scrL = \ZZ^d$.) The {\em coupling constant} $\lambda>0$ will remain fixed throughout, and $h$ is a scaling parameter which measures the characteristic wave lengths of the quantum particle. We will assume throughout that $h$ is comparable with the potential scaling $r$.

We assume that $W$ is in the Schwartz class $\scrS(\RR^d)$ and real-valued. This short-range assumption on the potential is key -- a small wavepacket moving through this potential should experience long stretches of almost free evolution, followed by occasional interactions with localised scatterers. One could also consider adding an external potential living on the macroscopic scale although we will not pursue that idea here.

We denote by
\begin{equation}\label{Hamiltonian-loc}
H_\lambda^{\text{loc}}=-\frac{1}{8\pi^2}\; \Delta + \lambda \Op(W)
\end{equation}
the single-scatterer Hamiltonian for the unscaled potential $W$ with $h=1$. Its resolvent is denoted by $\Green_\lambda(E)=(E-H_\lambda^{\text{loc}})^{-1}$, and the corresponding $T$-operator is defined as
\begin{equation}\label{def:TE}
T(E)=\lambda \Op(W) + \lambda^2 \Op(W)\, \Green_\lambda(E)\, \Op(W).
\end{equation}

Rather than consider solutions of the Schr\"{o}dinger equation directly, we instead consider the time evolution of a quantum observable $A$ given by the Heisenberg evolution
\begin{equation}\label{At}
A(t)=U_{h,\lambda}(t)\, A\, U_{h,\lambda}(-t) .
\end{equation}
Here $U_{h,\lambda}(t)= \e^{-\frac{2\pi\i}{h} H_{h,\lambda} t}$ is the propagator corresponding to the Hamiltonian $H_{h,\lambda}$. 

Let us now take an observable $A=\Op(a)$ given by the quantisation of a classical phase-space density $a=a(\vecx,\vecy)$, with $\vecx$ denoting particle position and $\vecy$ momentum. The question now is whether, in the low density limit $r= h\to 0$ and with the appropriate rescaling of length and time units, the phase-space density of the time-evolved quantum observable $A(t)$ (i.e. its principal symbol) can be described asymptotically by a function $f(t,\vecx,\vecy)$ governed by a random flight process. Eng and Erd\"os \cite{EngErdos} confirmed this in the case of random scatterer configurations, and established that the density $f(t,\vecx,\vecy$) is a solution of the linear Boltzmann equation,
\begin{equation}\label{KINETICEQ008}
\begin{cases}
\displaystyle
\bigl(\partial_t+\vecy\cdot\nabla_\vecx\bigr) f(t,\vecx,\vecy) 
=
\int_{\RR^d}  \big[ \Sigma(\vecy,\vecy') f(t,\vecx,\vecy') -  \Sigma(\vecy',\vecy) f(t,\vecx,\vecy) \big] \,d\vecy' & \\
 f(0,\vecx,\vecy)=a(\vecx,\vecy) & 
 \end{cases}
\end{equation}
where 
$\Sigma(\vecy,\vecy')$ is the collision kernel of the single site potential $W(\vecx)$ and $\vecy'$ and $\vecy$ denote the incoming and outgoing momenta, respectively. The collision kernel is given by the formula
\begin{equation}\label{Sigma001}
\Sigma(\vecy,\vecy') =4 \pi^2 \, |T(\vecy,\vecy')|^2 \delta(\tfrac12 \|\vecy\|^2-\tfrac12 \|\vecy'\|^2) 
\end{equation}
where the {\em $T$-matrix} $T(\vecy,\vecy')$ is the kernel of $T(E)$ in momentum representation, with $E=\tfrac12 \|\vecy\|^2$ (``on-shell'').
The {\em total} scattering cross section is defined as
\begin{equation}\label{Sigmatot}
\Sigma_\tot(\vecy) = \int_{\RR^d} \Sigma(\vecy',\vecy) \, \d \vecy'.
\end{equation} 
Solutions of the linear Boltzmann equation can be written in terms of the collision series
\begin{equation}\label{collseries}
f_{\LB}(t,\vecx,\vecy) = \sum_{k=1}^\infty f_{\LB}^{(k)}(t,\vecx,\vecy) 
\end{equation}
with the zero-collision term
\begin{equation}\label{collseries0}
f_{\LB}^{(1)}(t,\vecx,\vecy) = a( \vecx - t \vecy,\vecy)\, \e^{- t \Sigma_\tot(\vecy)} ,
\end{equation}
and the $(k-1)$-collision term
\begin{multline}\label{collseriesk}
f_{\LB}^{(k)}(t,\vecx,\vecy) = \int_{(\RR^d)^k}\int_{\RR_{\geq 0}^k}
\delta(\vecy-\vecy_1)\, a\bigg( \vecx -\sum_{j=1}^k  u_j \vecy_j ,\vecy_k \bigg) \\
\times \rho_{\LB}^{(k)}(\vecu,\vecy_1,\ldots,\vecy_k)\, \delta\bigg(t-\sum_{j=1}^k u_j\bigg) \,  \d \vecu\, \d\vecy_1\cdots\d\vecy_k 
\end{multline}
with 
\begin{equation}\label{colldensityBG}
 \rho_{\LB}^{(k)}(\vecu,\vecy_1,\ldots,\vecy_k)= \prod_{i=1}^k \e^{-u_i \Sigma_\tot(\vecy_i)}\; \prod_{j=1}^{k-1} \Sigma(\vecy_j,\vecy_{j+1}) .
\end{equation}
The product form of the density $\rho_{\LB}^{(k)}$ shows that the corresponding random flight process is Markovian, and describes a particle moving along a random piecewise linear curve with momenta $\vecy_i$ and exponentially distributed flight times $u_i$.

The principal result of this paper is that, for periodic potentials of the form \eqref{Vdef} and using the same scaling as in the random setting \cite{EngErdos}, there exists a limiting random flight process describing macroscopic transport. The derivation requires a hypothesis on the fine-scale distribution of lattice points which is discussed in detail as Assumption \ref{hyp0} in Section \ref{secPoisson}.

\begin{thm}[Main Result]
Under Assumption \ref{hyp0} on the distribution of lattice points, there exists an evolution operator $L(t)$, distinct from that of the linear Boltzmann equation, such that for any $a\in\scrS(\RR^d\times\RR^d)$ we have
\begin{equation}
\| U_{h,\lambda}(t r^{1-d})\, \Op_{r,h}(a)\, U_{h,\lambda}(-t r^{1-d}) - \Op_{r,h}(L(t) a) \|_\HiS \to 0
\end{equation}
where $\Op_{r,h}(a)$ is the Weyl quantisation of the phase-space symbol $a$ in the Boltzmann-Grad scaling and $\| \cdot \|_{\HiS}$ is the Hilbert-Schmidt norm.
\end{thm}

The precise scaling of the quantum observable $\Op_{r,h}(a)$ is explained in Section \ref{secScaling}. The limiting evolution operator $L(t)$ is given by the series
\begin{equation}\label{collseriesLt}
L(t)a(\vecx,\vecy)= f(t,\vecx,\vecy) = \sum_{k=1}^\infty f^{(k)}(t,\vecx,\vecy) ,
\end{equation}
where $f^{(k)}$ coincides with $f_{\LB}^{(k)}$ for $k=1$,
\begin{equation}\label{collseries001}
f^{(1)}(t,\vecx,\vecy) = a( \vecx - t \vecy,\vecy)\, \e^{- t \Sigma_\tot(\vecy)} ,
\end{equation}
but deviates significantly at higher order. For $k\geq 2$, the $(k-1)$-collision term is given by
\begin{multline}\label{collserieskOUR}
f^{(k)}(t,\vecx,\vecy) =  \frac{1}{k!}  \sum_{\ell,m=1}^k \int_{(\RR^d)^k}\int_{\RR_{\geq 0}^k}
\delta(\vecy-\vecy_\ell)\,  a\bigg( \vecx -\sum_{j=1}^k  u_j \vecy_j ,\vecy_m \bigg) \\
\times \rho_{\ell m}^{(k)}(\vecu,\vecy_1,\ldots,\vecy_k)\, \delta\bigg(t-\sum_{j=1}^k u_j\bigg) \, \d \vecu\, \d\vecy_1\cdots\d\vecy_k ,
\end{multline}
with the collision densities
\begin{equation} \label{rhokalt_intro} 
\rho_{\ell m}^{(k)}(\vecu,\vecy_1,\dots,\vecy_k)  = \big| g_{\ell m}^{(k)}(\vecu,\vecy_1,\dots,\vecy_k) \big|^2
\,\omega_k(\vecy_1,\ldots,\vecy_k)\, \prod_{i=1}^{k} \e^{- u_i \Sigma_\tot(\vecy_i)} .
\end{equation}
Here 
\begin{equation} 
\omega_k(\vecy_1,\ldots,\vecy_k)= \prod_{j=1}^{k-1} \delta\big(\tfrac12 \| \vecy_j\|^2-\tfrac12\|\vecy_{j+1}\|^2\big) 
\end{equation}
and $g_{\ell m}^{(k)}$ are the coefficients of the matrix valued function
\begin{equation} \label{rhokdfinalGdef0}
\GG^{(k)}(\vecu,\vecy_1,\dots,\vecy_k) 
= \frac{1}{(2\pi\i)^k} \ointccw\cdots \ointccw \big( \DD(\vecz)- \WW \big)^{-1}  \exp(\vecu\cdot\vecz)\, \d z_1 \cdots \d z_k ,
\end{equation}
where $\DD(\vecz)=\diag(z_1,\ldots,z_k)$ and $\WW = \WW (\vecy_1,\dots,\vecy_k)$ with entries
\begin{equation}
w_{ij} = 
\begin{cases}
0 & (i=j) \\
-2\pi\i T(\vecy_i,\vecy_j) & (i\neq j) .
\end{cases}
\end{equation}
The paths of integration in \eqref{rhokdfinalGdef0} are circles around the origin with radius strictly greater than $r_0= k \max |w_{ij}|$. The matrix $\GG^{(k)}(\vecu,\vecy_1,\dots,\vecy_k)$ is in fact the derivative $\partial_{u_1}\cdots\partial_{u_k}$ of the Borel transform of the function $F(\vecu)=(\DD(\vecu)^{-1} - \WW)^{-1}$. We furthermore note that the above formulas are independent of the choice of scatterer configuration $\scrL$, as in the classical setting.
For the one-collision terms we will furthermore derive the following explicit representation in terms of the Lorentz-Boltzmann density \eqref{colldensityBG} and $J$-Bessel functions,
\begin{equation} \label{rho112}
\rho_{1 1}^{(2)}(\vecu,\vecy_1,\vecy_2)  = \rho_\LB^{(2)}(\vecu,\vecy_1,\vecy_2) \, \bigg| \frac{u_1 T(\vecy_2,\vecy_1)}{u_2 T(\vecy_1,\vecy_2)}\bigg| \, \big| J_1\big(4\pi [u_1 u_2 T(\vecy_1,\vecy_2) T(\vecy_2,\vecy_1)]^{1/2} \big) \big|^2 .
\end{equation}
and
\begin{equation} \label{rho122}
\rho_{1 2}^{(2)}(\vecu,\vecy_1,\vecy_2)  = \rho_\LB^{(2)}(\vecu,\vecy_1,\vecy_2) \, \big| J_0\big(4\pi [u_1 u_2 T(\vecy_1,\vecy_2) T(\vecy_2,\vecy_1)]^{1/2} \big) \big|^2 
\end{equation}
The remaining matrix elements can be computed via the identities
\begin{equation}
\rho_{22}^{(2)}(u_1,u_2,\vecy_1,\vecy_2) = \rho_{11}^{(2)}(u_2,u_1,\vecy_2,\vecy_1),
\end{equation}
\begin{equation}
\rho_{21}^{(2)}(u_1,u_2,\vecy_1,\vecy_2) = \rho_{12}^{(2)}(u_2,u_1,\vecy_2,\vecy_1) .
\end{equation}
A notable difference with the solution \eqref{collseriesk} to the linear Boltzmann equation is that in \eqref{collserieskOUR} there is a non-zero probability that the final momentum $\vecy_\ell$ is equal to the initial momentum $\vecy_m$.

The paper is organised as follows. We will first explain in Section \ref{secScaling} the precise scaling needed to observe our limiting process and state the main result. In Section \ref{secFloquet} we recall the well-known Floquet-Bloch decomposition for periodic potentials and in Section \ref{secT} we recall an explicit formula for the $T$-operator in our specific setting. Section \ref{secDuhamel} explains the perturbative approach to calculate the series expansion for the time evolution of $A(t)$. This is followed by a discussion of the main hypothesis in this study in Section \ref{secPoisson} which in brief can be viewed as a phase-space generalisation of the Berry-Tabor conjecture for the statistics of quantum energy levels for integrable systems \cite{BerryTabor,Marklof00}. In Section \ref{sec:explicit} we provide an explicit computation of terms appearing in the formal series, and in Section \ref{sec:decay} we prove that the series is absolutely convergent provided $\lambda$ is small enough. In Section \ref{sec:takingthelimit} we take the low-density limit using the formulas from Section \ref{sec:explicit} and show how the limiting object can be written in terms of the $T$-operator described in Section \ref{secT}. In Section \ref{secpositivity} we establish positivity of this limiting expansion and derive the formulas for \eqref{collserieskOUR}. A key observation is that the one-collision term is distinctly different from the corresponding term for the linear Boltzmann equation. We conclude the paper with a discussion and outlook in Section \ref{secDiscussion}. The appendix provides detailed background of the combinatorial structures used in this paper.

\subsection*{Acknowledgements} 

We thank S\o ren Mikkelsen for valuable comments on the first draft of this paper.

\section{Microlocal Boltzmann-Grad scaling}\label{secScaling}

The phase space of the underlying classical Hamiltonian dynamics is $\T(\RR^d)=\RR^d\times\RR^d$, where the first component parametrises the position $\vecx$ and the second the momentum $\vecy$ of the particle. Given a function $a:\T(\RR^d)\to\RR$, we associate with it the observable $A=\Op(a)$  acting on functions $f\in \scrS(\RR^d)$ through the Weyl quantisation
\begin{equation}\label{def:Opa}
\Op(a) f(\vecx) = \int_{\T(\RR^d)} a(\tfrac12(\vecx+\vecx'),\vecy) \, \e((\vecx-\vecx')\cdot\vecy)\, f(\vecx')\, \d \vecx' \d \vecy ,
\end{equation}
where we have used the shorthand $\e(z)=\e^{2\pi\i z}$. The Weyl quantisation is useful to capture the phase space distribution of quantum states. In the case of free quantum dynamics, with $\lambda=0$ and $h=1$, we have for example the well-known quantum-classical correspondence principle 
\begin{equation} \label{egorov}
U_{1,0}(t) \Op(a) U_{1,0}(-t) = \Op(L_0(t)a)
\end{equation}
with the classical free evolution $[L_0(t) a](\vecx,\vecy)=a(\vecx-t\vecy,\vecy)$.
It is convenient to incorporate the scaling parameter $h>0$ in \eqref{Hamiltonian} by setting 
\begin{equation}\label{def:Opah}
\Op_h(a) f(\vecx) = h^{-d/2} \int_{\T(\RR^d)} a(\tfrac12(\vecx+\vecx'),\vecy) \, \e_h((\vecx-\vecx')\cdot\vecy)\, f(\vecx')\, \d \vecx' \d \vecy 
\end{equation}
with $\e_h(z)=\e^{\tfrac{2\pi\i}{h} z}$. Note that we have $\Op_h(a)=\Op(D_{1,h} a)$ for $D_{1,h} a(\vecx,\vecy) = h^{d/2} \, a(\vecx, h \vecy)$. We refer to $D_{1,h}$ as the {\em microlocal scaling}. In particular, \eqref{egorov} becomes
\begin{equation} \label{egorovh}
U_{h,0}(t) \Op_h(a) U_{h,0}(-t) = \Op_h(L_0(t)a), 
\end{equation}

The mean free path length of a particle travelling in a potential of the form \eqref{Vdef} is asymptotic (for $r$ small) to the inverse total scattering cross section of the single-site potential $W$ \cite{Marklof15,MS2019}; the total scattering cross section in turn equals $r^{d-1}$, up to constants. In the low-density it is natural to measure length units in terms of the mean free path lengths or, equivalently, in units of $r^{1-d}$. We refer to the corresponding scaling $D_{r,1}$ defined by $D_{r,1} a(\vecx,\vecy) = r^{d(d-1)/2} \, a( r^{d-1} \vecx, \vecy)$ as the {\em Boltzmann-Grad scaling}, and the combined scaling
\begin{equation}\label{BGscaling}
D_{r,h} a(\vecx,\vecy) =  r^{d(d-1)/2} h^{d/2} \, a( r^{d-1} \vecx, h \vecy),
\end{equation}
as the {\em microlocal Boltzmann-Grad scaling}.
We define the corresponding scaled Weyl quantisation by $\Op_{r,h}=\Op\circ D_{r,h}$. The quantum-classical correspondence \eqref{egorov} for the free dynamics reads in this scaling
\begin{equation} \label{egorovrescaled}
U_{h,0}(t r^{1-d}) \Op_{r,h}(a) U_{h,0}(-t r^{1-d}) = \Op_{r,h} (L_0(t)a).
\end{equation}
The key point here is that we require an extra scaling in time relative to the mean free path.
The challenge for the present study is thus to understand the asymptotics of $A(t r^{1-d})$ as in \eqref{At}, for every fixed $t>0$, with initial data $A=\Op_{r,h}(a)$ and $a$ in the Schwartz class $\scrS(\T(\RR^d))$ (i.e. $a$ is infinitely differentiable and all its derivatives decay rapidly as $\|\vecx\|,\|\vecy\|\to \infty$). The question is, more precisely, whether there is a family of linear operators $L(t)$ so that 
\begin{equation}\label{Ath}
\| U_{h,\lambda}(t r^{1-d})\, \Op_{r,h}(a)\, U_{h,\lambda}(-t r^{1-d}) - \Op_{r,h}(L(t) a) \|_\HiS \to 0
\end{equation}
in the Hilbert-Schmidt norm, defined as
\begin{equation}
\| X \|_\HiS = \langle X, X\rangle_\HiS^{1/2} , \qquad \langle X, Y\rangle_\HiS = \Tr (X^\dagger\; Y).
\end{equation}
To understand \eqref{Ath}, it is sufficient to establish the convergence of
\begin{equation}\label{Ath2}
\langle B, A(t r^{1-d})  \rangle_\HiS \to \langle b, L(t) a\rangle
\end{equation}
with $A(t)$ as in \eqref{At}, $A=\Op_{r,h}(a)$, $B=\Op_{r,h}(b)$, and $a,b\in\scrS(\T(\RR^d))$.
The inner product on the right hand side of \eqref{Ath2} is defined by
\begin{equation}\label{innerp}
\langle f, g\rangle = \int_{\T(\RR^d)} \overline{f(\vecx,\vecy)}\, g(\vecx,\vecy)\, \d \vecx \d \vecy.
\end{equation}
 We direct the reader towards \cite[Appendix A]{GM} for an explanation of how \eqref{Ath2} can be reformulated as a statement about solutions of the Schr\"{o}dinger equation. As mentioned previously, we will here restrict our attention to the case when $r$ is of the same order of magnitude as $h$, i.e. $r=h\, c_0$ for fixed effective scattering radius $c_0$. By adjusting $W$, we may in fact assume without loss of generality that $c_0=1$. This is precisely the scaling used in \cite{EngErdos} for the case of random potentials, although in a slightly different formulation in terms of Husimi functions for the phase-space presentation of quantum states.

\section{Floquet-Bloch decomposition}\label{secFloquet}

Floquet-Bloch theory allows us to reduce the quantum evolution in periodic potentials to invariant Hilbert spaces $\scrH_\vecalf$ of quasiperiodic functions $\psi$, satisfying
\begin{equation}\label{quasip}
\psi(\vecx+\vecb) = \e(\vecb\cdot\vecalf) \psi(\vecx) ,
\end{equation}
for all $\vecb \in \scrL$ where $\vecalf \in \TT^*=\RR^d/\scrL^*$ is the {\em quasimomentum} and
\begin{equation}
\scrL^*=\{ \veck\in\RR^d \mid \veck\cdot\vecb\in\ZZ \text{ for all } \vecb\in\scrL\}
\end{equation}
is the dual (or reciprocal) lattice of $\scrL$.
We denote by $\scrH_\vecalf$ the Hilbert space of such functions that have finite $\L^2$-norm with respect to the inner product
\begin{equation}
\langle \psi, \varphi  \rangle_\vecalf = \int_{\TT} \overline{\psi(\vecx)}\, \varphi(\vecx)\, \d \vecx ,
\end{equation}
with $\TT=\RR^d/\scrL$.
We define the corresponding Hilbert-Schmidt product for linear operators on $\scrH_\vecalf$ by
\begin{equation}
\langle X, Y  \rangle_{\HiS,\vecalf} = \Tr (X^\dagger\, Y).
\end{equation}
For a given quasi-momentum $\vecalf\in\TT^*$, consider the Bloch functions 
\begin{equation}
\varphi_\veck^\vecalf(\vecx)=\e((\veck+\vecalf)\cdot\vecx), \qquad \veck\in\scrL^*,
\end{equation}
and define the \emph{Bloch projection} $\Pi_\vecalf: \scrS(\RR^d) \to\scrH_\vecalf$ by
\begin{equation}\label{Pi-def}
\Pi_\vecalf f(\vecx) = \sum_{\veck\in\scrL^*} \langle \varphi_\veck^\vecalf, f\rangle \;  \varphi_\veck^\vecalf(\vecx) 
\end{equation}
with inner product
\begin{equation}
\langle f, g\rangle = \int_{\RR^d} \overline{f(\vecx)}\, g(\vecx)\, \d \vecx .
\end{equation}
Note that, by Poisson summation,
\begin{equation} \label{eq:Frep}
\Pi_\vecalf f(\vecx) = \sum_{\vecb\in\scrL} \e(\vecb\cdot\vecalf) f(\vecx-\vecb),
\end{equation}
and hence that by integrating over $\vecalf \in \TT^*$ one regains $f(\vecx)$. The kernel of $\Pi_\vecalf$ is thus
\begin{equation} \label{eq:Frep2}
\Pi_\vecalf (\vecx,\vecx') = \sum_{\vecb\in\scrL} \e(\vecb\cdot\vecalf) \delta_\vecb(\vecx-\vecx').
\end{equation}
Instead of \eqref{Ath2} the plan is now to consider the convergence
\begin{equation}\label{Ath3}
\langle B,\Pi_\vecalf A(t r^{1-d})  \rangle_\HiS \to \langle b,L(t) a\rangle
\end{equation}
for typical $\vecalf$. The advantage is that we are working in a Hilbert space with discrete basis. One can then obtain information on \eqref{Ath2} by integrating over $\vecalf$. In fact we will argue that the right hand side of \eqref{Ath3} is independent of $\vecalf$ for almost every $\vecalf$.

\section{The {\em T}-operator for a single scatterer}\label{secT}

Recall from \eqref{def:TE} that the $T$-operator for the single scatterer potential $W$ is defined by
\begin{equation}
T(E)=\lambda \Op(W) + \lambda^2 \Op(W)\, \Green_\lambda(E)\, \Op(W)  
\end{equation}
in the half-plane $\Im E>0$ where the resolvent $\Green_\lambda(E)$ is resonance-free, and then extended by analytic continuation.  
The Born series for $\Green_\lambda(E)$ leads to the formal series expansion
\begin{equation}
T(E) = \lambda \Op(W) \sum_{n=0}^\infty  (\lambda \Green_0(E)\Op(W))^n .
\end{equation}
Using $\varphi_\vecy(\vecx)=\e(\vecy\cdot\vecx)$ as a basis for the momentum representation, the free resolvent $\Green_0(E)$ has the kernel
\begin{equation}
\Green_0(\vecy,\vecy',E) = \langle \varphi_\vecy, \Green_0(E) \varphi_{\vecy'}\rangle = \frac{\delta(\vecy-\vecy')}{E - \tfrac12 \|\vecy\|^2}  ,
\end{equation}
and similarly $\Op(W)$ has kernel  $\langle \varphi_\vecy, \Op(W) \varphi_{\vecy'}\rangle=\hat W(\vecy-\vecy')$.
The $T$-matrix is defined as the kernel of $T(E)$ in momentum representation, i.e., 
\begin{equation}
T(\vecy,\vecy',E)=\langle \varphi_\vecy, T(E) \varphi_{\vecy'}\rangle.
\end{equation}
It will be convenient to set $E=\tfrac12\|\vecy\|^2+\i\gamma$, with $\Re\gamma\geq 0$, and define 
\begin{equation}\label{Tg}
T^\gamma(\vecy,\vecy')= T(\vecy,\vecy',\tfrac12\|\vecy\|^2+\i \gamma),  \qquad 
g^\gamma(\vecy,\vecy')= \frac{1}{\tfrac12\|\vecy\|^2 - \tfrac12 \|\vecy'\|^2 + \i\gamma} .
\end{equation}
The corresponding perturbation series is 
\begin{equation} \label{Tseries}
T^\gamma(\vecy,\vecy')= \sum_{n=1}^\infty \lambda^n T_n^\gamma(\vecy,\vecy'),
\end{equation}
where the $T_n^\gamma$ are defined by
\begin{equation}
T_1^\gamma(\vecy_0,\vecy_1) = \hat W(\vecy_0-\vecy_1)
\end{equation}
and for $n\geq 2$
\begin{multline} \label{DtoT0}
T_n^\gamma(\vecy_0,\vecy_n) = \int_{(\RR^d)^{n-1}} \hat W(\vecy_0-\vecy_1) \cdots \hat W(\vecy_{n-1}-\vecy_n)  \\ 
\times\bigg( \prod_{j=1}^{n-1}  g^\gamma(\vecy_0,\vecy_j) \bigg)
  \,\d\vecy_1\cdots \d\vecy_{n-1} .
\end{multline}
The analytic continuation of $T^\gamma$ from $\Re\gamma>0$ to the boundary $\Re\gamma =0$ is obtained via the integral representation
\begin{multline} \label{DtoT011}
T_n^\gamma(\vecy_0,\vecy_n)  = (-2 \pi \i)^{n-1}\int_{\RR_{\geq 0}^{n-1}} \bigg\{ \int_{(\RR^d)^{n-1}} \hat W(\vecy_0-\vecy_1) \cdots \hat W(\vecy_{n-1}-\vecy_n)  \\ 
\times e\bigg[\sum_{j=1}^{n-1} \theta_j \bigg( \tfrac12\|\vecy_0\|^2 - \tfrac12 \|\vecy_j\|^2 + \i\gamma \bigg) \bigg]
  \,\d\vecy_1\cdots \d\vecy_{n-1} \bigg\} \,\d\theta_1\cdots\d\theta_{n-1}.
\end{multline}
The choice $\gamma=0$ is referred to as ``on-shell''. We drop the syperscript if $\gamma=0$, i.e., $T(\vecy,\vecy')=T^0(\vecy,\vecy')$, $T_n(\vecy,\vecy')=T_n^0(\vecy,\vecy')$.
The $T$-matrix $T(\vecy,\vecy')$ is then related to the on-shell scattering matrix via
\begin{equation}
S(\vecy,\vecy') = \delta(\vecy-\vecy') - 2\pi \i \, \delta(\tfrac12 \|\vecy\|^2-\tfrac12 \|\vecy'\|^2)\,T(\vecy,\vecy') .
\end{equation}
The unitarity of the $S$-matrix is equivalent to the relation, for $\|\vecy\|=\|\vecy'\|$,
\begin{equation}
T(\vecy,\vecy')-\overline T(\vecy',\vecy) = -2\pi\i \int_{\RR^d} \delta(\tfrac12 \|\vecy\|^2-\tfrac12 \|\vecy''\|^2)\, T(\vecy,\vecy'') \overline T(\vecy',\vecy'') \d\vecy''.
\end{equation}
This in particular implies the optical theorem
\begin{equation} \label{opticaltheorem}
\Im T(\vecy,\vecy) = -\frac{1}{4\pi} \Sigma_\tot(\vecy).
\end{equation}

We will now prove that the integrals defining \eqref{DtoT011} converge uniformly in $\Re \gamma \geq 0$ in dimensions $d > 2$. 
For $f \in \scrS(\RR^{dk})$ and $S \subset \{1,\dots,k\}$ we denote by $f_S \in \scrS(\RR^{dk})$ the partial inverse Fourier transform of $f$ in the variables $\vecy_i$ for $i \in S$:
\begin{equation}
f_S(\vecy_1,\dots,\vecy_k) =  \int_{\RR^{dk}} f(\vecz_1,\dots,\vecz_k) [\prod_{i \in S} \e(\vecz_k \cdot \vecy_k)] \, [\prod_{i\notin S} \delta(\vecz_i -\vecy_i)] \, \d \vecz_1 \cdots \d \vecz_k.
\end{equation}
We use the notation $\langle x \rangle = \sqrt{1+x^2}$ and $\langle \vecx \rangle = \sqrt{1+\|\vecx\|^2}$.
\begin{lem} \label{lem51}
Let $f \in \scrS(\RR^{dk})$. Then,
\begin{multline}
\left|\int_{\RR^{dk}} f(\vecy_1,\dots,\vecy_k) \e(-\tfrac12\theta_1\|\vecy_1\|^2-\cdots-\tfrac12\theta_k\|\vecy_k\|^2) \, \d \vecy_1 \cdots \d \vecy_k \right| \\ \leq 2^{dk/4}\langle \theta_1 \rangle^{-d/2} \cdots \langle\theta_k\rangle^{-d/2} \, \sup_{S \subset \{1,\dots,k\} } \| f_S \|_{L^1}.
\end{multline}
\end{lem}
\begin{proof}
We partition $\RR^k$ into $2^k$ regions according to whether $|\theta_i| \leq 1$ or $|\theta_i|>1$. Take $S \subset \{1,\dots,k\}$ and assume that for $i \in S$, $|\theta_i|>1$ and for $i \notin S$, $|\theta_i| \leq 1$. We have that
\begin{multline}
\int_{\RR^{dk}} f(\vecy_1,\dots,\vecy_k) \e(-\tfrac12\theta_1\|\vecy_1\|^2-\cdots-\tfrac12\theta_k\|\vecy_k\|^2) \, \d \vecy_1 \cdots \d \vecy_k \\
= \int_{\RR^{dk}} \int_{\RR^{dk}} f_S(\veceta_1,\dots,\veceta_k)  [\prod_{i \in S} \e(\tfrac12 \theta_i^{-1} \|\veceta_i\|^2) \, \e(-\tfrac12 \theta_i \|\vecy_i + \theta_i^{-1} \veceta_i\|^2) ]  \\
\hfill \times [\prod_{i\notin S} \e(-\tfrac12 \theta_i \|\vecy_i\|^2) \delta(\vecy_i-\veceta_i)] \,\d \veceta_1 \cdots \d \veceta_k  \, \d \vecy_1 \cdots \d \vecy_k. 
\end{multline}
Therefore, using the identity
\begin{equation}
\left| \int_{\RR^d} \e(-\tfrac12 \theta_i \, \| \vecy_i\|^2) \, \d \vecy_i \right| = |\theta_i|^{-d/2}
\end{equation}
we obtain
\begin{multline}
\left|\int_{\RR^{dk}} f(\vecy_1,\dots,\vecy_k) \e(-\tfrac12\theta_1\|\vecy_1\|^2-\cdots-\tfrac12\theta_k\|\vecy_k\|^2) \, \d \vecy_1 \cdots \d \vecy_k \right|  \\
\leq  \bigg(\prod_{i \in S} |\theta_i|^{-d/2} \bigg) \| f_S\|_{L^1}.
\end{multline}
Taking a supremum over all $2^k$ regions we see
\begin{multline}
\left|\int_{\RR^{dk}} f(\vecy_1,\dots,\vecy_k) \e(-\tfrac12\theta_1\|\vecy_1\|^2-\cdots-\tfrac12\theta_k\|\vecy_k\|^2) \, \d \vecy_1 \cdots \d \vecy_k \right| \\
\leq \min\{1,|\theta_1|^{-d/2} \} \cdots \min \{1 , |\theta_k|^{-d/2} \}\sup_{S \subset \{1,\dots,k\}} \| f_S \|_{L^1}.
\end{multline}
The result then follows since 
\begin{equation}
\min \{1, |\theta_i|^{-d/2} \} = (\max\{1,|\theta_i|\})^{-d/2}
\end{equation}
and $\max\{1,x\} \geq \tfrac{1}{\sqrt{2}} \langle x \rangle$.
\end{proof}

Let us apply this Lemma in our situation, in particular to the inner integral in \eqref{DtoT011}. For multi-indices $\vecalf, \vecbeta$ we define 
$$ \vecx^\vecbeta = x_1^{\beta_1} \cdots x_d^{\beta_d}, \quad D^\vecalf =   D_\vecx^\vecalf = (\tfrac{1}{2\pi\i} \partial_{x_1})^{\alpha_1} \cdots (\tfrac{1}{2\pi\i} \partial_{x_n})^{\alpha_n}$$
and the norm
\begin{equation}
\| f \|_{M,N,p} = \sup_{\substack{|\vecalf|\leq M\\ |\vecbeta| \leq N}} \| \vecx^\vecbeta (D^\vecalf f)(\vecx)  \|_{L^p}.
\end{equation}

\begin{prop} \label{prop51}
There exists a constant $C_d$ depending only on the dimension $d$ such that for all $\vecy_0,\vecy_n \in \RR^d$, $W\in\scrS(\RR^d)$, $\vectheta\in\RR^{n-1}$, $\Re\gamma\geq 0$,
\begin{multline}
\bigg| \int_{(\RR^d)^{n-1}} \hat W(\vecy_0-\vecy_1) \cdots \hat W(\vecy_{n-1}-\vecy_n)  \\ 
\times e\bigg[\sum_{j=1}^{n-1} \theta_j \bigg( \tfrac12\|\vecy_0\|^2 - \tfrac12 \|\vecy_j\|^2+\i\gamma \bigg) \bigg]
  \,\d\vecy_1\cdots \d\vecy_{n-1} \bigg| \\
 \leq  C_d^n\, \| W \|_{2d+2,d+1,1}^n \,\e^{-2\pi (\theta_1+\cdots+\theta_{n-1}) \gamma}\, \langle \theta_1 \rangle^{-d/2} \cdots \langle\theta_{n-1}\rangle^{-d/2}.
\end{multline}
\end{prop}
\begin{proof}
The exponentially decaying factors can be pulled outside immediately. We then want to apply Lemma \ref{lem51}. Let $S = \{s_1,\dots,s_k\} \subset \{1,\dots,n-1\}$ and consider the norm $
\left\| f_S\right\|_{L^1} $
where $f(\vecy_1,\dots,\vecy_{n-1}):= \hat W(\vecy_0-\vecy_1) \cdots  \hat W(\vecy_{n-1}-\vecy_n)$ with $\vecy_0$ and $\vecy_n$ constant. By definition we have that
\begin{multline} \label{422}
\|f_S\|_{L^1} = \int_{\RR^{d(n-1)}} \bigg| \int_{\RR^{d|S|}} \hat W(\vecy_0-\vecy_{1}) \dots \hat W(\vecy_{n-1}-\vecy_{n}) \\
 \hfill \times [\prod_{i\in S} \e(\vecy_{i}\cdot\vecx_{i}) \d \vecy_{i}]  \bigg| [\prod_{i\notin S} \d \vecy_i] [\prod_{i\in S} \d \vecx_{i}].
 \end{multline}
Let $\vecm = (m_1,\dots,m_d) \in \ZZ_{\geq 0}^{d}$ and put $|\vecm| = m_1 +\cdots+m_d$. We define the multinomial coefficient
$$ \begin{pmatrix} N \\ \vecm \end{pmatrix} = \frac{N!}{m_1!\cdots m_d! (N-m_1-\cdots-m_d)!}.$$
and use that
\begin{equation} \langle \vecx \rangle^N \leq (1+|x_1|+\cdots+|x_d|)^N  = \sum_{| \vecm | \leq N} \begin{pmatrix} N \\ \vecm \end{pmatrix} \prod_{j=1}^d |x_j|^{m_j}
\end{equation}
to bound \eqref{422} above by
\begin{equation} \label{4230}
\begin{split}
\|f_S\|_{L^1} & =  \sum_{|\vecm_{s_1}|\leq d+1} \begin{pmatrix} d+1 \\ \vecm_{s_1} \end{pmatrix} \cdots \sum_{|\vecm_{s_k}|\leq d+1} \begin{pmatrix} d+1 \\ \vecm_{s_k} \end{pmatrix} \\
 & \times \int_{\RR^{d(n-1)}} \bigg| [\prod_{i \in S} \vecx_{i}^{\vecm_{i}}] \int_{\RR^{d|S|}} \hat W(\vecy_0-\vecy_{1}) \dots \hat W(\vecy_{n-1}-\vecy_{n}) \\
 & \times [\prod_{i\in S} \e(\vecy_{i}\cdot\vecx_{i}) \d \vecy_{i}]  \bigg| [\prod_{i\notin S} \d \vecy_i]  [\prod_{i \in S} \langle \vecx_i \rangle^{-d-1} \d \vecx_{i}].
\end{split}
\end{equation}
Integrating by parts with respect to $\vecy_{i}$ for $i \in S$ and pulling the absolute value inside the integral yields the upper bound
\begin{multline} \label{423}
\|f_S\|_{L^1} \leq  \bigg( \int_{\RR^d} \langle \vecx \rangle^{-d-1} \d \vecx \bigg)^k  \sum_{|\vecm_{s_1}| \leq d+1} \begin{pmatrix} d+1 \\ \vecm_{s_1} \end{pmatrix} \cdots \sum_{|\vecm_{s_k}| \leq d+1} \begin{pmatrix} d+1 \\ \vecm_{s_k} \end{pmatrix}\\
\times \int_{\RR^{d(n-1)}}  \bigg| [\prod_{i\in S} D_{\vecy_i}^{\vecm_i}] \hat W(\vecy_0-\vecy_{1}) \dots \hat W(\vecy_{n-1}-\vecy_{n}) \bigg| \d \vecy_1 \cdots \d \vecy_{n-1}.
\end{multline}
Using the triangle inequality, the $\vecy_i$ integral can be bounded above by a sum of $\prod_{i\in S} 2^{|\vecm_i|}$ terms of the form
\begin{multline} \label{four26}
\int_{\RR^{d(n-1)}} \big|  \varphi_0(\vecy_0-\vecy_{1}) \cdots \varphi_{n-1}(\vecy_{n-1}-\vecy_{n})  \big|
\d \vecy_1 \cdots \d \vecy_{n-1} \\
= \int_{\RR^{d(n-1)}} \big|  \varphi_0(\vecy_0-\vecy_{1}-\cdots-\vecy_n) \varphi_1(\vecy_1) \cdots \varphi_{n-1}(\vecy_{n-1})  \big|
\d \vecy_1 \cdots \d \vecy_{n-1}
\end{multline}
where each $\varphi_i$ is a derivative of $\hat W$ of order $\leq 2d+2$. Pulling absolute values inside tells us that \eqref{four26} is bounded above by
\begin{equation}\label{42six}
\| \varphi_0 \|_{\L^\infty} \| \varphi_1 \|_{\L^1} \cdots \| \varphi_{n-1} \|_{\L^1} .  
\end{equation}
We have that 
\begin{equation}
\begin{split}\label{infnorms}
\|\varphi_0\|_{L^\infty} &\leq \sup_{|\vecalf| < 2d+2} \sup_{\vecy \in \RR^d}  \left| D_{\vecy}^\vecalf \int_{\RR^d}  W(\vecx) \, \e(-\vecx \cdot\vecy) \d \vecx \right| \\ 
&=  \sup_{|\vecalf| < 2d+2} \sup_{\vecy \in \RR^d} \left|\int_{\RR^d} \vecx^\vecalf W(\vecx) \, \e(-\vecx \cdot\vecy) \d \vecx \right|\leq \| W \|_{2d+2,0,1}.
\end{split} 
\end{equation}
For the $L^1$ norms we similarly have
\begin{equation}
\begin{split}
\|\varphi_i\|_{L^1} &\leq \sup_{|\vecalf|<2d+2} \int_{\RR^d} \left| D_{\vecy}^\vecalf \int_{\RR^d} W (\vecx) \e(-\vecx\cdot\vecy) \d \vecx \right| \, \d \vecy \\
&= \sup_{|\vecalf|<2d+2} \int_{\RR^d} \left|  \int_{\RR^d}\vecx^{\vecalf} W (\vecx) \e(-\vecx\cdot\vecy) \d \vecx \right| \, \d \vecy \\
&\leq \sup_{|\vecalf|<2d+2} \left(\int_{\RR^d} \langle \vecy\rangle^{-d-1} \d \vecy\right) \sum_{\vecm \leq d+1} \begin{pmatrix} d+1 \\ \vecm \end{pmatrix} \int_{\RR^d}  \left|  D_\vecx^{\vecm} \vecx^{\vecalf} W (\vecx)\right| \d \vecx.
\end{split} 
\end{equation}
Note that $D_\vecx^\vecm \vecx^\vecalf W(\vecx)$ yields a sum of up to $2^{|\vecm|}$ terms each of which is of the form $\vecx^{\tilde \vecalpha} D_{\vecx}^{\tilde \vecm} W(\vecx)$ where $|\tilde \vecalf|<|\vecalf|$ and $|\tilde \vecm| < |\vecm|$. Hence, using the identity 
\begin{equation} \label{multinomialidentity}
\sum_{|\vecm|<d+1} \begin{pmatrix} d+1 \\ \vecm \end{pmatrix} 2^{|\vecm|} = (2d+1)^{d+1}
\end{equation}
we obtain
\begin{equation} \label{onenorms}
\|\varphi_i\|_{L^1} \leq  \left( \int_{\RR^d} \langle \vecy\rangle^{-d-1} \d \vecy \right) (2d+1)^{d+1} \, \| W \|_{2d+2,d+1,1}.
\end{equation}
Finally then, combining these bounds with \eqref{423} and applying \eqref{multinomialidentity} we obtain
\begin{equation}
\|f_S\|_{L^1} \leq \left(\int_{\RR^d} \langle \vecx \rangle^{-d-1} \d \vecx\right)^{k+n-1} (2d+1)^{(k+n-1)(d+1)} \, \| W \|_{2d+2,d+1,1}^n.
\end{equation}
\end{proof}

This proposition shows that 
\begin{equation}
|T_n^\gamma(\vecy_0,\vecy_n) | \leq (2 \pi)^{n-1} C_d^n\, \|W \|_{2d+2,d+1,1}^n \bigg(\int_0^\infty \langle \theta \rangle^{-d/2} \,\d\theta \bigg)^{n-1}
\end{equation}
for all $\Re\gamma\geq 0$, provided $d >2 $. Therefore, the series \eqref{Tseries} converges absolutely, uniformly for all $\Re\gamma\geq 0$, if
\begin{equation}
|\lambda| < \left(2\pi C_d\, \|W \|_{2d+2,d+1,1} \int_0^\infty \langle \theta \rangle^{-d/2} \,\d\theta \right)^{-1}.
\end{equation}

\section{The perturbation series}\label{secDuhamel}

We set $H_\lambda=H_{1,\lambda}$, $U_\lambda(t)=U_{1,\lambda}(t)$, and note that
\begin{equation} \label{lambdaisrescaled}
H_{h,\lambda}= h^2 H_{\lambda/h^2}, \qquad
U_{h,\lambda}(t)= U_{\lambda/h^2}(h t) .
\end{equation}
The first term in \eqref{Ath} thus can be written
$$
U_{\lambda/h^2}(t h r^{1-d}) \, \Op_{r,h}(a) \, U_{\lambda/h^2}(-t h r^{1-d}) ,
$$ 
which is the quantity of interest for this work.
To simplify notation, we will write in the upcoming discussion $\lambda$ for $\lambda/h^2$ and $t$ for $t h r^{1-d}$ and later re-substitute when taking limits.
Furthermore, we can declutter our expressions by passing to the so-called interaction picture,
$$ U_{\lambda}(t) U_0(-t) \, \Op_{r,h}(a) \, U_0(t) U_{\lambda}(-t). $$
After the relevant calculations we then simply replace $a$ by $L_0(t) a$ due to \eqref{egorovh}. Because of the gauge invariance of $A(t)$ in \eqref{At} under the substitution $H_{h,\lambda}\mapsto H_{h,\lambda}+E$ for any $E\in\RR$, we may replace the potential $V$ by $V-\int_\scrL V(\vecx)\d\vecx$ in the following. 
This means that the potential now has the Fourier series
\begin{equation}\label{nonconsec}
V(\vecx)= r^d \sum_{\vecb\in\scrL^*\setminus\{\vecnull\}} \hat W(r \vecb) \,\e(\vecb \cdot \vecx) , 
\end{equation}
with $\hat W(\vecy) = \int W(\vecx) e(-\vecx\cdot\vecy)\,\d\vecy$.
Thus we may ignore in the following expansions all terms with $\hat W(\vecnull)$; but note that we have {\em not} assumed here that $\hat W(\vecnull)=0$.

We proceed using Duhamel's principle. In particular one has that
\begin{equation}
U_{\lambda}(t)  = U_0(t) -  2 \pi \i \lambda  \int_0^t U_\lambda(s)\, \Op(V)\,  U_0(t-s) \,   \d s.
\end{equation}
Iterating this expression yields a formal perturbative expansion for $U_\lambda(t) U_0(-t)$ and $U_0(t) U_{\lambda}(-t)$. After multiplying these two series together one obtains as in \cite[\S 5]{GM} the formal perturbative expansion
\begin{multline} \label{scrIdef}
\Tr[ \Pi_\vecalf U_\lambda(t ) U_0(-t) \Op_{r,h}( a) U_0(t ) U_\lambda(-t ) \Op_{r,h}(b) ] \\
= 
\sum_{n=0} ^{\infty} (2 \pi \i \lambda)^n  \scrI_n(t) + O(r^\infty),
\end{multline}
with $\scrI_0(t)= \Tr[ \Pi_\vecalf \Op_{r,h}( a) \Op_{r,h}(b) ]$ and for $n\geq 1$
\begin{equation}
 \scrI_n(t) = \sum_{\ell=0}^{n} (-1)^\ell \scrI_{\ell,n}(t)  .
\end{equation}
For $1\leq \ell \leq n-1$ we have, with the shorthand $\scrP=\scrL^*+\vecalf$,
\begin{equation}
\begin{split}
& \scrI_{\ell,n}(t) =r^{nd} h^d \sum_{\substack{\vecp_1,\dots,\vecp_n=\vecp_0\in\scrP\\ \text{non-consec}}} 
\scrW(r\vecp_0,\ldots,r\vecp_n)
\int_{\RR^d}   \int_{\substack{0<s_1<\cdots<s_\ell<t \\ 0<s_n<\cdots<s_{\ell+1}<t} } 
\\
& \times  \e\bigg(-\tfrac12 \, s_1\|\vecp_0\|^2+\tfrac12 \sum _{j=1}^{\ell-1} (s_j-s_{j+1}) \|\vecp_j \|^2  + \tfrac12 \, s_\ell \|\vecp_\ell \|^2
-\tfrac12 \, s_{\ell+1}\|\vecp_\ell + r^{d-1} \veceta\|^2 \bigg) \\
& \times \e\bigg(\tfrac12 \sum_{j=\ell+1}^{n-1} (s_j-s_{j+1}) \|\vecp_j +r^{d-1} \veceta \|^2 + \tfrac12 \, s_n \|\vecp_n+ r^{d-1} \veceta \|^2\bigg)  \\
& \times \tilde a (-\veceta, h (\vecp_\ell+\tfrac12 r^{d-1}\veceta)) \,   \tilde b( \veceta, h (\vecp_n+ \tfrac12 r^{d-1} \veceta)) \, \d \vecs  \, \d \veceta,
\end{split}
\end{equation}
where 
\begin{equation}\label{scrWdef}
\scrW(\vecy_0,\ldots,\vecy_n)= \prod_{j=0}^{n-1}  \hat W(\vecy_j-\vecy_{j+1})  ,
\end{equation}
and the summation ``non-consec'' is restricted to terms with $\vecp_j\neq\vecp_{j+1}$; recall the comment after \eqref{nonconsec}. We make the variable substitutions $u_0= s_1$, $u_j = s_{j+1}-s_j$ for $j = 1,\ldots, \ell-1$, and $u_j = s_j - s_{j+1}$ for $j = \ell+1,\ldots, n-1$ and $u_n = s_n$. Let $\Box_{\ell,n}(t)$ denote the simplex
\begin{multline}
\Box_{\ell,n}(t)=\{ \vecu=(u_0,\ldots,u_n) \in\RR_{\geq 0}^{n+1} \mid  \\
u_0+\cdots+u_{\ell-1} < t,\; u_\ell=0,\;
u_{\ell+1} + \cdots + u_n < t\} ,
\end{multline}
and let $\d^\perp\vecu=\prod_{j\neq \ell} \d u_j$. Then
\begin{equation}\label{LOT}
\begin{split}
\scrI_{\ell,n}(t) &=r^{nd} h^d  \sum_{\substack{\vecp_1,\dots,\vecp_n=\vecp_0\in\scrP\\ \text{non-consec}}} 
\scrW(r\vecp_0,\ldots,r\vecp_n)\int_{\RR^d}   \int_{\Box_{\ell,n}(t)}
   \\
&\times  \e\bigg(\tfrac12 \, \sum_{j=0}^{\ell-1}  u_j (\|\vecp_\ell\|^2-\|\vecp_j \|^2) +
\tfrac12 \sum_{j=\ell+1}^{n} u_j (\|\vecp_j + r^{d-1} \veceta\|^2-\|\vecp_\ell+r^{d-1}\veceta\|^2)  \bigg) \\
& \times \tilde a (-\veceta, h (\vecp_\ell+\tfrac12 r^{d-1}\veceta)) \,  \tilde b( \veceta, h (\vecp_n+ \tfrac12 r^{d-1} \veceta)) \, \d^\perp\vecu \, \d \veceta .
\end{split}
\end{equation}
The terms $\scrI_{0,n}(t)$ and $\scrI_{n,n}(t)$ have an analogous representation. 

As in the case of the $T$-matrix, it will be useful to embed these quantities in an analytic family by extending to complex energy. To this end, we define for $\Re\gamma\geq 0$,
\begin{equation}\label{LOT2}
\begin{split}
\scrI_{\ell,n}^\gamma(t) &=r^{nd} h^d  \sum_{\substack{\vecp_1,\dots,\vecp_n=\vecp_0\in\scrP\\ \text{non-consec}}} 
\scrW(r\vecp_0,\ldots,r\vecp_n)\int_{\RR^d}   \int_{\Box_{\ell,n}(t)}
   \\
&\times  \e\bigg(\tfrac12 \, \sum_{j=0}^{\ell-1}  u_j (\|\vecp_\ell\|^2-\|\vecp_j \|^2+\i\gamma) \\ & \qquad +
\tfrac12 \sum_{j=\ell+1}^{n} u_j (\|\vecp_j + r^{d-1} \veceta\|^2-\|\vecp_\ell+r^{d-1}\veceta\|^2+\i\gamma)  \bigg) \\
& \times \tilde a (-\veceta, h (\vecp_\ell+\tfrac12 r^{d-1}\veceta)) \,  \tilde b( \veceta, h (\vecp_n+ \tfrac12 r^{d-1} \veceta)) \, \d^\perp\vecu \, \d \veceta .
\end{split}
\end{equation}
In the following, we will drop the superscript $\gamma$ in the case $\gamma=0$.

We wish to consider the limit of this quantity as $h=r \to 0$, uniformly for $\Re\gamma\geq 0$. The first simplification we make is to replace the second argument of $\tilde a$ and $\tilde b$ by $h \vecp_\ell$ and $h \vecp_n$ respectively which incurs an error of order $r^d$. Recall now that, in view of \eqref{Ath2} and \eqref{lambdaisrescaled}, we are interested in the quantity $h^{-2n}\scrI_{\ell,n}^\gamma(t h r^{1-d})$. Since $h^{-2n} r^{nd} h^d (h r^{1-d})^n = h^{d-n} r^n = r^d$ we see for $h=r$
\begin{multline}\label{multieq}
h^{-2n}\scrI_{\ell,n}^\gamma(t h r^{1-d})  \\
= r^d \sum_{\substack{\vecp_1,\dots,\vecp_n=\vecp_0\in\scrP\\ \text{non-consec}}} H_{t,\ell,n}^{r^{2-d}\gamma}\big( r^{2-d} (\tfrac12\|\vecp_0\|^2,\ldots,\tfrac12\|\vecp_n\|^2),  r \vecp_0,\ldots,r\vecp_n \big) (1+O(r^d)) ,
\end{multline}
where
\begin{equation}\label{Hdef}
\begin{split}
& H_{t,\ell,n}^\gamma(\vecxi,\vecy_0,\ldots,\vecy_n) =\scrW(\vecy_0,\ldots,\vecy_n) \int_{\RR^{d}} \int_{\Box_{\ell,n}(t)}  \\
&\times \e\bigg(- \sum_{j=0}^{\ell-1}  u_j (\xi_j-\xi_\ell-\i\gamma) +  \sum_{j=\ell+1}^{n}  u_j (\xi_j-\xi_\ell+\i\gamma+(\vecy_j-\vecy_\ell)\cdot\veceta)\bigg) \\
& \times \tilde a (-\veceta, \vecy_\ell)\, \tilde b( \veceta, \vecy_n) \, \d^\perp\vecu \, \d \veceta.
\end{split}
\end{equation}
Using the definition of $\tilde a$ and $\tilde b$ and integrating over $\veceta$ yields
\begin{equation} \label{Hdef2}
\begin{split}
H_{t,\ell,n}^\gamma(\vecxi,\vecy_0,\ldots,\vecy_n)&= \scrW(\vecy_0,\ldots,\vecy_n) \int_{\RR^{d}} \int_{\Box_{\ell,n}(t)}
  \\
& \times  \e\bigg(- \sum_{j=0}^{\ell-1}  u_j (\xi_j-\xi_\ell-\i\gamma) +  \sum_{j=\ell+1}^{n}  u_j (\xi_j-\xi_\ell+\i\gamma) \bigg)  \\
& \times \scrA_{\ell,n}(\vecy_0,\ldots,\vecy_n,\vecu)\, \d^\perp\vecu ,
\end{split}
\end{equation}
where
\begin{equation}\label{scrAdef}
\scrA_{\ell,n}(\vecy_0,\ldots,\vecy_n,\vecu)=   \int_{\RR^d}   a\bigg( \vecx - \sum_{j=\ell+1}^n  u_j (\vecy_j-\vecy_\ell),\vecy_\ell \bigg) b( \vecx, \vecy_n) \, \d \vecx.
\end{equation}

We recall here that we are working in the interaction picture. To return to the original lab frame, we replace $a$ by the evolved symbol (i.e., replacing $\vecx$ by $\vecx-t\vecy_\ell$), so that 
$$
a\bigg( \vecx - \sum_{j=\ell+1}^n  u_j (\vecy_j-\vecy_\ell),\vecy_\ell \bigg)
$$
becomes
$$
a\bigg( \vecx - \bigg(t-\sum_{j=\ell+1}^n  u_j\bigg)\vecy_\ell - \sum_{j=\ell+1}^n  u_j \vecy_j,\vecy_\ell \bigg) .
$$
One can interpret this as corresponding to a classical trajectory in which the particle initially has momentum $\vecy_\ell$, undergoes straight line motion for time $(t-u_{\ell+1}-\cdots-u_n)$, and then experiences $n-\ell$ collisions separated by straight line motion for times $u_j$ with momenta $\vecy_{j}$ for $j= \ell+1,\dots,n$.

\section{The Poisson model}\label{secPoisson}

We note that the momenta $\vecp_j$ in the summation \eqref{multieq} are of order $r^{-1}$, and that for $r\to 0$
\begin{equation}
|\{ \vecp\in\scrP : \|\vecp \|^2 \leq r^{-2} \}| \sim \vol(\scrB_1^d) \, r^{-d}
\end{equation}
where $\scrB_1^d$ is the $d$-dimensional unit ball. This means that the average spacing between consecutive values of the set $\{ \|\vecp\|^2 \leq r^{-2}:\vecp\in\scrP \}$ is of the order $r^{d-2}$. Thus \eqref{multieq} measures correlations between the $\|\vecp_j\|^2$ precisely on the scale of the mean spacing. Starting with the influential work of Berry and Tabor in the context of quantum chaos \cite{BerryTabor}, it has been conjectured that the statistics on this correlation scale should be governed by a one-dimensional Poisson process. Rigorous results towards a proof of this conjecture are mostly limited to two-point statistics, where the problem reduces to a variant of quantitative versions of the Oppenheim conjecture \cite{Bleher95,EMM,Marklof02,Marklof03,MargulisMohammadi,Sarnak96}; results on higher correlation functions are obtained in \cite{VanderKam}.

\begin{assump}\label{hyp0}
We assume in the following that in the $r=h\to 0$ asymptotics of \eqref{multieq} the lattice $\scrP = \scrL^* + \vecalf$ with $\scrL$ fixed (arbitrary) and $\vecalf$ random can be replaced by a Poisson process in $\RR^d$ with unit intensity.
\end{assump}

This assumption should be thought of as a generalisation of the Berry-Tabor conjecture on the Poisson distribution of energy levels of quantum systems with integrable classical Hamiltonian. We assume both that the lengths of lattice vectors (which represent the energy levels) behave as if they belonged to a Poisson process, and that the angular distribution of the lattice vectors is uniform on the $(d-1)$-sphere and independent of the length (on the correct scale). 

Assumptions of this kind have previously been used in modeling spectral correlations of diffractive systems, see for instance \cite{Bogomolny00,Bogomolny02,Letendre19}. To formulate Assumption \ref{hyp0} in precise terms, define $\scrJ_{\ell,n}^\gamma(t)$ via the relation 
\begin{multline}\label{multieq007}
h^{-2n}\scrJ_{\ell,n}^\gamma(t h r^{1-d})  \\
= r^d\; \EE \sum_{\substack{\vecp_1,\dots,\vecp_n=\vecp_0\in\scrP\\ \text{non-consec}}} H_{t,\ell,n}^{r^{2-d}\gamma}\big( r^{2-d} (\tfrac12\|\vecp_0\|^2,\ldots,\tfrac12\|\vecp_n\|^2),  r \vecp_0,\ldots,r\vecp_n \big) ,
\end{multline}
where $\scrP$ is a Poisson point process in $\RR^d$ with intensity one, and $\EE$ denotes expectation.
Then Assumption \ref{hyp0} should be understood as
\begin{equation}\label{hyp007}
\lim_{r=h\to 0} h^{-2n} \int_{\RR^d/\scrL^*}  \sum_{n=1}^\infty \sum_{\ell=0}^n (-1)^{\ell} \big[ \scrI_{\ell,n}^\gamma(t h r^{1-d}) - \scrJ_{\ell,n}^\gamma(t h r^{1-d}) \big] \d\vecalf = 0
\end{equation}
for all $\Re\gamma\geq 0$, $t>0$. Similarly, it is a conjecture that 
\begin{equation} \label{hyp008}
\lim_{r=h\to 0} h^{-2n} \sum_{n=1}^\infty \sum_{\ell=0}^n (-1)^{\ell} \big[ \scrI_{\ell,n}^\gamma(t h r^{1-d}) - \scrJ_{\ell,n}^\gamma(t h r^{1-d}) \big]  = 0
\end{equation}
for Lebesgue almost every $\vecalf$, and indeed for $\vecalf$ satisfying a mild diophantine condition as in \cite{Marklof02,Marklof03}. Statement \eqref{hyp008} is more subtle than \eqref{hyp007}, though the implication $\eqref{hyp008}\Rightarrow\eqref{hyp007}$ would require uniform upper bounds for dominated convergence; cf.~\cite[Sect.~12]{GM}. Statement \eqref{hyp007} is the only heuristic assumption made in this study.

As an illustrative example, fix $\vecalf =(\sqrt{2},\sqrt{3})$ and consider the sequence $(\lambda_i, \theta_i)_{i \in \NN}$ of elements of the set $\{( \pi \|\vecn+\vecalf\|^2, \frac{1}{2\pi } \arg (\vecn+\vecalf) ) \in \RR_{\geq 0} \times [0,1) \mid \vecn \in \ZZ^2 \} $ arranged in increasing order according to the first component where $0\leq \arg \vecz < 2 \pi$ is the polar angle of $\vecz$. Our assumption is concerned with the distribution of points $(\lambda_i,\theta_i)$ restricted to a strip $[R- \Delta R,R) \times [0,1)$ for $\Delta R > 0$ fixed and $R \to \infty$. Due to the choice of normalisation, a strip of this form contains roughly $\Delta R$ points. Broadly speaking, the points contained in the strip should behave more and more randomly as $R$ increases -- see Figure \ref{fig1}.
\begin{center}
\begin{figure}[h]
\includegraphics[width=0.8\textwidth]{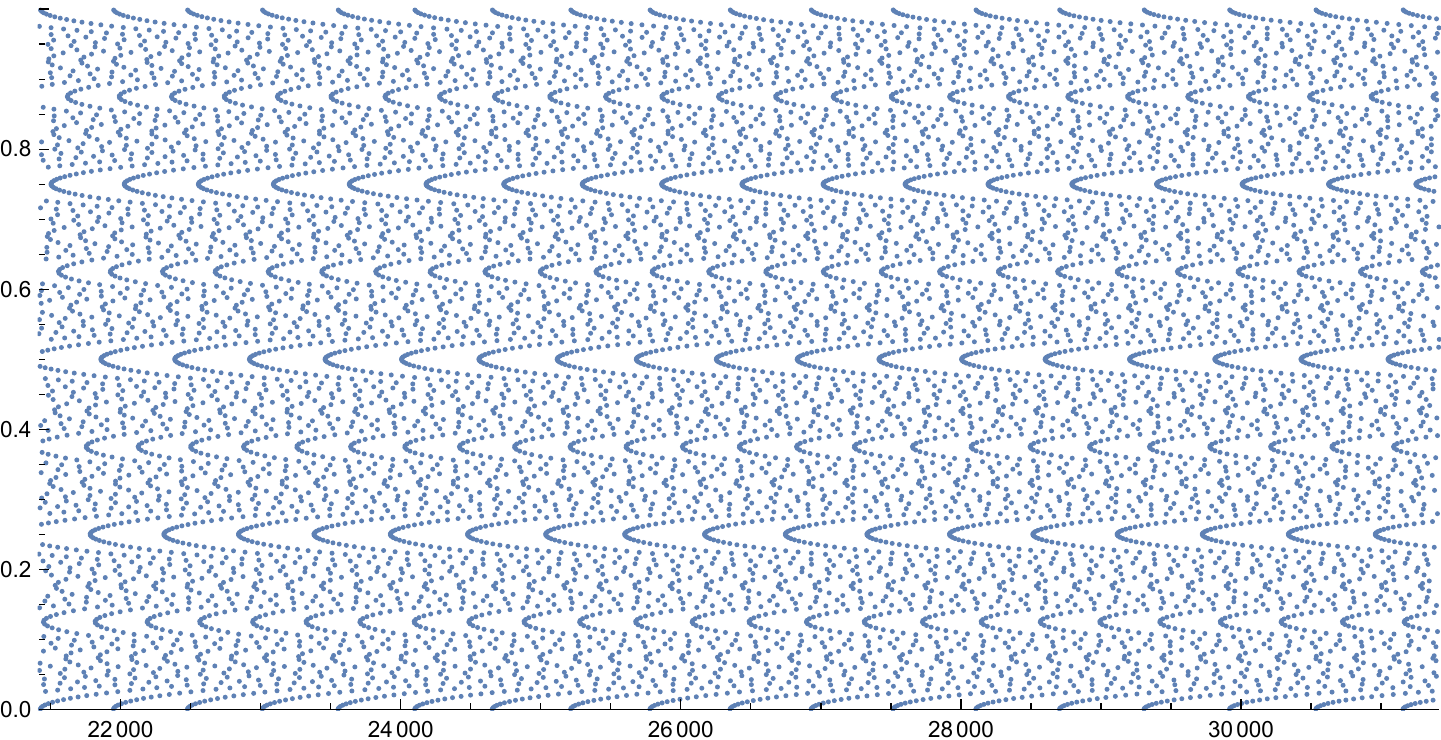}
\includegraphics[width=0.8\textwidth]{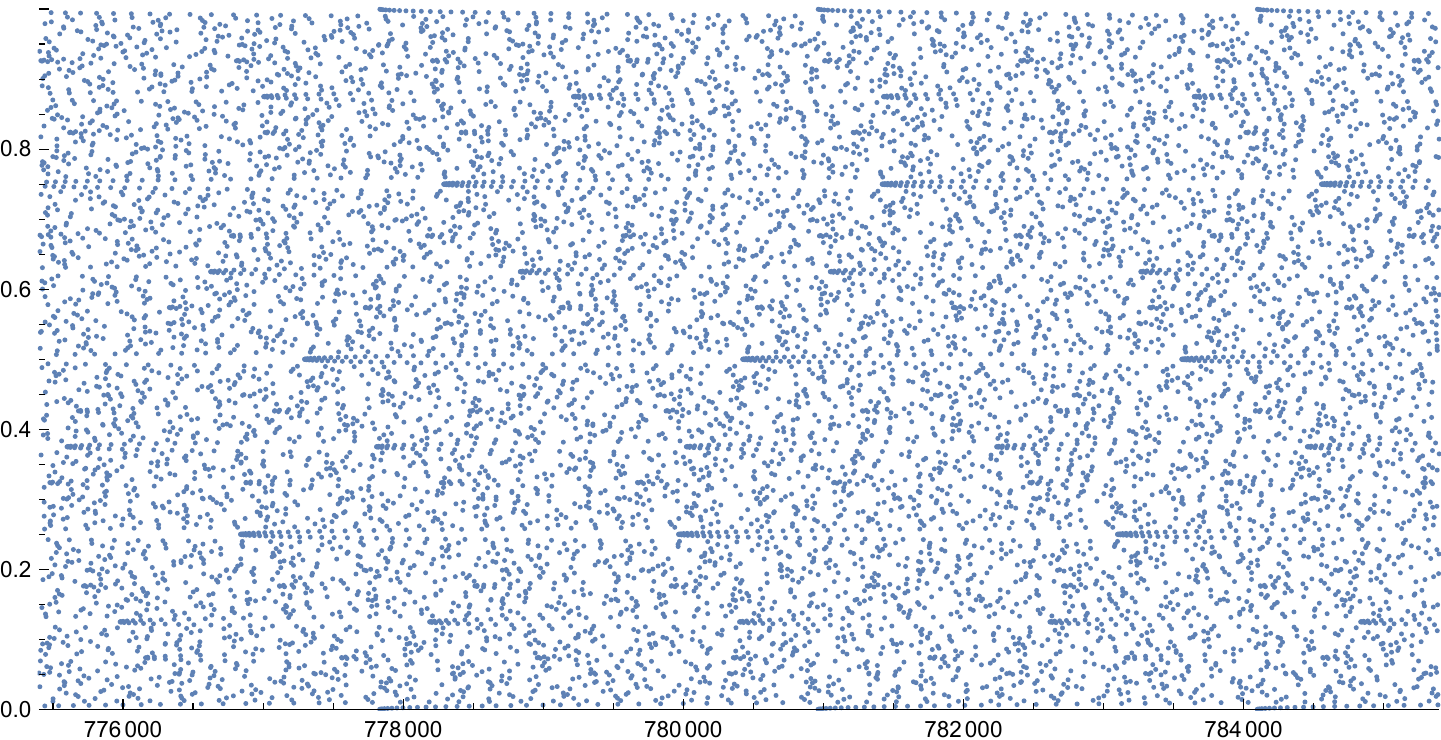}
\caption{Scatter plots of $(\lambda_i,\theta_i)$ in the strip $[R-\Delta R,R) \times [0,1)$ for $R = \pi \times 100^2$ and $R = \pi \times 500^2$,  respectively, with $\Delta R = 10^4$. For large R we expect the point set to be modelled by a Poisson point process, cf. Assumption \ref{hyp0}.}\label{fig1}
\end{figure}
\end{center}

The Berry-Tabor conjecture, and by extension our assumption, is more readily expressed in terms of the gap distribution. Consider the sequence $(\lambda_{i+1}-\lambda_i,\theta_i)$ for all points in the window $[R-\Delta R,R) \times [0,1)$. The Berry-Tabor conjecture states that in the limit $R \to \infty$, the sequence of gaps $\lambda_{i+1}-\lambda_i$ has an exponential distribution with mean $1$. Our Assumption \ref{hyp0} then implies that in the limit $R\to \infty$, the sequence of pairs $(\lambda_{i+1}-\lambda_i,\theta_i)$ is distributed according to the product of an exponential distribution with mean $1$ and the uniform distribution on $[0,1)$ -- see Figure \ref{fig2}.

\begin{center}
\begin{figure}[h]
\includegraphics[width=0.8\textwidth]{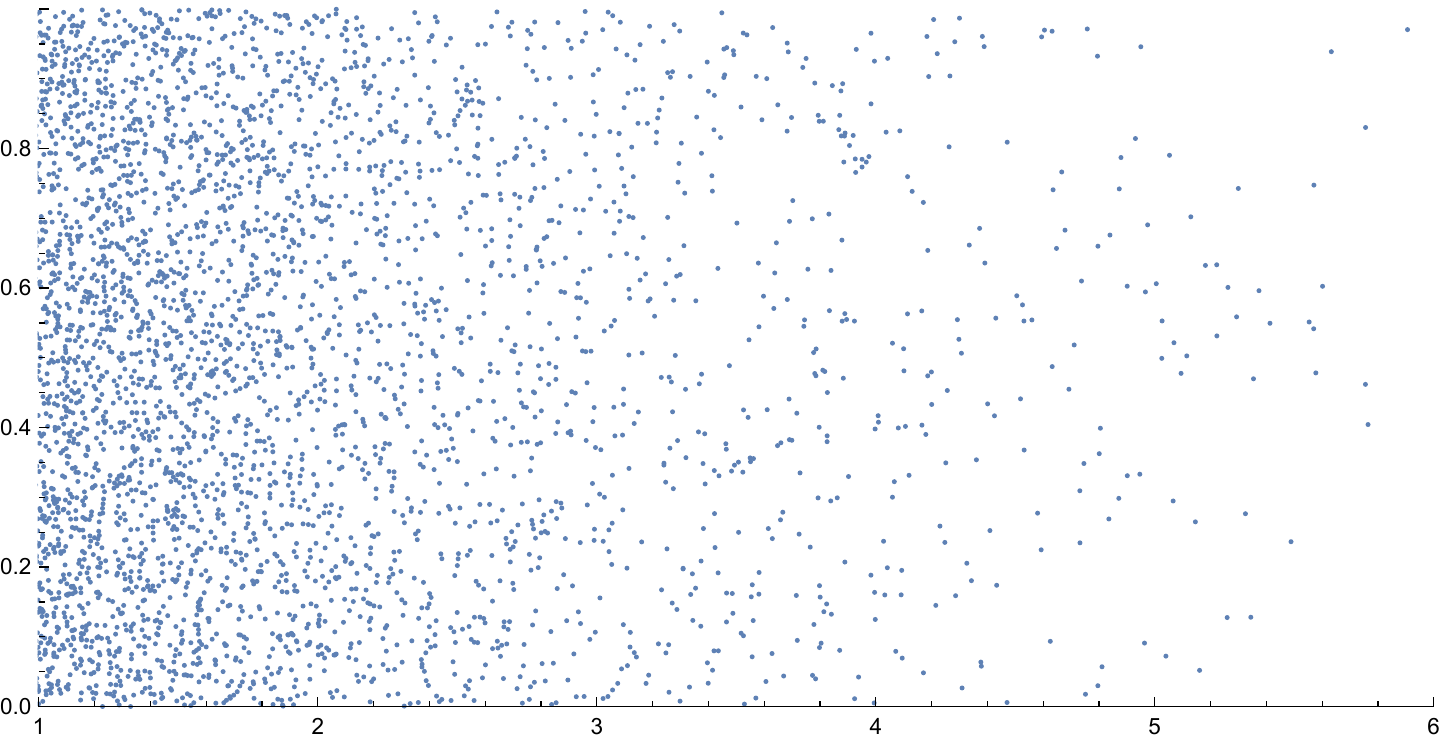}
\includegraphics[width=0.8\textwidth]{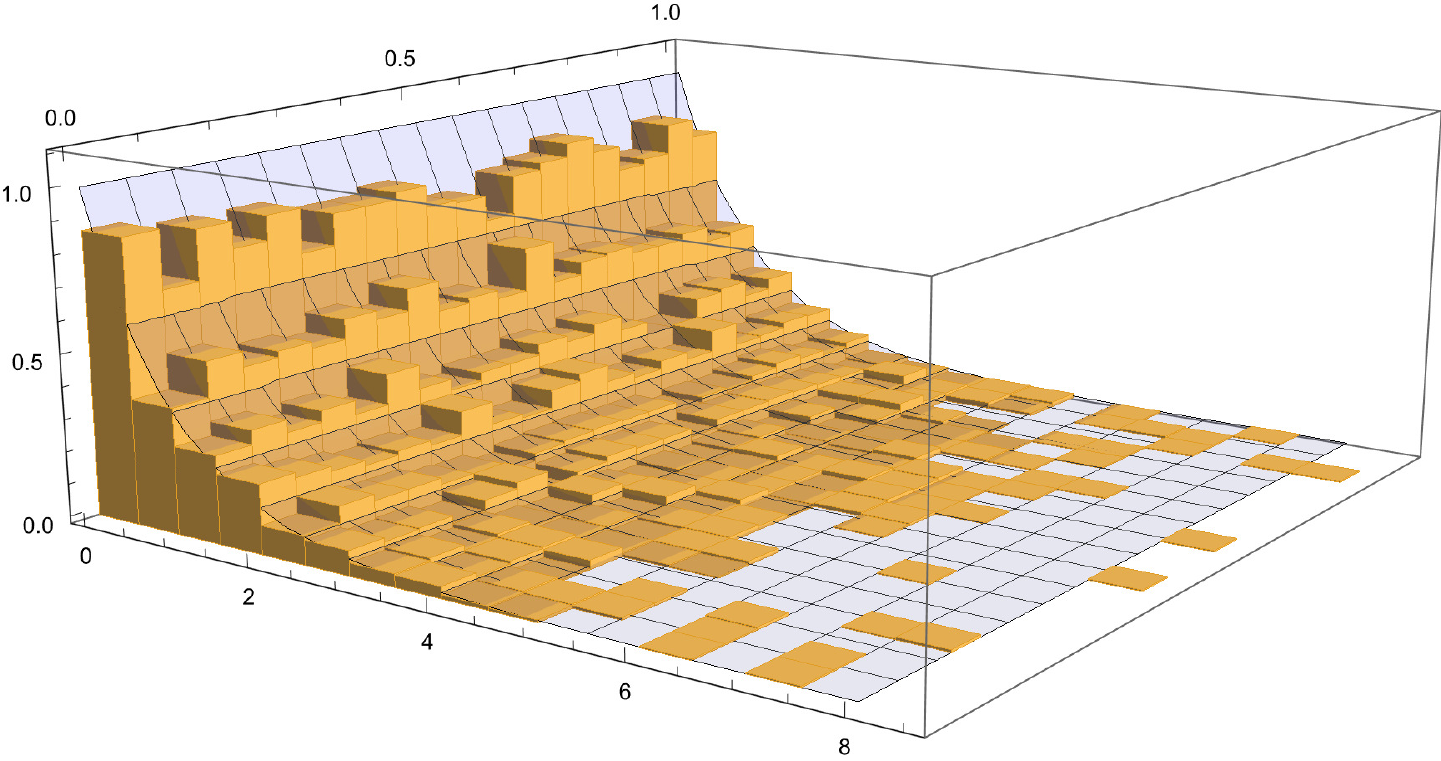}
\caption{A scatter plot and histogram for the sequence $(\lambda_{i+1}-\lambda_i,\theta_i)$ for $R=\pi \times 500^2$ and $\Delta R=10^4$. The surface superimposed on the histogram is the density of the conjectured limiting distribution under Assumption \ref{hyp0}.}  \label{fig2}
\end{figure}
\end{center}

Let $\scrJ_n^\gamma = \sum_{\ell=0}^n (-1)^\ell\scrJ_{\ell,n}^\gamma$. The principal objective of this paper is to now prove that the limit 
\begin{equation}
\lim_{h=r\to 0}   \sum_{n=1}^\infty  (2 \pi \i \lambda h^{-2})^n \; \scrJ_n^\gamma(t h r^{1-d})
\end{equation}
exists for $\lambda>0$ sufficiently small but fixed, every $t>0$ and $\Re\gamma>0$, and to evaluate the limit at $\gamma=0$.
The convergence is stated in Proposition \ref{prop:tmatexpansion} and explicit formulas for the limit are discussed in Section \ref{secpositivity}.

Now, in order to take the expectation value of the sum over the momenta in in \eqref{multieq}, with $\scrP$ a Poisson point process, one needs to keep track of terms where the various $\vecp_i$ and $\vecp_j$ are equal or distinct. This is best done through the notion of set partitions, which are presented in detail in Appendix \ref{sec:partitions}. We denote by $\Pi(n,k)$ the collection of set partitions $\F=[F_1,\ldots,F_k]$ of the set $\{0,\ldots,n\}$ into $k$ blocks $F_1,\ldots,F_k$. Let us then define, for a given set partition $\F\in\Pi(n,k)$, $\vecxi\in\RR^{k}$, $\vecy_1,\ldots,\vecy_k\in\RR^d$ 
\begin{equation}
H_{t,\ell,\F}^\gamma(\vecxi,\vecy_1,\ldots,\vecy_k) = H_{t,\ell,n}^\gamma(\iota_{\F}(\vecxi),\iota_{\F}(\vecy_1,\ldots,\vecy_k)) ,
\end{equation}
where $\iota_{\F}(\vecxi)$ is the embedding $\iota_{\F}: \RR^k\to \RR^{n+1}$ defined by
$\vecxi=(\xi_1,\ldots,\xi_k)\mapsto (x_0,\ldots,x_n)$ where $(x_j=\xi_i \iff j\in F_i)$, and $\iota_{\F}(\vecy_1,\ldots,\vecy_k)$ is the vector analogue $\iota_{\F}: (\RR^d)^k\to(\RR^d)^{n+1}$. See Section \ref{subsec:ordered} for details.

With this we can write
\begin{multline}\label{multieq20}
\sum_{\substack{\vecp_1,\dots,\vecp_n=\vecp_0\in\scrP\\ \text{non-consec}}} H_{t,\ell,n}^{r^{2-d}\gamma}\big( r^{2-d} (\tfrac12\|\vecp_0\|^2,\ldots,\tfrac12\|\vecp_n\|^2),  r \vecp_0,\ldots,r\vecp_n \big) \\
= \sum_{k=1}^{n+1} \sum_{\F\in\Pi_\circ(n,k)} \sum_{\substack{\vecp_1,\dots,\vecp_k\in\scrP\\ \text{distinct}}} H_{t,\ell,\F}^{r^{2-d}\gamma}\big( r^{2-d} (\tfrac12\|\vecp_1\|^2,\ldots,\tfrac12\|\vecp_k\|^2),  r \vecp_1,\ldots,r\vecp_k \big) .
\end{multline}

Note that for $\vece=(1,1,\ldots,1)\in\RR^k$ and any $\omega\in\RR$
\begin{equation}\label{traninv}
H_{t,\ell,\F}^\gamma(\vecxi+\omega \vece,\vecy_1,\ldots,\vecy_k)=H_{t,\ell,\F}^\gamma(\vecxi,\vecy_1,\ldots,\vecy_k).
\end{equation}
This decomposition allows us to compute the expectation over $\scrP$. Specifically, Campbell's theorem yields, for $\scrP$ a Poisson process with intensity one,
\begin{equation}\label{multieq3}
\begin{split}
\EE & \sum_{\substack{\vecp_1,\dots,\vecp_k\in\scrP\\ \text{distinct}}} H_{t,\ell,\F}^{r^{2-d}\gamma}\big( r^{2-d} (\tfrac12\|\vecp_1\|^2,\ldots,\tfrac12\|\vecp_k\|^2),  r \vecp_1,\ldots,r\vecp_k \big) \\
& =\int_{(\RR^d)^k} H_{t,\ell,\F}^{r^{2-d}\gamma}\big( r^{2-d} (\tfrac12\|\vecy_1\|^2,\ldots,\tfrac12\|\vecy_k\|^2),  r \vecy_1,\ldots,r\vecy_k \big)\,\d\vecy_1\cdots\d\vecy_k\\
& =r^{-kd} \int_{(\RR^d)^k} H_{t,\ell,\F}^{r^{-d}\gamma}\big( r^{-d} (\tfrac12\|\vecy_1\|^2,\ldots,\tfrac12\|\vecy_k\|^2),  \vecy_1,\ldots,\vecy_k \big)\,\d\vecy_1\cdots\d\vecy_k .
\end{split}
\end{equation}
Due to its translation invariance, $H_{t,\ell,\F}^\gamma$ is determined by its values in the $l^\text{th}$ coordinate plane $\Sigma^\perp=\{\vecx\in\RR^k: x_l=0\}$. It will be convenient in our calculations below to fix $l$ so that $\ell\in F_l$. We define the corresponding Fourier transform of $H_{t,\ell,\F}^\gamma$ restricted to $\Sigma^\perp$ by
\begin{equation}\label{Gperp}
G_{t,\ell,\F}^\gamma\big(\vectheta ,  \vecy_1,\ldots,\vecy_k) = 
\int_{\Sigma^\perp} H_{t,\ell,\F}^\gamma\big(\vecxi,  \vecy_1,\ldots,\vecy_k) \, e(-\vecxi\cdot\vectheta)\, \d^\perp\vecxi ,
\end{equation}
where $\vectheta\in\Sigma^\perp$ and $\d^\perp\vecxi= \prod_{\substack{j = 1\\ j\neq l}}^k \d \vecxi_j$ denotes the standard Lebesgue measure on $\Sigma^\perp$. For $\vecxi\in\Sigma^\perp$ we then have 
\begin{equation} \label{Gperpinv}
H_{t,\ell,\F}^\gamma\big(\vecxi,  \vecy_1,\ldots,\vecy_k\big) = 
\int_{\Sigma^\perp}  G_{t,\ell,\F}^\gamma\big( \vectheta ,  \vecy_1,\ldots,\vecy_k\big) \, e(\vecxi\cdot\vectheta)\,  \d^\perp\vectheta .
\end{equation}

Define the $(k\times (n+1))$-matrix $\Delta_{\ell,\F}$ by
\begin{equation}
(\Delta_{\ell,\F})_{ij}= 
\begin{cases}
0 & \text{if $i=l$} \\
-1  & \text{if $j\in F_i \cap [0,\ell-1]$} \\
1  & \text{if $j\in F_i \cap [\ell+1, n]$}\\
0 & \text{otherwise.}
\end{cases}
\end{equation}
In particular $(\Delta_{\ell,\F})_{i\ell} = 0$ for all $i=1,\dots,k$. In view of \eqref{Hdef2} and \eqref{Gperp} we find in the case $\gamma=0$ that
\begin{equation} \label{Gperp2}
G_{t,\ell,\F}\big( \vectheta ,  \vecy_1,\ldots,\vecy_k\big) \\
= 
\scrW(\iota_{\F}(\vecy_1,\ldots,\vecy_k))\, \widehat\scrA_{t,\ell,\F}(\vectheta,\vecy_1,\ldots,\vecy_{k})  
\end{equation}
with
\begin{equation} \label{scrAdef2}
\widehat\scrA_{t,\ell,\F}(\vectheta,\vecy_1,\ldots,\vecy_{k}) = \int_{\Box_{\ell,n}(t)} \scrA_{\ell,n}(\iota_{\F}(\vecy_1,\ldots,\vecy_{k}),\vecu) \, \delta(\vectheta-\Delta_{\ell,\F}\vecu)\, \d^\perp\vecu,
\end{equation}
and
\begin{equation}\label{exdelta}
\delta(\vectheta-\Delta_{\ell,\F}\vecu)= 
\prod_{\substack{i=1\\ \ell \notin F_i}}^k \delta\bigg(\theta_i+\sum_{j=0}^{\ell-1}  u_j \mathbb{1}(j\in F_i)-\sum_{j=\ell+1}^{n} u_j  \mathbb{1}(j\in F_i) \bigg) .
\end{equation}
Thus if $\ell\in F_l$
\begin{multline}\label{Hdef01234}
H_{t,\ell,\F}(\vecxi,\vecy_1,\ldots,\vecy_k)  =\scrW(\iota_{\F}(\vecy_1,\ldots,\vecy_k)) 
 \int_{\Sigma^\perp}
 \e\bigg(\sum_{\substack{i=1\\ i\neq l}}^k  \theta_i (\xi_i-\xi_l) \bigg)  \\
\times \widehat\scrA_{t,\ell,\F}(\vectheta,\vecy_1,\ldots,\vecy_{k}) \, \d^\perp\vectheta.
\end{multline}
More generally, for $\Re \gamma \geq 0$, we have the formula
\begin{multline}\label{Hdef01234gamma}
H_{t,\ell,\F}^\gamma(\vecxi,\vecy_1,\ldots,\vecy_k)  
=\scrW(\iota_{\F}(\vecy_1,\ldots,\vecy_k))  \\
\times
\int_{\Sigma^\perp} \e\bigg(\sum_{\substack{i=1\\ i\neq l}}^k \big[ \theta_i (\xi_i-\xi_l) + \i |\theta_i| \gamma\big] \bigg)   
\widehat\scrA_{t,\ell,\F}(\vectheta,\vecy_1,\ldots,\vecy_{k}) \, \d^\perp\vectheta .
\end{multline}
Our task is now to determine the convergence of
\begin{equation} \label{decayrecap}
 \begin{split}
& \sum_{n=1}^\infty  (2 \pi \i \lambda h^{-2})^n\; \scrJ_{\ell,n}^\gamma(t h r^{1-d}) \\
& =  \sum_{n=1}^\infty (2 \pi \i \lambda)^n\sum_{\ell=0}^n (-1)^{\ell} \sum_{k=1}^n \sum_{\F \in \Pi_\circ(n,k)} \int_{\RR^{dk}}\scrW(\iota_{\F}(\vecy_1,\ldots,\vecy_k))  \\
& \times
\int_{\Sigma^\perp} \e\bigg(\sum_{\substack{i=1\\ i\neq l}}^k \big[ \theta_i (\xi_i-\xi_l) + \i |\theta_i| \gamma\big] \bigg)   
\widehat\scrA_{t,\ell,\F}(r^d \vectheta,\vecy_1,\ldots,\vecy_{k}) \, \d^\perp\vectheta\,  \d \vecy_1 \cdots \d \vecy_k 
\end{split}
\end{equation}
as $r=h\to 0$ and calculate the limit. 

\section{Explicit formulas} \label{sec:explicit}
 In this section we provide more explicit formulas for $\widehat\scrA_{t,\ell,\F}(\vectheta,\vecy_1,\ldots,\vecy_{k})$. The main results are the expression \eqref{lem6.1proof1003} for general $\vectheta$ and \eqref{eq716} for $\vectheta = \vecnull$. We assume that $\F$ is some set partition into $k$ blocks with $\ell \in F_l$. We define
\begin{equation}
\begin{split}
I_-(\F) &= \{ i \in \{1,\dots,k\}\setminus \{ l\} \, \mid F_i \cap [0,\ell-1] \neq\emptyset \} , \\
I_+(\F) &= \{ i \in \{1,\dots,k\}\setminus \{l\} \, \mid F_i \cap [\ell+1,n] \neq\emptyset \},
\end{split}
\end{equation}
and similarly
\begin{equation}
\begin{split}
I_-^*(\F) &= \{ i \in \{1,\dots,k\}\setminus \{l\} \, \mid F_i \cap [\ell+1,n] =\emptyset \}, \\
I_+^*(\F) &= \{ i \in \{1,\dots,k\}\setminus \{l\} \, \mid F_i \cap [0,\ell-1] =\emptyset \}.
\end{split}
\end{equation}
We also write
\begin{equation}
J_- = I_- \cap I_-^*,\quad  J_+ = I_+ \cap I_+^*,\quad J = I_- \cap I_+. 
\end{equation}
Note that $J_-$ (resp. $J_+$) contains indices corresponding to \emph{one-sided} blocks of the partition, i.e. blocks, all of whose elements are less than or equal to (resp. greater than or equal to)  $\ell$. The set $J$ contains indices corresponding to blocks which are not one-sided, i.e. they contain elements both less than and greater than $\ell$. This provides a complete categorisation of blocks $F_i$:
 \begin{equation}
 \{1,\dots,k\} = \{l\} \cup J_- \cup J_+ \cup J.
 \end{equation}
Let us furthermore define
\begin{equation}
\mu_i=\mu_{i,\ell,\F} =  |F_i \cap [0,\ell]|-1, \qquad \nu_i=\nu_{i,\ell,\F} =  |F_i \cap [\ell,n]|-1 .
\end{equation}
Combining \eqref{scrAdef} and \eqref{scrAdef2} we have
\begin{multline} \label{lem6.1proof1}
\widehat\scrA_{t,\ell,\F}(\vectheta,\vecy_1,\ldots,\vecy_{k})  \\
=\int_{\Box_{\ell,n}(t)} \left( \int_{\RR^d}   a\bigg( \vecx - \sum_{\substack{i=1\\ i\neq l}}^k(\vecy_i-\vecy_l)\sum_{j \in F_i \cap [\ell+1,n]}  u_j ,\vecy_l \bigg) b( \vecx, \vecy_1) \, \d \vecx \right) \\
\times \prod_{\substack{i=1\\ i\neq l}}^k \delta\bigg(\theta_i+\sum_{j \in F_i \cap [0,\ell-1]}  u_j -\sum_{j \in F_i \cap [\ell+1,n]} u_j \bigg) 
 \d^\perp\vecu.
\end{multline}
First we can freely integrate over all $u_j$ for $j \in F_l$. This yields a factor of 
\begin{equation}
\frac{(t- \sum_{j=0, j\notin F_l}^{\ell-1}  u_j)_+^{\mu_l} (t- \sum_{j=\ell+1, j\notin F_l}^{n} u_j)_+^{\nu_l}}{\mu_l! \nu_l!}
\end{equation}
where $(x)_+ = \max \{0 , x \}$. If we use the convention $(x)_+^0 = \one [x > 0]$ then, writing $r_i = \sum_{j \in F_i \cap [0,\ell-1]} u_j$ and $s_i = \sum_{j \in F_i \cap [\ell+1,n]} u_j$, we obtain
\begin{equation} \label{lem6.1proof1001}
\begin{split}
& \widehat\scrA_{t,\ell,\F}(\vectheta,\vecy_1,\ldots,\vecy_{k})  \\
& = \int_{\RR_{\geq 0}^{|I_-|}}\int_{\RR_{\geq 0}^{|I_+|}} \bigg( \int_{\RR^d}   a\bigg( \vecx - \sum_{i\in I_+}(\vecy_i-\vecy_l) s_i ,\vecy_l \bigg) b( \vecx, \vecy_1) \, \d \vecx \bigg) \\
& \times \frac{(t- \sum_{i \in I_- } r_i)_+^{\mu_l} (t- \sum_{i \in I_+ } s_i)_+^{\nu_l}}{\mu_l! \nu_l!}
\prod_{i\in J_-} \Lambda_i^-(r_i) \delta(\theta_i+r_i)\,\d r_i\\
& \times \prod_{i\in J_+} \Lambda_i^+(s_i) \delta(\theta_i -s_i ) \,\d s_i 
\prod_{i\in J} \Lambda_i^-(r_i) \Lambda_i^+(s_i) \delta(\theta_i+r_i -s_i ) \,\d r_i\,\d s_i ,
\end{split}
\end{equation}
 with
\begin{equation}
\Lambda_i^-(r_i) = \int_{\RR_{\geq 0}^{\mu_i+1}} \delta\bigg( r_i - \sum_{j \in F_i \cap [0,\ell-1]} u_j \bigg) \prod_{j \in F_i \cap [0,\ell-1]} \d u_j
=
\frac{r_i^{\mu_i}}{\mu_i !}
\end{equation}
and
\begin{equation}
\Lambda_i^+(s_i) = \int_{\RR_{\geq 0}^{\nu_i+1}} \delta\bigg( s_i - \sum_{j \in F_i \cap [\ell+1,n]} u_j \bigg) \prod_{j \in F_i \cap [\ell+1,n]} \d u_j
=
\frac{s_i^{\nu_i}}{\nu_i !}.
\end{equation}
In other words
\begin{equation} \label{lem6.1proof1002}
\begin{split}
&\widehat\scrA_{t,\ell,\F}(\vectheta,\vecy_1,\ldots,\vecy_{k})  \\
& =\int_{\RR_{\geq 0}^{|I_-|}}\int_{\RR_{\geq 0}^{|I_+|}}  \bigg( \int_{\RR^d}   a\bigg( \vecx - \sum_{i\in I_+}(\vecy_i-\vecy_l) s_i ,\vecy_l \bigg) b( \vecx, \vecy_1) \, \d \vecx \bigg) \\
& \times \frac{(t- \sum_{i \in I_- } r_i)_+^{\mu_l} (t- \sum_{i \in I_+ } s_i)_+^{\nu_l}}{\mu_l!\nu_l!}
  \prod_{i\in J_-} \frac{r_i^{\mu_i}}{\mu_i !} \delta(\theta_i+r_i)\,\d r_i\\
& \times \prod_{i\in J_+} \frac{s_i^{\nu_i}}{\nu_i !}  \delta(\theta_i -s_i ) \,\d s_i 
\prod_{i\in J} \frac{r_i^{\mu_i} s_i^{\nu_i}}{\mu_i ! \nu_i !} \delta(\theta_i+r_i -s_i ) \,\d r_i\,\d s_i .
\end{split}
\end{equation}
Integrating over $r_i$ for $i \in I_- = J_- \cup J$ and $s_i$ for $i \in J_+$ yields
\begin{equation} \label{lem6.1proof1003}
\begin{split}
& \widehat\scrA_{t,\ell,\F}(\vectheta,\vecy_1,\ldots,\vecy_{k})  \\
& =\int_{\RR_{\geq 0}^{|J|}} \bigg( \int_{\RR^d}   a\bigg( \vecx - \sum_{i\in  J_+}(\vecy_i-\vecy_l) \theta_i - \sum_{i\in J}(\vecy_i-\vecy_l) s_i ,\vecy_l \bigg) b( \vecx, \vecy_1) \, \d \vecx \bigg) \\
&\times \frac{1}{\mu_l!\nu_l!} \left( t + \sum_{i \in J_- } \theta_i - \sum_{i \in J} (s_i-\theta_i) \right)_+^{\mu_l} \left(t- \sum_{i \in J_+ } \theta_i - \sum_{i \in J} s_i \right)_+^{\nu_l}\\
&\times \prod_{i\in J_-} \frac{(-\theta_i)_+^{\mu_i}}{\mu_i !} \prod_{i\in J_+} \frac{(\theta_i)_+^{\nu_i}}{\nu_i !} 
\prod_{i\in J} \frac{(s_i-\theta_i)_+^{\mu_i} s_i^{\nu_i}}{\mu_i ! \nu_i !}\d\vecs.
\end{split}
\end{equation}
When $\vectheta \to \vecnull$ note that the integrand vanishes unless $\mu_i = 0$ for all $i \in J_-$ and $\nu_i = 0$ for all $i \in J_+$, that is to say that a partition $\F$ only contributes in the limit if the only one-sided blocks are singletons. This motivates the definition of $\Omega(n,k) \subset \{0,\ldots, n\}\times \Pi_\circ(n,k)$ as the set of marked partitions $(\ell,\F)$ such that every block $F_i$ not containing $\ell$ either (i) is a singleton or (ii) contains at least one number strictly less than $\ell$ and at least one number strictly greater than $\ell$ (cf.~Appendix \ref{subsec:Marked}). 

The support of $\widehat\scrA_{t,\ell,\F}(\vectheta,\vecy_1,\ldots,\vecy_{k})$ as a function of $\vectheta$ is contained in the domain
\begin{equation}
\{\vectheta\in\Sigma^\perp \mid -t\leq\theta_i\leq 0\, \forall i\in J_-,\;  0\leq \theta_i\leq t\, \forall i\in J_+ \};
\end{equation}
and is continuous in this domain in a sufficiently small neighbourhood of the origin. Putting $\vectheta=\vecnull$ in \eqref{lem6.1proof1003} yields
\begin{multline} \label{lem6.1proof1004}
\widehat\scrA_{t,\ell,\F}(\vecnull,\vecy_1,\ldots,\vecy_{k}) =\int_{\RR_{\geq 0}^{|J|}} \bigg( \int_{\RR^d}   a\bigg( \vecx - \sum_{i\in J}(\vecy_i-\vecy_l) s_i ,\vecy_l \bigg) b( \vecx, \vecy_1) \, \d \vecx \bigg) \\
\times \frac{1}{\mu_l!\nu_l!} \left( t - \sum_{i \in J} s_i \right)_+^{\mu_l+\nu_l}\prod_{i\in J_- \cup J_+ } \one(|F_i|=1) \prod_{i\in J} \frac{s_i^{\mu_i+\nu_i}}{\mu_i ! \nu_i !} \,\d\vecs.
\end{multline}
In other words, for $(\ell,\F)\in\Omega(n,k)$ we have
\begin{multline}\label{eq716}
\widehat\scrA_{t,\ell,\F}(\vecnull,\vecy_1,\ldots,\vecy_{k}) 
=\int_{\RR_{\geq 0}^k} 
\bigg(\int_{\RR^d}   a\bigg( \vecx - \sum_{j=1}^k u_j (\vecy_j-\vecy_l),\vecy_l \bigg) b( \vecx, \vecy_1) \, 
 \d \vecx\bigg) \\ 
\times \bigg(\prod_{i=1}^k \frac{u_i^{\mu_{i}+\nu_{i}} }{\mu_{i}! \nu_{i}!}\bigg) \, \delta(u_1+\cdots+u_k-t) \, \d \vecu ,
\end{multline}
whereas for $(\ell,\F)\notin\Omega(n,k)$ we have that $\widehat\scrA_{t,\ell,\F}(\vecnull,\vecy_1,\ldots,\vecy_{k})=0$.

\section{Decay estimates} \label{sec:decay}

 In this section we prove Proposition \ref{cor:Hseries} which establishes the absolute and uniform convergence of the series defining \eqref{decayrecap}. The proof will be similar in spirit to that of Proposition \ref{prop51}, and the key ingredient is the following decay estimate. 

\begin{prop}\label{prop91} 
Let $a,b\in\scrS(\T(\RR^d))$. There exists a constant $C_{a,b,d}$ such that for all $W\in\scrS(\RR^d)$, $t>0$, $\Re\gamma\geq 0$, $r>0$,
\begin{equation}
\begin{split}
& \bigg|\int_{(\RR^d)^k} \scrW(\iota_{\F}(\vecy_1,\ldots,\vecy_k)) 
 \e\bigg(\sum_{\substack{i=1\\ i\neq l}}^k  \big[\theta_i (\tfrac12\|\vecy_i\|^2-\tfrac12\|\vecy_l\|^2) + \i |\theta_i| \gamma\big] \bigg)  \\[-0.3cm]
& \qquad \times \widehat\scrA_{t,\ell,\F}(r^d \vectheta,\vecy_1,\ldots,\vecy_{k}) \,\d\vecy_1\cdots\d\vecy_k \bigg| \\
& \leq  C_{a,b,d}^n \langle t \rangle^n \, \|W\|_{2d+2,d+1,1}^n\, \bigg(\prod_{i=1}^k  \frac{1}{|F_i|!} \bigg) \, \bigg(\prod_{\substack{i=1\\ i \neq l}}^k \e^{-2 \pi |\theta_i| \gamma} \langle \theta_i \rangle^{-d/2} \bigg).
\end{split}
\end{equation}
\end{prop}
\begin{proof}
We first write
\begin{equation}
\begin{split}
& \bigg|\int_{(\RR^d)^k} \scrW(\iota_{\F}(\vecy_1,\ldots,\vecy_k)) 
 \e\bigg(\sum_{\substack{i=1\\ i\neq l}}^k  \big[\theta_i (\tfrac12\|\vecy_i\|^2-\tfrac12\|\vecy_l\|^2) + \i |\theta_i| \gamma\big] \bigg)  \\
&\qquad \times \widehat\scrA_{t,\ell,\F}(r^d \vectheta,\vecy_1,\ldots,\vecy_{k}) \,\d\vecy_1\cdots\d\vecy_k \bigg| \\
& \leq \e^{-2 \pi \sum_{i \neq l} |\theta_i| \gamma}  \bigg| \int_{(\RR^d)^{k}} \scrW(\iota_{\F}(\vecy_1,\ldots,\vecy_k)) 
 \e\bigg(\tfrac12 \sum_{i=1}^k \theta_i \|\vecy_i\|^2 \bigg)  \\
& \qquad \times \widehat\scrA_{t,\ell,\F}(r^d \vectheta,\vecy_1,\ldots,\vecy_{k}) \,\d\vecy_1\cdots\d\vecy_k\bigg|
\end{split}
\end{equation}
where we have used the shorthand $\theta_l := -\sum_{i=1,i \neq l}^k \theta_i$. The inner integral can be treated using Lemma \ref{lem51}, so all we must do is compute the relevant norm appearing inside the Lemma. Let $K=\{ 1,\dots,k \}$ and $S =\{s_1,\dots,s_p\} \subset K$. We need to compute the norm $\| f_S\|_{L^1}$ where
\begin{equation}
 f(\vecy_1,\dots,\vecy_k) = \scrW(\iota_{\F}(\vecy_1,\dots,\vecy_k)) \, \widehat\scrA_{t,\ell,\F}(r^d \vectheta,\vecy_1,\dots,\vecy_k).
\end{equation}
We have
\begin{multline}
\|f_S\|_{L^1} = \int_{\RR^{dk}}\bigg| \int_{\RR^{d|S|}} \scrW(\iota_{\F}(\vecy_1,\dots,\vecy_k))
\, \widehat\scrA_{t,\ell,\F}(r^d \vectheta,\vecy_1,\dots,\vecy_k) \\
\times [\prod_{i \in S} \e(\vecy_{i}\cdot \vecx_{i}) \d \vecy_{i}] 
\bigg| [\prod_{i\in K\setminus S} \d \vecy_i] [\prod_{i\in S} \d \vecx_{i}] 
\end{multline}
and, more explicitly,
\begin{equation}
\begin{split}
& \|f_S\|_{L^1} = \int_{\RR^{dk}}\bigg| \int_{\RR^{d|S|}} \scrW(\iota_{\F}(\vecy_1,\dots,\vecy_k)) \, \int_{\RR_{\geq 0}^{|J|}}  \int_{\RR^d} \tilde a(\veceta,\vecy_l) \, \tilde b(-\veceta,\vecy_1)   \\
&\times \frac{1}{\mu_l! \nu_l!} \bigg( t + r^d\sum_{i \in J_- } \theta_i - \sum_{i \in J} (s_i-r^d\theta_i) \bigg)_+^{\mu_l} \bigg(t- r^d\sum_{i \in J_+ } \theta_i - \sum_{i \in J} s_i \bigg)_+^{\nu_l}\\
&\times \prod_{i\in J_-} \frac{(-r^d\theta_i)_+^{\mu_i}}{\mu_i !} \prod_{i\in J_+} \frac{(r^d\theta_i)_+^{\nu_i}}{\nu_i !} 
\prod_{i\in J} \frac{(s_i-r^d\theta_i)_+^{\mu_i} s_i^{\nu_i}}{\mu_i ! \nu_i !}\\
&\times [\prod_{i \in S} \e(\vecy_{i}\cdot (\vecx_{i}-\veceta \tau_{i})) \d \vecy_{i}]  \,\d\vecs \, \d \veceta \bigg| [\prod_{i\in K\setminus S} \d \vecy_i] [\prod_{i\in S} \d \vecx_{i}]
\end{split}
\end{equation}
where we write 
\begin{equation}
\tau_i = \begin{cases} r^d \theta_i & i \in J_+ \\ s_i & i \in J . \end{cases}
\end{equation}
We pull the integral over $\veceta$ and $\vecs$ outside the absolute value and bound $\|f_S\|_{L^1}$ above by
\begin{equation}
\begin{split}
& \int_{\RR^{d (k+1)}}\bigg| \int_{\RR^{d|S|}} \scrW(\iota_{\F}(\vecy_1,\dots,\vecy_k)) \,\tilde a(\veceta,\vecy_l) \, \tilde b(-\veceta,\vecy_1)  \\
& \times [\prod_{i \in S} \e(\vecy_{i}\cdot \vecx_{i} )\d \vecy_{i}]  \bigg| [\prod_{i\in K\setminus S} \d \vecy_i] [\prod_{i\in S} \d \vecx_{i}] \, \d \veceta\\
& \times \int_{\RR_{\geq 0}^{|J|} } \frac{1}{\mu_l! \nu_l!} \bigg( t + r^d\sum_{i \in J_- } \theta_i - \sum_{i \in J} (s_i-r^d\theta_i) \bigg)_+^{\mu_l} \bigg(t- r^d\sum_{i \in J_+ } \theta_i - \sum_{i \in J} s_i \bigg)_+^{\nu_l}\\
& \times \prod_{i\in J_-} \frac{(-r^d\theta_i)_+^{\mu_i}}{\mu_i !} \prod_{i\in J_+} \frac{(r^d\theta_i)_+^{\nu_i}}{\nu_i !} 
\prod_{i\in J} \frac{(s_i-r^d\theta_i)_+^{\mu_i} s_i^{\nu_i}}{\mu_i ! \nu_i !} \, \d \vecs.
\end{split}
\end{equation}
This can be bounded above by
\begin{multline}
\frac{t^{|J|}}{|J|!}[\prod_{i=1}^k \frac{t^{\tilde \mu_i+ \tilde\nu_i}}{\tilde\mu_i! \tilde\nu_i!}]\int_{\RR^{d (k+1)}} \bigg| \int_{\RR^{d|S|}} \scrW(\iota_{\F}(\vecy_1,\dots,\vecy_k))  \, \tilde a(\veceta,\vecy_l) \, \tilde b(-\veceta,\vecy_1)  \\
\times [\prod_{i \in S} \e(\vecy_{i}\cdot \vecx_{i} )\d \vecy_{i}]  \bigg| [ \prod_{i\in K\setminus S} \d \vecy_i]  [\prod_{i\in S} \d \vecx_{i}] \, \d \veceta
\end{multline}
where we use the convention $\tilde \mu_i = (\mu_i)_+ = \max\{0,\mu_i\}$. As in the proof of Proposition \ref{prop51} we bound this above by

\begin{multline}
\frac{t^{|J|}}{|J|!}[\prod_{i=1}^k \frac{t^{\tilde \mu_i+ \tilde\nu_i}}{\tilde\mu_i! \tilde\nu_i!}] \,\sum_{|\vecm_{s_1}|<d+1} \begin{pmatrix} d+1 \\ \vecm_{s_1} \end{pmatrix} \cdots \sum_{|\vecm_{s_k}| < d+1} \begin{pmatrix} d+1 \\ \vecm_{s_k} \end{pmatrix} \\
\int_{\RR^{d (k+1)}} \bigg| [\prod_{i \in S} \vecx_i^{\vecm_i}] \int_{\RR^{d|S|}} \scrW(\iota_{\F}(\vecy_1,\dots,\vecy_k))  \, \tilde a(\veceta,\vecy_l) \, \tilde b(-\veceta,\vecy_1)  \\
\times [\prod_{i \in S} \e(\vecy_{i}\cdot \vecx_{i} )\d \vecy_{i}]  \bigg| [ \prod_{i\in K\setminus S} \d \vecy_i]  [\prod_{i\in S} \langle \vecx_i\rangle^{-d-1} \d \vecx_{i}] \, \d \veceta.
\end{multline}
Integrating by parts with respect to $\vecy_i$ for $i\in S$ and pulling the absolute value inside the integral gives the upper bound
\begin{multline}
\frac{t^{|J|}}{|J|!}[\prod_{i=1}^k \frac{t^{\tilde \mu_i+ \tilde\nu_i}}{\tilde\mu_i! \tilde\nu_i!}] \, \left(\int_{\RR^d} \langle \vecx \rangle^{-d-1} \d \vecx \right)^{p} \,\sum_{|\vecm_{s_1}|<d+1} \begin{pmatrix} d+1 \\ \vecm_{s_1} \end{pmatrix} \cdots \sum_{|\vecm_{s_p}| < d+1} \begin{pmatrix} d+1 \\ \vecm_{s_p} \end{pmatrix} \\
\int_{\RR^{d (k+1)}} \bigg| [\prod_{i \in S} D_{\vecy_i}^{\vecm_i}] \scrW(\iota_{\F}(\vecy_1,\dots,\vecy_k))  \, \tilde a(\veceta,\vecy_l) \, \tilde b(-\veceta,\vecy_1)  \bigg| \d \vecy_1 \cdots \d \vecy_k  \, \d \veceta.
\end{multline}
If $l \neq 1$, each $\vecy_i$ appears $2|F_i|$ times in the product of $\scrW$, $\tilde a$ and $\tilde b$, except for $\vecy_1$ and $\vecy_l$ which appear $2|F_i| + 1$ times (due to their appearance in $\tilde a$ and $\tilde b$). If $l=1$, then each $\vecy_i$ appears $2|F_i|$ times, except for $\vecy_1$ which appears $2|F_i|+2$ times. Either way we can find a constant $c_d$ such that the number of terms inside the absolute value can be bounded above by
$$ c_d^{pd} |F_1|^{d+1} \cdots |F_k|^{d+1}.$$
Applying the triangle inequality yields a sum of terms of the form
\begin{equation} \label{eightten}
\int_{\RR^{d (k+1)}} \bigg| \Phi(\iota_{\F}(\vecy_1,\dots,\vecy_k))  \,\mathfrak{a} (\veceta,\vecy_l)  \mathfrak{b} (-\veceta,\vecy_1)  \bigg| \d \vecy_1 \cdots \d \vecy_k \, \d \veceta
\end{equation}
where
\begin{equation}
\Phi(\veceta_0,\veceta_1,\dots,\veceta_n) = \varphi_1(\veceta_0-\veceta_1) \cdots \varphi_{n }(\veceta_{n-1}-\veceta_n),
\end{equation}
each $\varphi_i$ is a derivative of $\hat W$ of order $\leq 2d+2$ and $\mathfrak{a}$ and $\mathfrak{b}$ are derivatives of $\tilde a$ and $\tilde b$ with respect to the second argument of order $\leq d+1$. Define the map $\kappa : \{0,\dots,n\} \to \{1,\dots,k\}$ implicitly dependent on $\F$ by $\kappa(i) = j$ if $i \in F_j$. We then have that
\begin{equation}
\Phi(\iota_{\F}(\vecy_1,\dots,\vecy_k) ) = \varphi_1(\vecy_{\kappa(0)}-\vecy_{\kappa(1)}) \cdots \varphi_{n}(\vecy_{\kappa(n-1)} - \vecy_{\kappa(n)}).
\end{equation}
By the definition of $\Omega(n,k)$ we have that $\kappa(0)= \kappa(n) = 1$. For $i = 2,\dots,k$, we define the partial inverse
$$ \kappa^{-1}(i) = \min \{ j \in \{1,\dots,n-1\} \mid \kappa(j) = i \}.$$
For $i =2,\dots,k$, the first factor of $\Phi(\iota_{\F}(\vecy_1,\dots,\vecy_k))$ in which $\vecy_i$ appears is then by definition
$$ \varphi_{\kappa^{-1}(i)}(\vecy_{\kappa(\kappa^{-1}(i)-1)} -\vecy_i).$$
We define $\scrK \subset \{1,\dots,n\}$ to be the image of $\kappa^{-1}(\{2,\dots, k\}\setminus\{l\})$. Equation \eqref{eightten} can thus be bounded above by 
\begin{equation}
\begin{split}
& [\prod_{i \in \{1,\dots,n\} \setminus \scrK} \| \varphi_i \|_{L^\infty} ] \\
& \times \int_{\RR^{d (k+1)}} \bigg| \mathfrak{a} (\veceta,\vecy_l) \mathfrak{b}(-\veceta,\vecy_1) [\prod_{i=1}^{|\scrK|} \varphi_{i}(\vecy_{\kappa(i-1)} - \vecy_{\kappa(i)}) ]  \bigg| \d \vecy_1 \cdots \d \vecy_k \, \d \veceta.
\end{split}
\end{equation}
Make the variable substitutions $\vecy_{\kappa(i)} \to \vecy_{\kappa(i-1)}-\vecy_{\kappa(i)}$ to bound this above by
\begin{equation}
\bigg[\prod_{i \in \{1,\dots,n\} \setminus \scrK} \| \varphi_i \|_{L^\infty} \bigg] \left(\int_{\RR^{3d}} | \mathfrak{a}(\veceta,\vecy) \mathfrak{b}(-\veceta,\vecy')|  \d \vecy\d\vecy'\d\veceta\right) \, \bigg[\prod_{i=1}^{|\scrK|} \| \varphi_i\|_{L_1} \bigg].
\end{equation}
As in \eqref{infnorms} we have that
\begin{equation}
\|\varphi_i\|_{L^\infty} \leq \| W \|_{2d+2,0,1}
\end{equation}
and as in \eqref{onenorms} we have that
\begin{equation}
\| \varphi_i\|_{L^1} \leq   \left(\int_{\RR^d} \langle \vecy \rangle^{-d-1} \d \vecy \right) \, (2d+1)^{d+1} \, \| W \|_{2d+2,d+1,1}.
\end{equation}
Combining with the fact that the functions $\tilde a$ and $\tilde b$ are Schwartz class implies that there exists a constant $ c_{a,b,d}$ such that
\begin{equation}
\| f_S\|_{L^1} \leq  c_{a,b,d}^n \| W \|_{2d+2,d+1,2}^n \, \frac{t^{|J|}}{|J|!}[\prod_{i=1}^k \frac{t^{\tilde\mu_i+\tilde\nu_i}}{\tilde\mu_i! \tilde\nu_i!} |F_i|^{d+1}].
\end{equation}
We then use the fact that 
\begin{equation}
|F_i| = \begin{cases} \mu_i+1 & i \in J_- \\ \nu_i + 1 & i \in J_+ \\ \mu_i+\nu_i+2 & i \in J \\ \mu_i+\nu_i+1 & i = l
\end{cases}
\end{equation}
so in particular for $i \in J_+$ or $i \in J_-$
\begin{equation}
\frac{1}{\mu_i!} = \frac{|F_i|}{|F_i|!} \quad \text{ or } \quad \frac{1}{\nu_i!} = \frac{|F_i|}{|F_i|!}
\end{equation}
respectively. For $i \in J$ we have
\begin{equation}
\frac{1}{\mu_i! \nu_i!} \leq \frac{|F_i|^2}{|F_i|!} 2^{|F_i|-2} 
\end{equation}
and for $i = l$ we have 
\begin{equation}
\frac{1}{\mu_l! \nu_l!} \leq \frac{|F_i|}{|F_i|!} 2^{|F_i|-1}. 
\end{equation}
For simplicity we bound  all of these uniformly by 
\begin{equation}
\frac{|F_i|^2}{|F_i|!} 2^{|F_i|}.
\end{equation}
Using the fact that $|J| + \sum_{i=1}^k (\tilde\mu_i+\tilde\nu_i) = n-k+1$ thus yields
\begin{equation}
\|f_S\|_{L^1} \leq  c_{a,b,d}^n \frac{t^{n-k+1}}{|J|!}  \,  \| W \|_{2d+2,d+1,2}^n  [\prod_{i=1}^k \frac{|F_i|^{d+3}}{|F_i|!} \, 2^{|F_i|} ].
\end{equation}
Finally, observe that, since $\sum_{i=1}^k |F_i| = n+1$, we have that
\begin{equation}
\prod_{i=1}^k 2^{|F_i|} = 2^{n+1}, \quad \prod_{i=1}^k |F_i|^{d+3} < 2^{(n+1)(d+3)}.
\end{equation}
This completes the proof.
\end{proof}

This upper bound allows us to ensure convergence of the series \eqref{decayrecap}.

\begin{prop} \label{cor:Hseries}
Let $a,b\in\scrS(\T(\RR^d))$, $W\in\scrS(\RR^d)$, $t>0$ and set
\begin{equation}\label{Rdef}
R= \bigg(2 \pi \, C_{a,b,d} \, \langle t \rangle \|W\|_{2d+2,d+1,1} \max \left\{ 1 , \int_{\RR} \langle \theta \rangle^{-d/2} \d \theta\right\}\bigg)^{-1}.
\end{equation}
Then the series 
\begin{equation}
\sum_{n=1}^\infty  (2 \pi \i \lambda h^{-2})^n \; \scrJ_{\ell,n}^\gamma(t  r^{2-d})
\end{equation}
in \eqref{decayrecap} converges absolutely for all $|\lambda| < R$, uniformly in $\Re\gamma\geq 0$, $0< r\leq 1$.
\end{prop}

\begin{proof}
First of all, by integrating over $\vectheta$ we can obtain from Proposition \ref{prop91}
\begin{equation} \label{decayrecap2}
 \begin{split}
&  \sum_{n=1}^\infty (2 \pi \i |\lambda|)^n\sum_{\ell=0}^n \sum_{k=1}^n \sum_{\F \in \Pi_\circ(n,k)} \bigg| \int_{\RR^{dk}}\scrW(\iota_{\F}(\vecy_1,\ldots,\vecy_k))  \\
& \times
\int_{\Sigma^\perp}  \e\bigg(\sum_{\substack{i=1\\ i\neq l}}^k  \big[\theta_i (\tfrac12\|\vecy_i\|^2-\tfrac12\|\vecy_l\|^2) + \i |\theta_i| \gamma\big] \bigg)  
\widehat\scrA_{t,\ell,\F}(r^d \vectheta,\vecy_1,\ldots,\vecy_{k}) \, \d^\perp\vectheta  \d \vecy_1 \cdots \d \vecy_k \bigg| \\
& \leq \sum_{n=1}^\infty (n+1) (2 \pi A |\lambda| \langle t \rangle \|W\|_{2d+2,d+1,1})^n   \sum_{k=1}^n \sum_{\F \in \Pi_\circ(n,k)} \left( \prod_{i=1}^k \frac{1}{|F_i|!} \right) ,
\end{split}
\end{equation}
where
$$ A =  C_{a,b,d} \,\max \left\{ 1 , \int_{\RR} \langle \theta \rangle^{-d/2} \d \theta\right\}.$$
We can replace the set $\Pi_\circ(n,k)$ by the set of all partitions of $\{1,\dots,n\}$ into $k$ blocks to obtain the upper bound
\begin{equation}
\sum_{\F \in \Pi_\circ(n,k)} \bigg( \prod_{i=1}^k \frac{1}{|F_i|!} \bigg) < \frac{1}{(n+1)!} \sum_{m_1+\dots+m_k = n+1} \begin{pmatrix} n+1 \\ m_1, m_2, \dots, m_{k}\end{pmatrix} =\frac{ k^{n+1}}{(n+1)!}.
\end{equation}
Inserting this into our upper bound yields
\begin{equation}
\eqref{decayrecap2}  < \sum_{n=1}^\infty  (2 \pi A |\lambda| \langle t \rangle \|W\|_{2d+2,d+1,1})^n  \frac{ n^{n+1}}{(n-1)!}
\end{equation}
which converges for $2 \pi A |\lambda| \langle t \rangle \|W\|_{2d+2,d+1,1} < \e^{-1}$ by Stirling's formula.
\end{proof}

\section{The microlocal Boltzmann-Grad limit} \label{sec:takingthelimit}
 In this section we combine the results of Sections \ref{sec:explicit} and \ref{sec:decay} to prove Proposition \ref{thm:multieq} which establishes the limit of the full perturbation series. Given $(\ell,\F)\in\Omega(n,k)$, 
let $\scrC^\perp_{\ell,\F}$ be the set of $\vectheta\in\Sigma^\perp$ with $i$th coordinate $\theta_i$ ranging over
\begin{equation} \label{scrCdef}
\begin{cases} 
\RR_{\leq 0}  & \text{if $F_i=\{j\}$ with $j<\ell$} \\
\RR_{\geq 0}  & \text{if $F_i=\{j\}$ with $j>\ell$} \\
0 & \text{if $\ell\in F_i$}\\
\RR  & \text{if $\ell\notin F_i$ and $|F_i|>1$.}
\end{cases}
\end{equation}
For $(\ell,\F)\in\Omega(n,k)$ we define
\begin{equation}
D_{\ell,\F}^\gamma(\vecy_1,\ldots,\vecy_k) = \int_{\scrC^\perp_{\ell,\F}}  \e\bigg(\sum_{\substack{i=1\\ i\neq l}}^k  \big[\theta_i (\tfrac12\|\vecy_i\|^2-\tfrac12\|\vecy_l\|^2) + \i |\theta_i| \gamma\big] \bigg)   \, \d^\perp\vectheta ,
\end{equation}
which converges for $\Re\gamma>0$ and can be extended (in the distributional sense) by analytic continuation to $\Re\gamma\geq 0$.
In other words, we have that
\begin{equation}
D_{\ell,\F}^\gamma(\vecy_1,\ldots,\vecy_k)=  \prod_{\substack{i=1\\ \ell\notin F_i}}^k d_{\ell,\F,i}^\gamma(\vecy_l,\vecy_i) 
\end{equation}
with
\begin{equation} \label{ddef}
d_{\ell,\F,i}^\gamma(\vecy,\vecy')=\frac{1}{2\pi\i} \times
\begin{cases} 
-g^{\gamma}(\vecy,\vecy') & \text{if $F_i=\{j\}$ with $j<\ell$} \\
\overline {g^{\gamma}(\vecy,\vecy')} & \text{if $F_i=\{j\}$ with $j>\ell$} \\
\overline{g^\gamma(\vecy,\vecy')} - g^\gamma(\vecy,\vecy')  & \text{if $|F_i|>1$,}
\end{cases}
\end{equation}
with $g^\gamma$ as in \eqref{Tg}. Note furthermore that 
\begin{equation}
\frac{\overline{g^\gamma(\vecy,\vecy')} - g^{\gamma}(\vecy,\vecy')}{2\pi\i}  = \frac{1}{\pi}\;\frac{\gamma}{(\tfrac12\|\vecy\|^2-\tfrac12\|\vecy'\|^2)^2+\gamma^2} \to \delta(\tfrac12\|\vecy\|^2-\tfrac12\|\vecy'\|^2),
\end{equation}
as $\gamma\to 0$.

\begin{prop}\label{thm:multieq}
Let $a,b\in\scrS(\T(\RR^d))$, $W\in\scrS(\RR^d)$, $t>0$, $|\lambda| < R$ with $R$ as in \eqref{Rdef}, and $\Re\gamma\geq 0$. Then
\begin{equation}
\begin{split}
& \lim_{h=r\to 0}   \sum_{n=1}^\infty  (2 \pi \i \lambda h^{-2})^n \; \scrJ_n^\gamma(t h r^{1-d})  \\
& = \sum_{n=1}^\infty (2 \pi \i \lambda)^n \sum_{k=1}^n \sum_{(\ell,\F)\in\Omega(n,k)}  (-1)^\ell  \int_{(\RR^d)^k} 
  \scrW(\iota_{\F}(\vecy_1,\ldots,\vecy_k))
\\
& \times  \, \widehat\scrA_{t,\ell,\F}(\vecnull,\vecy_1,\ldots,\vecy_{k})  
  D_{\ell,\F}^\gamma(\vecy_1,\ldots,\vecy_k)\,\d\vecy_1\cdots\d\vecy_k .
\end{split}
\end{equation}
\end{prop}

\begin{proof}
In view of the uniform convergence of the series in $n$ (Proposition \ref{cor:Hseries}), it is sufficient to establish convergence term by term. Now
\begin{equation} \label{decayrecapcap}
 \begin{split}
&  h^{2n}\; \scrJ_{\ell,n}^\gamma(t h r^{1-d}) \\
& =  \sum_{k=1}^n \sum_{\F \in \Pi_\circ(n,k)} \int_{\Sigma^\perp}  
\bigg( \int_{\RR^{dk}}\scrW(\iota_{\F}(\vecy_1,\ldots,\vecy_k))  \\
& \times
 \e\bigg(\sum_{\substack{i=1\\ i\neq l}}^k  \big[\theta_i (\tfrac12\|\vecy_i\|^2-\tfrac12\|\vecy_l\|^2) + \i |\theta_i| \gamma\big] \bigg)  
\widehat\scrA_{t,\ell,\F}(r^d \vectheta,\vecy_1,\ldots,\vecy_{k}) \,\d \vecy_1 \cdots \d \vecy_k \bigg)\,   \d^\perp\vectheta.
\end{split}
\end{equation}
Due to the uniform decay $\prod_{i\neq l} \langle \theta_i \rangle^{-d/2}$ guaranteed by Proposition \ref{prop91}, the outer integral converges uniformly in $r>0$ (and $\Re\gamma\geq 0$), and we can therefore take the limit $r\to 0$ inside. Relations \eqref{lem6.1proof1003} and \eqref{eq716} tell us that the only-non zero terms come from the marked partitions $(\ell,\F)\in\Omega(n,k)$ and for $\vectheta\in\scrC^\perp_{\ell,\F}$.
\end{proof}

\section{The collision series} \label{secOrganisationoftheseries}
 The main result of this section is Proposition \ref{prop:tmatexpansion} which specialises Proposition \ref{thm:multieq} to the case of $\gamma=0$. Let us define
\begin{equation}
\scrT_{n}(\vecy_0,\ldots,\vecy_n) = (-2 \pi \i)^n \prod_{j=0}^{n-1} T(\vecy_j,\vecy_{j+1})  , \qquad \scrT_{0}(\vecy_0)=1, 
\end{equation}
\begin{equation}
\begin{split}
\scrT_{\ell,n}(\vecy_0,\ldots,\vecy_n) & = \scrT_{\ell}(\vecy_0,\ldots,\vecy_\ell) \overline{\scrT_{n-\ell}}(\vecy_n,\ldots,\vecy_{\ell}) \\
& = (2 \pi \i)^n (-1)^\ell  \prod_{j=0}^{\ell-1} T(\vecy_j,\vecy_{j+1}) \prod_{j=\ell}^{n-1} T^\dagger(\vecy_{j},\vecy_{j+1}) ,
\end{split}
\end{equation}
and for $\vecm \in \ZZ_{> 0}^{n}$,
\begin{equation}
\scrT_{\ell,n,\vecm} ( \vecy_0,\ldots,\vecy_n) = (2 \pi \i)^n (-1)^\ell \prod_{j=0}^{\ell-1}  T_{m_j}(\vecy_j,\vecy_{j+1})\prod_{j=\ell}^{n-1}  T_{m_j}^\dagger(\vecy_{j},\vecy_{j+1}) .
\end{equation}
Note that
\begin{equation} \label{summingscrT}
\sum_{\vecm \in \ZZ_{> 0}^{n}}\lambda^{m_1+\ldots+m_{n}}  \scrT_{\ell,n,\vecm}(\vecy_0,\ldots,\vecy_n) = \scrT_{\ell,n}(\vecy_0,\ldots,\vecy_n).
\end{equation}
Define furthermore
\begin{multline}\label{defRkt}
\scrR^{(k)}(t) = \sum_{n=k-1}^\infty\sum_{(\ell,\F)\in\widehat\Omega(n,k) } \int_{(\RR^d)^{k}} \widehat\scrA_{t,\ell,\F}(\vecnull,\vecy_1,\ldots,\vecy_{k})\\
\times \scrT_{\ell,n}(\iota_{\F}(\vecy_1,\ldots,\vecy_k)) \, \omega_k(\vecy_1,\ldots,\vecy_k)\, \d\vecy_1\cdots\d\vecy_k.
\end{multline}
Note that in \eqref{defRkt} for $k>1$ only terms with $n\geq k$ contribute. 

\begin{prop}\label{prop:tmatexpansion}
Let $a,b\in\scrS(\T(\RR^d))$, $W\in\scrS(\RR^d)$, $t>0$, $|\lambda| < R$ with $R$ as in \eqref{Rdef}. Then
\begin{equation}\label{thm:multieq234} 
 \lim_{h=r\to 0} \sum_{n=0}^\infty  (2 \pi \i \lambda h^{-2})^n \; \scrJ_n (t h r^{1-d}) 
=  \sum_{k=1}^\infty \scrR^{(k)}(t) .
\end{equation}
\end{prop}

\begin{proof} We begin from the result of Proposition \ref{thm:multieq}. Let $(\ell,\F) \in \Omega(n,k)$ and let $(\ell',\F') \in \widehat\Omega(n',k')$ be the corresponding reduced marked partition. Order the blocks of $(\ell,\F)$ such that the following three conditions hold:
\begin{enumerate}
\item $|F_i|>1$ for $i=1,\dots,k'-1$,
\item $\ell \in F_{k'}$,
\item $|F_i|=1$ for $i=k'+1,\dots,k$.
\end{enumerate}

Define $i_{k'+1},\dots,i_n$ so that $F_j = \{i_j\}$ for $j=k'+1,\dots,n$. We can then write
\begin{multline} \label{exdelta23456}
\delta(\Delta_{\ell,\F}\vecu)= 
\bigg(\prod_{i=1}^{k'-1} \delta\bigg(\sum_{j=0}^{\ell-1}  u_j \mathbb{1}(j\in F_i)-\sum_{j=\ell+1}^{n} u_j  \mathbb{1}(j\in F_i) \bigg) \bigg) \bigg( \prod_{j=k'+1}^n \delta(u_{i_j}) \bigg).
\end{multline}
By first integrating over $u_{i_{k'+1}},\dots,u_{i_{n}}$, and then relabelling the remaining $u_i$ variables with the indices $0,\dots,n'$ (preserving their order) we obtain
\begin{multline}
\int_{\Box_{\ell,n}(t)} \scrA_{\ell,n}(\iota_{\F}(\vecy_1,\ldots,\vecy_k),\vecu ) \, \delta(\Delta_{\ell,\F} \vecu) \, \d^\perp \vecu \\
= \int_{\Box_{\ell',n'}(t)} \scrA_{\ell',n'}(\iota_{\F'}(\vecy_1,\ldots,\vecy_{k'}), \vecu) \, \delta(\Delta_{\ell',\F'} \vecu) \, \d^\perp \vecu,
\end{multline}
or in other words
\begin{equation}
\widehat\scrA_{t,\ell,\F}(\vecnull,\vecy_1,\ldots,\vecy_{k}) = \widehat\scrA_{t,\ell',\F'}(\vecnull,\vecy_1,\ldots,\vecy_{k'}).
\end{equation}
Furthermore, by \eqref{ddef} every non-singleton in $\F$ contains indices both to the left and right of $\ell$, so every such term yields a delta function and we see that the distribution $D_{\ell,\F}$ is given by
\begin{equation}
D_{\ell, \F}(\vecy_1,\ldots,\vecy_k) 
= \bigg(\prod_{i=1}^{k'-1} \delta(\tfrac12 \|\vecy_i\|^2-\tfrac12 \|\vecy_{k'}\|^2)\bigg)\bigg(\prod_{ i = k'+1}^k d_{\ell,\F,i}(\vecy_{k'},\vecy_i)\bigg).  
\end{equation}
This allows us to replace instances of $\|\vecy_{k'}\|^2$ with $\|\vecy_j\|^2$ for any $j = 1,\ldots,k'-1$. Let   
$$0=j_1<\cdots<j_\mu<\ell = j_{\mu+1}< \cdots < j_{\mu+\nu+1} = n$$
be the list of elements of $\{0,\ldots,n\}$ that lie in non-singleton blocks. 
For $i=1,\dots,\mu+\nu$ we define $m_i = j_{i+1} - j_i -1$ as the number of singletons between $j_{i+1}$ and $j_i$, and set $M_i=m_1+\ldots+m_i$. Thus $M_{\mu+\nu}=M$ is the total number of singletons in $\F$. From the definition \eqref{ddef} one can see that
\begin{equation}
\prod_{ i = k'+1}^k d_{\ell,\F,i}^\gamma(\vecy_{k'},\vecy_i) 
 =\frac{(-1)^{M_\mu}}{(2 \pi \i)^M} \bigg( \prod_{i=k'+1}^{k'+M_\mu} g^\gamma(\vecy_{k'},\vecy_i) \bigg)\bigg( \prod_{i=k'+M_\mu+1}^{k'+M} \overline{ g^\gamma(\vecy_{k'},\vecy_i)}  \bigg).
\end{equation}
Combining the above with the definition of the $T$-matrix allows us to obtain
\begin{multline} \label{1112}
\int_{\RR^{Md}} \scrW(\iota_{\F}(\vecy_1,\ldots,\vecy_k) ) D_{\ell, \F}(\vecy_1,\ldots,\vecy_k) \prod_{ i = k'+1}^k \d \vecy_i\\
= \frac{(-1)^{M_\mu}}{(2\pi \i)^M} \int_{\RR^{Md}} \scrW(\iota_{\F}(\vecy_1,\ldots,\vecy_k) )  \omega_{k'}(\vecy_1,\ldots,\vecy_{k'}) \\
\hfill \times \bigg( \prod_{i=k'+1}^{k'+M_\mu} g(\vecy_{k'},\vecy_i) \bigg) \bigg( \prod_{i=k'+M_\mu+1}^{k'+M} \overline{ g(\vecy_{k'},\vecy_i)}  \bigg)  \, \prod_{ i = k'+1}^k \d \vecy_i.
\end{multline}
Due to the Dirac delta functions appearing in $\omega_{k'}$, for $i \in [k'+1, k'+M_\mu]$ we can replace $g(\vecy_{k'},\vecy_i)$ with $g(\vecy_{i_-},\vecy_i)$ where $i_-$ is defined to be the largest non-singleton element smaller than $i$. Similarly, for $i \in [k'+M_\mu + 1, k'+M]$ we replace $\overline{g(\vecy_{k'},\vecy_i)}$ with $\overline{g(\vecy_{i_+},\vecy_i)}$ where $i_+$ is defined to be the smallest non-singleton element larger than $i$. This allows us to conclude that \eqref{1112} is equal to
\begin{multline}
\frac{(-1)^{M_\mu} }{(2\pi \i)^M} \frac{(-1)^{\ell'}}{(2\pi \i)^{n'}} \scrT_{\ell',n',(m_1+1,\dots,m_M+1)}(\iota_{\F'}(\vecy_1,\ldots,\vecy_{k'}) ) \, \omega_{k'}(\vecy_1,\dots,\vecy_{k'}).
\end{multline} 
Now use the fact that $M+n' = n$ and $M_\mu +\ell' = \ell$ we can write
\begin{equation}
\frac{(-1)^{M_\mu} }{(2\pi \i)^M} \frac{(-1)^{\ell'}}{(2\pi \i)^{n'}}=\frac{(-1)^\ell }{(2\pi \i)^n}
\end{equation}
and so
\begin{equation}
\begin{split}
& \lim_{h=r\to 0} \sum_{n=0}^\infty (2 \pi \i \lambda h^{-2})^n \; \scrI_n(t h r^{1-d}) \\
& = \sum_{n=0}^\infty \sum_{k=1}^{n} \sum_{(\ell,\F)\in\widehat\Omega(n,k)} \sum_{\vecm \in \ZZ_{\geq 0}^n}\lambda^{n+m_1+\ldots+m_n}   \int_{(\RR^d)^{k}} \widehat\scrA_{t,\ell,\F}(\vecnull,\vecy_1,\ldots,\vecy_{k}) \\
& \times  \scrT_{\ell,n,(m_1+1,\dots,m_M+1)}(\iota_{\F}(\vecy_1,\ldots,\vecy_k)) \, \omega_{k}(\vecy_1,\dots,\vecy_k) \;\d \vecy_1 \cdots \d \vecy_k
\end{split}
\end{equation}
where on the right hand side we now write $k$, $\ell$, $n$ instead of $k'$, $\ell'$, $n'$.
The result then follows by changing the summation variables $m_j \to m_j -1$ and using \eqref{summingscrT}.
\end{proof}

\section{The limit process}\label{secpositivity}
 In this section we derive explicit formulas for $\scrR^{(k)}(t)$, assuming throughout that $k\geq 1$. These show in particular that $\scrR^{(k)}(t)$ can be expressed as the $(k-1)$-collision term with a real and non-negative kernel. The main results are equations \eqref{scrRkdll} and \eqref{scrRkofflm} which together yield the formula \eqref{collserieskOUR}, as well as equations \eqref{rhokdfinalG2} and \eqref{rhokndfinalG2} which respectively yield the expressions \eqref{rho112} and \eqref{rho122} for $\rho_{11}^{(2)}$ and $\rho_{12}^{(2)}$ in terms of Bessel functions.

Let us write $\scrR^{(k)}(t) = \scrR_{\d}^{(k)}(t) + \scrR_{\off}^{(k)}(t)$, where $\scrR_{\d}^{(k)}(t)$, $\scrR_{\off}^{(k)}(t)$ are as in the definition of $\scrR^{(k)}(t)$ \eqref{defRkt}, with $\widehat\Omega(n,k)$ replaced by $\widehat\Omega_\d(n,k)$ and $\widehat\Omega_\off(n,k)$, respectively. 

\subsection{Diagonal Terms}

This is the case $0,\ell\in F_1$ and so
\begin{multline} \label{scrRkd}
\scrR_{\d}^{(k)}(t) = \frac{1}{(k-1)!} \int_{\RR_{\geq 0}^k} 
\int_{(\RR^d)^{k+1}}   a\bigg( \vecx - \sum_{j=1}^k u_j (\vecy_j-\vecy_1),\vecy_1 \bigg) b( \vecx, \vecy_1)   \\
\times   \rho_{11}^{(k)}(\vecu,\vecy_1,\dots,\vecy_k) \,
\delta(u_1+\cdots+u_k-t) \, 
 \d \vecx \, \d\vecy_1\cdots\d\vecy_k\, \d \vecu
\end{multline}
with the function
\begin{multline} 
\rho_{11}^{(k)}(\vecu,\vecy_1,\dots,\vecy_k)  \\ = \omega_k(\vecy_1,\ldots,\vecy_k)
\sum_{n=k-1}^\infty\sum_{(\ell,\Fura)\in\underrightarrow{\widehat\Omega}_{\d}(n,k) } \scrT_{\ell,n}(\iota_{\Fura}(\vecy_1,\ldots,\vecy_k)) \prod_{i=1}^k \frac{u_i^{\mu_{i}+\nu_{i}} }{\mu_{i}! \nu_{i}!}.
\end{multline}
We have here used the symmetry of the integrand under permutation of the indices of the $\vecy_i$ with $i\geq 2$, and taken an average over all ordered partitions in $\underrightarrow{\widehat\Omega}_{\d}(n,k)$ rather than the original sum over the unordered partitions in $\widehat\Omega_{\d}(n,k)$.

We now apply the bijection $(\ell,\Fura)\mapsto (\Fura^+,\Fura^-)$ in \eqref{A9} (Lemma \ref{lemA2}), which, together with the relation
\begin{equation}
\scrT_{\ell,n}(\iota_{\Fura}(\vecy_1,\dots,\vecy_k)) = \scrT_n(\iota_{\Fura^+}(\vecy_1,\dots,\vecy_k) ) \, \overline{\scrT_n(\iota_{\Fura^-}(\vecy_1,\dots,\vecy_k) )} ,
\end{equation}
yields
\begin{multline} 
\rho_{11}^{(k)}(\vecu,\vecy_1,\dots,\vecy_k) \\
 = \omega_k(\vecy_1,\ldots,\vecy_k)
 \bigg| \sum_{n=k-1}^\infty\sum_{\Fura\in \underrightarrow\Pi_\circ(n,k)} \scrT_{n}(\iota_{\Fura}(\vecy_1,\ldots,\vecy_k)) \prod_{i=1}^k \frac{u_i^{|F_i|-1}}{(|F_i|-1)!} \bigg|^2.
\end{multline}
Next we use the bijection $\Fura \mapsto (\Fura',\vecm)$ \eqref{A7} to non-consecutive ordered partitions $\underrightarrow\Pi_\circ^\nc(n,k)$,
which yields
\begin{multline}
\sum_{n=k-1}^\infty\sum_{\Fura\in \underrightarrow\Pi_\circ(n,k)} \scrT_{n}(\iota_{\Fura}(\vecy_1,\ldots,\vecy_k)) \prod_{i=1}^k \frac{u_i^{|F_i|-1}}{(|F_i|-1)!} \\
= \sum_{n=k-1}^\infty \sum_{\Fura\in \underrightarrow\Pi_\circ^\nc(n,k)} 
\scrT_{n}(\iota_{\Fura}(\vecy_1,\ldots,\vecy_k)) 
\sum_{\vecm\in\ZZ^n} \bigg( \prod_{i=1}^k \frac{u_i^{|F_i|+|F_i|_{\vecm} -1} (-2\pi\i T(\vecy_i,\vecy_i))^{|F_i|_{\vecm} }}{(|F_i|+|F_i|_{\vecm} -1)!} \bigg), 
\end{multline}
with $|F_i|_{\vecm} = \sum_{j\in F_i} m_j$.

We identify $\vecm\in\ZZ_{\geq 0}^n$ with $(\vecm^{(1)},\ldots,\vecm^{(k)})\in\ZZ_{\geq 0}^{|F_1|}\times\cdots\times\ZZ_{\geq 0}^{|F_k|}$.
For each $i$ we then use the identity
\begin{equation}
\sum_{m_1,\dots,m_p=0}^\infty f(m_1+\dots+m_p) = \sum_{\mu=0}^\infty \frac{(\mu+p-1)!}{\mu! (p-1)!} f(\mu)
\end{equation}
with $p=|F_i|$ to write
\begin{equation}
\begin{split}
& \sum_{\vecm\in\ZZ^n} \bigg( \prod_{i=1}^k \frac{u_i^{|F_i|+|F_i|_{\vecm} -1} (-2\pi\i T(\vecy_i,\vecy_i))^{|F_i|_{\vecm} }}{(|F_i|+|F_i|_{\vecm} -1)!} \bigg) \\
& =
\prod_{i=1}^k \bigg( \sum_{\mu=0}^\infty \frac{u_i^{|F_i|+\mu -1} (-2\pi\i T(\vecy_i,\vecy_i))^{\mu}}{\mu! (|F_i|-1)!} \bigg).
\end{split}
\end{equation}
\begin{equation}
=\prod_{i=1}^k \bigg( \frac{u_i^{|F_i|-1}\e^{-2\pi\i u_i T(\vecy_i,\vecy_i)}}{(|F_i|-1)!}  \bigg)
\end{equation}
This yields in view of the optical theorem \eqref{opticaltheorem},
\begin{multline} \label{rhokdfinal}
\rho_{11}^{(k)}(\vecu,\vecy_1,\dots,\vecy_k)  = \omega_k(\vecy_1,\ldots,\vecy_k) \prod_{i=1}^{k} \e^{- u_i \Sigma_\tot(\vecy_i)} \\ \times
\bigg| \sum_{n=k-1}^\infty\sum_{\Fura\in \underrightarrow\Pi_\circ^\nc(n,k)} \scrT_{n}(\iota_{\Fura}(\vecy_1,\ldots,\vecy_k)) \prod_{i=1}^k \frac{u_i^{|F_i|-1}}{(|F_i|-1)!} \bigg|^2.
\end{multline}
When $k=1$ we are summing over partitions into $1$ block. Note that $\underrightarrow\Pi_\circ^\nc(n,1)$ is empty unless $n=0$, in which case $\underrightarrow\Pi_\circ^\nc(0,1)$ contains only the partition $[\{0\}]$. Formula \eqref{rhokdfinal} therefore yields
\begin{equation}
\rho_{11}^{(1)}(u_1,\vecy_1) = \e^{-u_1 \Sigma_\tot(\vecy_1)}.
\end{equation}

When $k\geq 2$, using the results in Appendix \ref{subsec:GaP}, we can write \eqref{rhokdfinal} as
\begin{equation}  \label{rhokdfinalG} 
\rho_{11}^{(k)}(\vecu,\vecy_1,\dots,\vecy_k)  = \omega_k(\vecy_1,\ldots,\vecy_k) \big| g_{11}^{(k)}(\vecu,\vecy_1,\dots,\vecy_k) \big|^2 \prod_{i=1}^{k} \e^{- u_i \Sigma_\tot(\vecy_i)} .
\end{equation}
Here $g_{11}^{(k)}$ is the $11$ coefficient of the $k\times k$ matrix-valued function (recall \eqref{rhokdfinalGdef0}),
\begin{equation} \label{rhokdfinalGdef}
\GG^{(k)}(\vecu,\vecy_1,\dots,\vecy_k) 
= \frac{1}{(2\pi\i)^k} \ointccw\cdots \ointccw \big( \DD(\vecz)- \WW\big)^{-1}  \exp(\vecu\cdot\vecz)\, \d z_1 \cdots \d z_k ,
\end{equation}
with the diagonal matrix $\DD(\vecz)=\diag(z_1,\ldots,z_k)$ and the matrix $\WW=\WW(\vecy_1,\dots,\vecy_k)$ with coefficients
\begin{equation}
w_{ij} = 
\begin{cases}
0 & (i=j) \\
-2\pi\i T(\vecy_i,\vecy_j) & (i\neq j) .
\end{cases}
\end{equation}
If we extend the definition of \eqref{rhokdfinalG} to
\begin{equation}  \label{rhokdfinalGmn} 
\rho_{\ell m}^{(k)}(\vecu,\vecy_1,\dots,\vecy_k)  = \omega_k(\vecy_1,\ldots,\vecy_k) \big| g_{\ell m}^{(k)}(\vecu,\vecy_1,\dots,\vecy_k) \big|^2 \prod_{i=1}^{k} \e^{- u_i \Sigma_\tot(\vecy_i)} ,
\end{equation}
we see that, by symmetry under the permutation of indices, the function \eqref{scrRkd} can be written as
\begin{multline} \label{scrRkdll}
\scrR_{\d}^{(k)}(t) = \frac{1}{k!} \sum_{\ell=1}^k \int_{\RR_{\geq 0}^k} 
\int_{(\RR^d)^{k+1}}   a\bigg( \vecx - \sum_{j=1}^k u_j (\vecy_j-\vecy_1),\vecy_\ell \bigg) b( \vecx, \vecy_\ell)   \\
\times   \rho_{\ell\ell}^{(k)}(\vecu,\vecy_1,\dots,\vecy_k) \,
\delta(u_1+\cdots+u_k-t) \, 
 \d \vecx \, \d\vecy_1\cdots\d\vecy_k\, \d \vecu ,
\end{multline}
which yields the diagonal part of the expression \eqref{collserieskOUR}.

In the case $k=2$, Eq.~\eqref{eqk2A1} yields an explicit formula in terms of the $J_1$-Bessel function (assuming here $u_1,u_2>0$),
\begin{equation}\label{eqk2A1J}
g_{11}^{(2)}(u_1,u_2) = -2\pi \sqrt\frac{u_1}{u_2}\, (T(\vecy_1,\vecy_2) T(\vecy_2,\vecy_1))^{1/2} \,J_1(4\pi (u_1 u_2 T(\vecy_1,\vecy_2) T(\vecy_2,\vecy_1))^{1/2}) .
\end{equation}
We conclude
\begin{multline}  \label{rhokdfinalG2} 
\rho_{11}^{(2)}(u_1,u_2,\vecy_1,\vecy_2)  = 4\pi^2 \, |T(\vecy_1,\vecy_2) T(\vecy_2,\vecy_1)|\, \omega_2(\vecy_1,\vecy_2)\, \e^{- u_1 \Sigma_\tot(\vecy_1)- u_2 \Sigma_\tot(\vecy_2)} \\
\times \frac{u_1}{u_2}\, \big| J_1\big(4\pi (u_1 u_2 T(\vecy_1,\vecy_2) T(\vecy_2,\vecy_1))^{1/2} \big) \big|^2 .
\end{multline}
This formula can also be obtained directly from the combinatorial expression \eqref{rhokdfinal}: note that for $n$ even the only $\Fura\in \underrightarrow\Pi_\circ^\nc(n,2)$ is
\begin{equation}
\Fura = [ \{0,2,4,\ldots,n-2,n\}, \{1,3,5,\ldots,n-1\}];
\end{equation}
and for $n$ odd $\underrightarrow\Pi_\circ^\nc(n,2)$ is empty.
Hence
\begin{multline} \label{rhokdfinal201}
\rho_{11}^{(2)}(u_1,u_2,\vecy_1,\vecy_2)  = \omega_2(\vecy_1,\vecy_2)\,\e^{- u_1 \Sigma_\tot(\vecy_1)- u_2 \Sigma_\tot(\vecy_2)}  \\ \times
\bigg| \sum_{m=1}^\infty (-4\pi^2 T(\vecy_1,\vecy_2) T(\vecy_2,\vecy_1))^m \frac{u_1^{m} u_2^{m-1}}{m! (m-1)!} \bigg|^2.
\end{multline}
By shifting the summation index this can be written
\begin{multline} 
\rho_{11}^{(2)}(u_1,u_2,\vecy_1,\vecy_2)  = \omega_2(\vecy_1,\vecy_2)\,\e^{- u_1 \Sigma_\tot(\vecy_1)- u_2 \Sigma_\tot(\vecy_2)}  \\ \times
\bigg|  4\pi^2 T(\vecy_1,\vecy_2) T(\vecy_2,\vecy_1)  u_1 \sum_{m=0}^\infty (-4\pi^2 T(\vecy_1,\vecy_2) T(\vecy_2,\vecy_1))^m \frac{u_1^{m} u_2^{m}}{m! (m+1)!} \bigg|^2.
\end{multline}
Equation \eqref{rhokdfinalG2} then follows from the series representation of the Bessel function
\begin{equation}\label{J_nSeries}
J_n(z) = (\tfrac12 z)^n \sum_{k=0}^\infty \frac{(-\tfrac14 z^2)^k}{k! (k+n)!}.
\end{equation}

\subsection{Off-diagonal Terms}

Here $0\in F_1$ and $\ell\in F_k$, and thus
\begin{multline} \label{scrRkoff}
\scrR_{\off}^{(k)}(t) = \frac{1}{(k-2)!} \int_{\RR_{\geq 0}^k} 
\int_{(\RR^d)^{k+1}}   a\bigg( \vecx - \sum_{j=1}^k u_j (\vecy_j-\vecy_k),\vecy_k \bigg) b( \vecx, \vecy_1)  \\
\times  \rho_{1k}^{(k)}(\vecu,\vecy_1,\dots,\vecy_k) 
\,  \delta(u_1+\cdots+u_k-t) \, 
 \d \vecx \, \d\vecy_1\cdots\d\vecy_k\, \d \vecu
\end{multline}
with the function
\begin{multline} 
\rho_{1k}^{(k)}(\vecu,\vecy_1,\dots,\vecy_k)  = \\
\omega_k(\vecy_1,\ldots,\vecy_k)
\sum_{n=k-1}^\infty\sum_{(\ell,\Fura)\in\underrightarrow{\widehat\Omega}_{\off}(n,k) } \scrT_{\ell,n}(\iota_{\Fura}(\vecy_1,\ldots,\vecy_k)) \prod_{i=1}^k \frac{u_i^{\mu_{i}+\nu_{i}} }{\mu_{i}! \nu_{i}!}.
\end{multline}
The argument is identical to the diagonal case, and we obtain
\begin{multline} \label{rhokOFFfinal}
\rho_{1k}^{(k)}(\vecu,\vecy_1,\dots,\vecy_k)  = \omega_k(\vecy_1,\ldots,\vecy_k) \prod_{i=1}^{k} \e^{- u_i \Sigma_\tot(\vecy_i)} \\ \times
\bigg| \sum_{n=k-1}^\infty\sum_{\Fura\in \underrightarrow\Pi_\baro^\nc(n,k)} \scrT_{n}(\iota_{\Fura}(\vecy_1,\ldots,\vecy_k)) \prod_{i=1}^k \frac{u_i^{|F_i|-1}}{(|F_i|-1)!} \bigg|^2.
\end{multline}
In the case $k=1$, there is no off-diagonal term as $\underrightarrow\Pi_\baro^\nc(n,1)$ is empty for every $n$.

Furthermore, for $k\geq 2$ (again using the results in Appendix \ref{subsec:GaP}),
\begin{equation}  \label{rhokofffinalG}
\rho_{1k}^{(k)}(\vecu,\vecy_1,\dots,\vecy_k)  = \omega_k(\vecy_1,\ldots,\vecy_k) \big| g_{1k}^{(k)}(\vecu,\vecy_1,\dots,\vecy_k) \big|^2 \prod_{i=1}^{k} \e^{- u_i \Sigma_\tot(\vecy_i)} ,
\end{equation}
with $\GG^{(k)}(\vecu,\vecy_1,\dots,\vecy_k)$ as in \eqref{rhokdfinalGdef}.
Again, by symmetry under the permutation of indices, \eqref{scrRkoff} can be expressed as
\begin{multline} \label{scrRkofflm}
\scrR_{\off}^{(k)}(t) = \frac{1}{k!} \sum_{\substack{\ell,m=1\\ \ell\neq m}}^k \int_{\RR_{\geq 0}^k} 
\int_{(\RR^d)^{k+1}}   a\bigg( \vecx - \sum_{j=1}^k u_j (\vecy_j-\vecy_1),\vecy_m \bigg) b( \vecx, \vecy_\ell)   \\
\times   \rho_{\ell m}^{(k)}(\vecu,\vecy_1,\dots,\vecy_k) \,
\delta(u_1+\cdots+u_k-t) \, 
 \d \vecx \, \d\vecy_1\cdots\d\vecy_k\, \d \vecu ,
\end{multline}
which yields the off-diagonal part of \eqref{collserieskOUR}.

For $k=2$, Eq.~\eqref{eqk2A2} yields
\begin{equation}\label{eqk2A2J}
g_{12}^{(2)}(u_1,u_2) = -2\pi\i T(\vecy_1,\vecy_2) \, J_0(4\pi (u_1 u_2 T(\vecy_1,\vecy_2) T(\vecy_2,\vecy_1))^{1/2} ) .
\end{equation}
We conclude
\begin{multline}   \label{rhokndfinalG2}
\rho_{12}^{(2)}(u_1,u_2,\vecy_1,\vecy_2)  = 4\pi^2 \, |T(\vecy_1,\vecy_2)|^2\, \omega_2(\vecy_1,\vecy_2)\, \e^{- u_1 \Sigma_\tot(\vecy_1)- u_2 \Sigma_\tot(\vecy_2)} \\
\times \big| J_0\big(4\pi (u_1 u_2 T(\vecy_1,\vecy_2) T(\vecy_2,\vecy_1))^{1/2} \big) \big|^2 .
\end{multline}

Let us derive this formula also directly from the combinatorial expression \eqref{rhokOFFfinal}: for $n$ odd the only element of $\underrightarrow\Pi_\baro^\nc(n,2)$ is
\begin{equation}
\Fura = [\{0,2,4,\ldots,n-1\}, \{1,3,5,\ldots,n\} ];
\end{equation}
and for $n$ even $\underrightarrow\Pi_\baro^\nc(n,2)$ is empty. Hence
\begin{multline} \label{rhok2nd2}
\rho_{12}^{(2)}(u_1,u_2,\vecy_1,\vecy_2) =\omega_2(\vecy_1,\vecy_2)\, \e^{- u_1 \Sigma_\tot(\vecy_1)- u_2 \Sigma_\tot(\vecy_2)}\,\\
\times \bigg| \sum_{m=0}^\infty \frac{u_1^{m}u_2^{m} }{(m!)^2} (-2\pi \i)^{2m+1} T(\vecy_1,\vecy_2) (T(\vecy_1,\vecy_2) T(\vecy_2,\vecy_1))^{m} \bigg|^2 .
\end{multline}
This can be written
\begin{multline} \label{rhok2nd3}
\rho_{12}^{(2)}(u_1,u_2,\vecy_1,\vecy_2) = 4 \pi^2 |T(\vecy_1,\vecy_2)|^2 \,\omega_2(\vecy_1,\vecy_2)\, \e^{- u_1 \Sigma_\tot(\vecy_1)- u_2 \Sigma_\tot(\vecy_2)}\\
\times \bigg| \sum_{m=0}^\infty \frac{1 }{(m!)^2} (-1)^{m}  \left(2 \pi\sqrt{u_1 u_2 \, T(\vecy_1,\vecy_2) T(\vecy_2,\vecy_1)}\right)^{2m} \bigg|^2 .
\end{multline}
Identifying the summation as a $J_0$ Bessel function yields the result.

\section{Discussion}\label{secDiscussion}

The main conclusion of this work is that quantum transport in a periodic potential converges, in the microlocal Boltzmann-Grad limit, to a limiting random flight process. Unlike in the random setting, there is a positive probability that a path of the limit process revisits the same momentum several times. This is ultimately a consequence of the Floquet-Bloch reduction to discrete Hilbert spaces.
The only hypothesis, Assumption \ref{hyp0}, in our derivation is that Bloch momenta have asymptotically the same fine-scale distribution as a Poisson point process. This assumption can be viewed as a phase-space extension of the Berry-Tabor conjecture in quantum chaos \cite{BerryTabor,Marklof00}, which to-date has been confirmed only in special cases \cite{Bleher95,EMM,Marklof02,Marklof03,MargulisMohammadi,Sarnak96,VanderKam}. In the setting discussed in this paper, present techniques permit a rigorous analysis up to second order perturbation theory which, perhaps surprisingly, is consistent with the linear Boltzmann equation as well as our limit process. Thus extending the perturbative analysis to higher order terms unconditionally is an important open challenge. This would require the rigorous understanding of higher-order correlation functions for lattice point statistics, and we refer the reader to \cite{VanderKam} for the best current results in this direction.

It follows from standard invariance principles for Markov processes that for large times the solution of the linear Boltzmann equation is governed by Brownian motion with the standard diffusive mean-square displacement (i.e., linear in time) \cite{Spohn78,ErdosSalmhoferYau08}. Therefore, the work of Eng and Erd\"os \cite{EngErdos} for random potentials implies convergence to Brownian motion, if we first take the Boltzmann-Grad and then the diffusive limit. (Note that Erd\"os, Salmhofer and Yau \cite{ErdosSalmhoferYau07,ErdosSalmhoferYau08} have established convergence to Brownian motion in long-time/weak-coupling scaling limits directly, i.e., without first taking the weak-coupling limit to obtain the linear Boltzmann equation as in \cite{ErdosYau}.) An immediate challenge is thus to understand the diffusive nature of the random flight process derived in the present paper. Recall that in the classical setting the Boltzmann-Grad limit of the periodic Lorentz gas does not satisfy the linear Boltzmann equation \cite{Caglioti10,partII}, and we have superdiffusion with a $t\log t$ mean-square displacement \cite{MT2016}.
A further challenge is to expand our current understanding to more singular single-site potentials (such as hard core and/or long-range potentials) and to include background electromagnetic fields.

\begin{appendix}

\section{Partitions, diagrams and graphs} \label{sec:partitions}

\subsection{Set partitions}   \label{subsec:partitions}

A {\em set partition} $\F=[F_1,\ldots,F_k]$ of the finite set $\{0,\ldots,n\}$ is a decomposition into disjoint and non-empty subsets $F_1,\ldots,F_k$. The order in which we list the $F_i$ is not relevant (we will discuss ordered partitons further down). We call $F_i$ a {\em block} of $\F$, and denote by $k=|\F|\leq n+1$ the number of blocks. We furthermore define $\nu(\F)=|F_1\cup\cdots\cup F_k|-1$. We denote by $\Pi(n)$ the collection of all set partitions $\F$ with $\nu(\F)=n$, and by $\Pi(n,k)$ the collection of $\F\in\Pi(n)$ with $|\F|=k$. We write $\F\preceq\underline G$, if every subset of $\underline G$ is a subset of unions of subsets of $\F$. This defines a partial ordering on $\Pi(n)$. The minimal and maximal elements of $\Pi(n)$ are $\underline O=[\{0\},\ldots,\{n\}]$ and $\underline N=[\{0,\ldots,n\}]$, respectively. 

We further denote by $\Pi_\circ(n,k)\subset\Pi(n,k)$ the sub-collection of all set partitions where $0$ and $n$ are in the same block ($F_1$, say) and by $\Pi_\circ^{\nc}(n,k)\subset\Pi_\circ(n,k)$ the sub-collection of \emph{non-consecutive} partitions where each subset $F_i$ does not contain consecutive indices; that is $|j_1-j_2|\neq 1$ for all $j_1,j_2\in F_i$.

\subsection{Marked and reduced set partitions} \label{subsec:Marked}

Given a set partition $\F\in\Pi(n)$ and integer $\ell\in \{0,\ldots, n\}$ we call $(\ell,\F)$ the corresponding {\em marked partition}. 
Let $\Omega(n) \subset \{0,\ldots, n\}\times \Pi_\circ(n)$ denote the set of marked partitions $(\ell,\F)$ such that every block $F_i$ not containing $\ell$ either (i) is a singleton or (ii) contains at least one number strictly less than $\ell$ and at least one number strictly greater than $\ell$. Let furthermore $\Omega(n,k) \subset \Omega(n) $ denote the subset where $\F$ has $k$ blocks.

These marked partitions are more easily understood diagrammatically. Given $(\ell,\F)\in\Omega(n,k)$ first draw $n+1$ circles in a horizontal line representing the indices $0,1,\ldots,n$; then fill in the circle corresponding to the index $\ell$; finally, connect indices with lines beneath if and only if they lie in the same block. An example diagram for a typical partition can be seen in Figure \ref{fig:markedpartition}.
\begin{figure}[ht!] 
\begin{align*}
(\ell,\F) &= \left( 6, \left[ \{0,5,11\}, \{1\},\{2,7\},\{3,8,10\},\{4\},\{6\},\{9\} \right] \right) \\[0.5cm]
&= \quad \begin{xy} 
(0,0)*{\circ}="0"; (10,0)*{\circ}="1"; (20,0)*{\circ}="2"; (30,0)*{\circ}="3"; (40,0)*{\circ}="4";(50,0)*{\circ}="5"; (60,0)*{\bullet} = "6"; (70,0)*{\circ} ="7"; (80,0)*{\circ} ="8"; (90,0)*{\circ}= "9"; (100,0)*{\circ} ="10"; (110,0)*{\circ} = "11";
(50,-9.5)*{\smallfrown};(50,-4.5)*{\smallfrown};(70,-4.5)*{\smallfrown};
"0"; (0,-15)**\dir{-}; "5"; (50,-15)**\dir{-}; "11";(110,-15)**\dir{-}; (0,-15);(110,-15)**\dir{-};
"2";(20,-10)**\dir{-}; "7";(70,-10)**\dir{-}; (20,-10);(48.5,-10)**\dir{-};(51.5,-10);(70,-10)**\dir{-};
"3";(30,-5)**\dir{-};"8";(80,-5)**\dir{-}; "10";(100,-5)**\dir{-};
(30,-5);(48.5,-5)**\dir{-};(51.5,-5);(68.5,-5)**\dir{-};(71.5,-5);(100,-5)**\dir{-};
\end{xy}
\end{align*}
\caption{Diagrammatic representation of a typical marked partition $(\ell,\F) \in \Omega(11,6)$.}
\label{fig:markedpartition}
\end{figure}
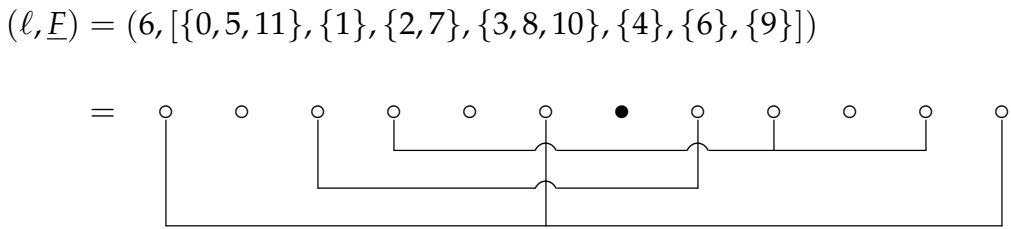

We say that $(\ell,\F) \in \Omega(n,k)$ is \emph{reduced} if for every block $F_i$ we have either $|F_i|>1$ or $F_i = \{\ell\}$, i.e. the partition contains no singleton blocks except possibly $\{\ell\}$. We denote by $\widehat\Omega(n,k) \subset \Omega(n,k)$ the collection of reduced marked partitions. From this point forward when we say `singleton' we will mean a block of the form $\{j\}$ \emph{and} $j\neq \ell$.

Given $(\ell,\F) \in \Omega(n,k)$ with $M$ singleton blocks we construct the corresponding reduced marked partition $(\ell',\F') \in \widehat\Omega(n-M,k-M)$ by removing all singletons and relabelling the remaining elements with the labels $\{0,\dots,n-M\}$ such that the order is preserved. This process is described explicitly below.

Let   
$$0=j_1<\ldots<j_\mu<\ell = j_{\mu+1}< \ldots < j_{\mu+\nu+1} = n$$
be the list of numbers in $\{0,\ldots,n\}$ that lie in non-singleton blocks. 
For $i=1,\dots,\mu+\nu$ we define $m_i = j_{i+1} - j_i -1$ as the number of singletons between $j_{i+1}$ and $j_i$, and set $M_i=m_1+\ldots+m_i$. Thus $M_{\mu+\nu}=M$ is the total number of singletons in $\F$. For $1\leq M_{\mu+\nu}\leq n$ define the map $\kappa_{\F}:\{0,\ldots,n\}\to \{0,\ldots,n'\}$ with $n'=n-M_{\mu+\nu}$ by
\begin{equation}
\kappa_{\F}( j ) = \sup_{i\geq 1}\{ j_i - M_{i-1} \mid j_i \leq j \}.
\end{equation}
If there are no singletons, i.e. $(\ell,\F)$ is already reduced, then $\kappa_{\F}=\id$. Given a pair $(\ell,\F) $ we obtain the reduced marked partition $(\ell', \F')$ by setting $\ell' = \kappa_{\F}(\ell) = \mu$, $\F'= [F'_1,\ldots,F'_{k'}]$ with $k' = k-M_{\mu+\nu}$ and \begin{equation}
F'_i=\{ \kappa_{\F}(j) \mid j\in F_i\} .
\end{equation}
Given $(\ell,\F)\in \Omega(n,k)$ the above provides a unique $(\ell',\F') \in \widehat\Omega(n',k')$. We thus have the following Lemma. 
\begin{lem}
The map
\begin{equation}
 \bigcup_{n=1}^\infty  \Omega(n) \to  \bigcup_{n=1}^\infty \big(\widehat\Omega (n) \times \ZZ_{\geq 0}^{n} \big), \qquad
 (\ell,\F) \mapsto ((\ell',\F'), (m_1,\cdots,m_{n'}))
 \end{equation}
is bijective.
\end{lem}

Furthermore note that every block $F_i\neq\{\ell\}$ in $\widehat\Omega(n)$ contains at least one number strictly less than $\ell$ and at least one number strictly greater than $\ell$. Diagrammatically, the reduction of a marked partition described above then simply corresponds to removing all isolated, unfilled circles - see Figure \ref{fig:reducedmarkedpartition}.

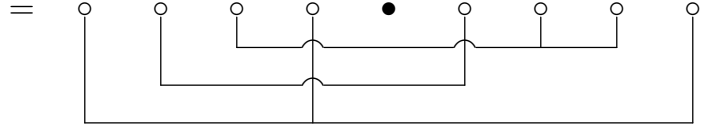
\begin{figure} 
\begin{align*}
(\ell',\F') &= \left( 4, \left[ \{0,3,8\},\{1,5\},\{2,6,7\}, \{4\}\right] \right)\\[0.5cm]
&= \quad \begin{xy} 
(0,0)*{\circ}="0"; (10,0)*{\circ}="1"; (20,0)*{\circ}="2"; (30,0)*{\circ}="3"; (40,0)*{\bullet}="4";(50,0)*{\circ}="5"; (60,0)*{\circ} = "6"; (70,0)*{\circ} ="7"; (80,0)*{\circ} ="8";
(30,-9.5)*{\smallfrown}; (30,-4.5)*{\smallfrown};(50,-4.5)*{\smallfrown};
"0"; (0,-15)**\dir{-}; "3"; (30,-15)**\dir{-}; "8";(80,-15)**\dir{-}; (0,-15);(80,-15)**\dir{-};
"1";(10,-10)**\dir{-}; "5";(50,-10)**\dir{-}; (10,-10);(28.5,-10)**\dir{-};(31.5,-10);(50,-10)**\dir{-};
"2";(20,-5)**\dir{-};"6";(60,-5)**\dir{-}; "7";(70,-5)**\dir{-};
(20,-5);(28.5,-5)**\dir{-};(31.5,-5);(48.5,-5)**\dir{-};(51.5,-5);(70,-5)**\dir{-};
\end{xy}
\end{align*}
\caption{Diagrammatic representation of a marked partition $(\ell',\F') \in \widehat\Omega(8,4)$. Note that $(\ell',\F')$ is obtained from $(\ell,\F)$ in Figure \ref{fig:markedpartition} by removal of all singletons.}
\label{fig:reducedmarkedpartition}
\end{figure}

Assume $k\geq 2$ in the following. Define $\widehat\Omega_{\d}(n,k) \subset \widehat\Omega(n,k)$ to be the set of reduced marked partitions $(\ell,\F)$ such that $0$ and $\ell$ (and therefore also $n$) lie in the same block i.e.,
$$\widehat\Omega_{\d}(n,k) = \{ (\ell,[F_1,\ldots,F_k]) \in \widehat\Omega(n,k) \mid  \{0,\ell,n\} \subset F_i \text{ for some $i$} \}.$$
We call this the set of \emph{diagonal} reduced marked partitions.
The corresponding set of {\em off-diagonal} reduced marked partitions is defined as
$$\widehat\Omega_{\off}(n,k) = \{ (\ell,[F_1,\ldots,F_k]) \in \widehat\Omega(n,k) \mid  \{0,\ell,n\} \not\subset F_i \text{ for any $i$} \}.$$
No order of blocks is specified here, so indeed any marked partition in $\widehat\Omega_{\d}(n,k)$ is either in $\widehat\Omega_{\d}(n,k)$ or in $\widehat\Omega_{\off}(n,k)$.

\subsection{Ordered partitions} \label{subsec:ordered}

We introduce an ordering of a partition $\F$ by specifying an order in which the blocks appear. That is, for a partition $\F=[F_1,\ldots,F_k]$
we have $k!$ corresponding ordered partitions, which we write as $\Fura=\langle F_{\sigma(1)},\ldots,F_{\sigma(k)}\rangle$; here $\sigma\in S_k$ (the symmetric group of $k$ elements). We denote the corresponding set of ordered partitions by $\underrightarrow\Pi(n,k)$. 
We call an ordering {\em canonical} if each $F_i$ contains the smallest of all elements in the blocks $F_j$ with $j\geq i$; in particular this means that $0\in F_1$. This yields a one-to-one correspondence between partitions and canonically ordered partitions.

Given $\Fura\in\underrightarrow\Pi(n,k)$ we define the embedding $\iota_{\Fura}: \RR^{k}\to\RR^{n+1}$, 
\begin{equation}\label{iota111} (x_1,\ldots,x_k)\mapsto (y_0,\ldots,y_n) \quad \text{ where } \quad y_j=x_i \iff j\in F_i. \end{equation}
By abuse of notation, we also define the vector analogue $\iota_{\Fura}: (\RR^d)^{k}\to(\RR^d)^{n+1}$,
\begin{equation}\label{iota222} (\vecx_1,\ldots,\vecx_k)\mapsto (\vecy_0,\ldots,\vecy_n) \quad \text{ where } \quad \vecy_j=\vecx_i \iff j\in F_i. \end{equation}
For the unordered partition $\F\in\Pi(n,k)$ we define the corresponding embedding $\iota_{\F}=\iota_{\Fura}$ where $\Fura$ has the canonical order.  

Let us define  
$$
\underrightarrow\Pi_\circ(n,k) = \{ \langle F_1,\ldots,F_k\rangle \mid [F_1,\ldots,F_k] \in \Pi_\circ(n,k),\; 0\in F_1\} .
$$ 
That is, we specify that the first block contains $0$ (and thus also $n$); with this convention there are $(k-1)!$ ordered partitions in $\underrightarrow\Pi_\circ(n,k)$ for every given $\F\in \Pi_\circ(n,k)$. 
Let $\Pi_\baro(n,k)$ denote the set of partitions of $\{0,\dots,n\}$ into $k$ blocks where $0$ and $n$ lie in different blocks. Define $\underrightarrow\Pi_\baro(n,k)$ to be the set of ordered partitions $\Fura=\langle F_1,\ldots,F_k\rangle$ where $0 \in F_1$ and $n \in F_k$.

Let $\underrightarrow\Pi^\nc(n,k) \subset \underrightarrow\Pi(n,k)$ (and equally $\underrightarrow\Pi_\circ^\nc(n,k) \subset \underrightarrow\Pi_\circ(n,k)$ and $\underrightarrow\Pi_\baro^\nc(n,k) \subset \underrightarrow\Pi_\baro(n,k)$) be the subset of ordered partitions such that all consecutive elements lie in separate blocks, that is $j\in F_i$ implies $j+1\notin F_i$. Let us take $\Fura\in\underrightarrow\Pi_\circ(n,k)$ and construct the corresponding non-consecutive partition; the construction is similar to that of removing singletons. Let $j_1<\ldots<j_{n'}$ be the list of elements $j\in\{1,\ldots,n\}$ for which $j_s\in F_i$ implies $j_s-1\notin F_i$. Let $m_s\in\ZZ_{\geq 0}$ be the largest integer so that $j_s+1,\ldots,j_s+m_s\in F_i$. Then $n=n'+m_1+\ldots+m_{n'}$. We map $\Fura \in \underrightarrow\Pi_\circ(n,k)$ to the pair $(\Fura',\vecm) \in \underrightarrow\Pi^\nc(n',k) \times \ZZ_{\geq 0}^{n'}$, where $F_i'=\{ s \mid j_s\in F_i\}$ and $\vecm=(m_1,\ldots,m_{n'})$ as defined as above. In particular note that if $\Fura \in \underrightarrow\Pi_\circ^\nc(n,k)$ then $(\Fura',\vecm)=(\Fura,\vecnull)$. The map $\Fura \mapsto (\Fura',\vecm)$ is clearly invertible, and we have the following.

\begin{lem}
Let $k\geq 1$. The maps
\begin{equation}\label{A76}
\bigcup_{n=1}^\infty\underrightarrow\Pi(n,k) \to \bigcup_{n=1}^\infty \big( \underrightarrow\Pi^\nc(n,k) \times \ZZ_{\geq 0}^n \big), \qquad \Fura \mapsto (\Fura',\vecm) ,
\end{equation}
\begin{equation}\label{A7}
\bigcup_{n=1}^\infty\underrightarrow\Pi_\circ(n,k) \to \bigcup_{n=1}^\infty \big( \underrightarrow\Pi_\circ^\nc(n,k) \times \ZZ_{\geq 0}^n \big), \qquad \Fura \mapsto (\Fura',\vecm) ,
\end{equation}
\begin{equation}\label{A8}
\bigcup_{n=1}^\infty\underrightarrow\Pi_\baro(n,k) \to \bigcup_{n=1}^\infty \big( \underrightarrow\Pi_\baro^\nc(n,k) \times \ZZ_{\geq 0}^n \big), \qquad \Fura \mapsto (\Fura',\vecm) 
\end{equation}
are bijective.
\end{lem}

The notion of a marked partition $(\ell,\F)$ naturally extends to an ordered marked partition $(\ell,\Fura)$, and we construct a reduced ordered partition from an ordered partition by preserving the order of surviving blocks in the construction. The ordered partitions corresponding to $\Omega(n,k)$ and $\widehat\Omega(n,k)$ are denoted by $\underrightarrow{\Omega}(n,k)$ and $\underrightarrow{\widehat\Omega}(n,k)$, respectively. We define the ordered partitions corresponding to $\widehat\Omega_{\d}(n,k)$ and $\widehat\Omega_{\off}(n,k)$ by
$$\underrightarrow{\widehat\Omega}_{\d}(n,k) = \{ (\ell,\langle F_1,\ldots,F_k\rangle) \in \underrightarrow{\widehat\Omega}(n,k) \mid  0,\ell\in F_1 \}$$
and
$$\underrightarrow{\widehat\Omega}_{\off}(n,k) = \{ (\ell,\langle F_1,\ldots,F_k\rangle) \in \underrightarrow{\widehat\Omega}(n,k) \mid  0\in F_1,\; \ell\in F_k \}.$$
Whereas there are $(k-1)!$ ordered partitions in $\underrightarrow{\widehat\Omega}_{\d}(n,k)$ for every fixed partition in $\widehat\Omega_{\d}(n,k)$, there are only $(k-2)!$ in $\underrightarrow{\widehat\Omega}_{\off}(n,k)$ for every fixed partition in $\widehat\Omega_{\off}(n,k)$. The notion of ordering can also be represented diagrammatically, namely, we insist that if $i < j$ then the line connecting elements in $F_i$ is lower than the line connecting elements in $F_j$ - see Figure \ref{fig:orderedmarkedpartition}. In the event that the marked partition contains a singleton, we can attach a vertical line below to the desired height - see Figure \ref{fig:orderedmarkedpartition2}

\begin{figure} 
\begin{align*}
\left( 3, \left\langle \{0,3,6\},\{1,4\},\{2,5\}\right\rangle \right)&= \quad \begin{xy} 
(0,0)*{\circ}="0"; (10,0)*{\circ}="1"; (20,0)*{\circ}="2"; (30,0)*{\bullet}="3"; (40,0)*{\circ}="4";(50,0)*{\circ}="5"; (60,0)*{\circ} = "6";
(30,-9.5)*{\smallfrown}; (30,-4.5)*{\smallfrown}; (40,-4.5)*{\smallfrown};
"0"; (0,-15)**\dir{-}; "3"; (30,-15)**\dir{-}; "6";(60,-15)**\dir{-}; (0,-15);(60,-15)**\dir{-};
"1";(10,-10)**\dir{-}; "4";(40,-10)**\dir{-}; (10,-10);(28.5,-10)**\dir{-}; (31.5,-10);(40,-10)**\dir{-};
"2";(20,-5)**\dir{-}; "5";(50,-5)**\dir{-}; (20,-5);(28.5,-5)**\dir{-}; (31.5,-5);(38.5,-5)**\dir{-}; (41.5,-5);(50,-5)**\dir{-}
\end{xy} \\
\left( 3, \left\langle \{0,3,6\},\{2,5\},\{1,4\}\right\rangle \right)&= \quad \begin{xy} 
(0,0)*{\circ}="0"; (10,0)*{\circ}="1"; (20,0)*{\circ}="2"; (30,0)*{\bullet}="3"; (40,0)*{\circ}="4";(50,0)*{\circ}="5"; (60,0)*{\circ} = "6";
(20,-4.5)*{\smallfrown}; (30,-9.5)*{\smallfrown}; (30,-4.5)*{\smallfrown};
"0"; (0,-15)**\dir{-}; "3"; (30,-15)**\dir{-}; "6";(60,-15)**\dir{-}; (0,-15);(60,-15)**\dir{-};
"1";(10,-5)**\dir{-}; "4";(40,-5)**\dir{-}; (10,-5);(18.5,-5)**\dir{-}; (21.5,-5);(28.5,-5)**\dir{-}; (31.5,-5);(40,-5)**\dir{-};
"2";(20,-10)**\dir{-}; "5";(50,-10)**\dir{-}; (20,-10);(28.5,-10)**\dir{-}; (31.5,-10);(50,-10)**\dir{-};
\end{xy}
\end{align*}
\caption{Diagrammatic representation showing two ordered marked partitions in $\protect\underrightarrow{\widehat\Omega}_\d (6,3)$ corresponding to the same marked partition.}
\label{fig:orderedmarkedpartition}
\end{figure}
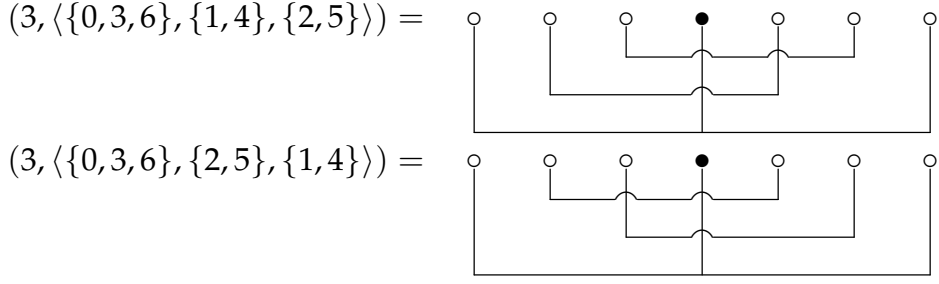

\begin{figure} 
\begin{align*}
\left( 3, \left\langle \{0,4\},\{1\},\{2\}\right\rangle \right)&= \quad \begin{xy} 
(0,0)*{\circ}="0"; (10,0)*{\circ}="1"; (20,0)*{\circ}="2"; (30,0)*{\bullet}="3"; (40,0)*{\circ}="4";
"0"; (0,-15)**\dir{-}; "4"; (40,-15)**\dir{-}; (0,-15);(40,-15)**\dir{-};
"1";(10,-10)**\dir{-}; (9,-10);(11,-10)**\dir{-};
"2";(20,-5)**\dir{-};(19,-5);(21,-5)**\dir{-};
\end{xy} \\
\left( 3, \left\langle \{0,4\},\{2\},\{1\}\right\rangle \right)&= \quad \begin{xy} 
(0,0)*{\circ}="0"; (10,0)*{\circ}="1"; (20,0)*{\circ}="2"; (30,0)*{\bullet}="3"; (40,0)*{\circ}="4";
"0"; (0,-15)**\dir{-}; "4"; (40,-15)**\dir{-}; (0,-15);(40,-15)**\dir{-};
"1";(10,-5)**\dir{-}; (9,-5);(11,-5)**\dir{-};
"2";(20,-10)**\dir{-};(19,-10);(21,-10)**\dir{-};
\end{xy}
\end{align*}
\caption{Diagrammatic representation showing two ordered marked partitions with singletons in $\protect\underrightarrow{\Omega}_\off(4,4)$ corresponding to the same marked partition.}
\label{fig:orderedmarkedpartition2}
\end{figure}
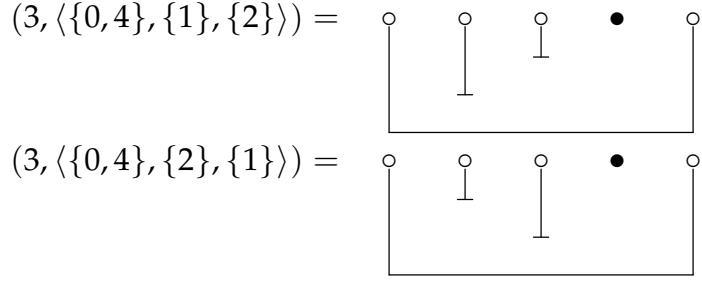

Given $(\ell,\Fura)\in\underrightarrow{\widehat\Omega}$, with $\Fura=\langle F_1,\ldots,F_k\rangle$, we define $\Fura^\pm=\langle F_1^\pm,\ldots,F_k^\pm\rangle$ with $F_i^+=F_i\cap [0,\ell]$ and $F_i^-=n-(F_i \cap [\ell,n])\subset[0,n-\ell]$. Recall that $F_i$ is either the singleton $\{\ell\}$ or contains at least one number strictly less than $\ell$ and at least one number strictly greater than $\ell$. This implies in particular that $F_i^\pm\neq\emptyset$ for all $i$, and thus $\Fura^+\in\Pi(\ell,k)$ and $\Fura^-\in\Pi(n-\ell,k)$. 
We in fact have the following bijection.

\begin{lem}\label{lemA2}
The maps 
\begin{equation}\label{A9}
\underrightarrow{\widehat\Omega}_{\d}(n,k) \to  \bigcup_{m=k-1}^{n-k+1} \big( \underrightarrow\Pi_\circ(m,k) \times \underrightarrow\Pi_\circ(n-m,k) \big) ,
\quad (\ell,\Fura) \mapsto (\Fura^+,\Fura^-) ,
\end{equation}
\begin{equation}\label{A10}
\underrightarrow{\widehat\Omega}_{\off}(n,k) \to \bigcup_{m=k-1}^{n-k+1} \big( \underrightarrow\Pi_\baro(m,k) \times \underrightarrow\Pi_\baro(n-m,k) \big) ,
\quad (\ell,\Fura) \mapsto (\Fura^+,\Fura^-) ,
\end{equation}
are bijective. 
\end{lem}

\begin{proof}
For $(\ell,\Fura)\in\underrightarrow{\widehat\Omega}(n,k)$, we have by construction $0,n\in F_1$; therefore $0\in F_1^+$ and $n\in n-F_1^-$; this in turn implies $0\in F_1^-$. 

If furthermore $(\ell,\Fura)\in\underrightarrow{\widehat\Omega}_{\d}(n,k)$, then $\ell\in F_1$ and thus $\ell\in F_1^+$ and $\ell\in n-F_1^-$; the latter can be written as $n-\ell\in F_1^-$. This shows $\Fura^+\in\Pi_\circ(\ell,k)$ and $\Fura^-\in\Pi_\circ(n-\ell,k)$. 
If on the other hand $(\ell,\Fura)\in\underrightarrow{\widehat\Omega}_{\off}(n,k)$, then $\ell\in F_k$. So $\ell\in F_k^+$ and $\ell\in n-F_k^-$; that is $n-\ell\in F_k^-$. This means that $\Fura^+\in\Pi_\baro(\ell,k)$ and $\Fura^-\in\Pi_\baro(n-\ell,k)$. We conclude that \eqref{A9}, \eqref{A10} have the correct range.

Now $(\Fura^+,\Fura^-)$ are uniquely determined by $\Fura$ and hence \eqref{A9}, \eqref{A10} are injective. The inverse maps are given by 
$$
(\Fura^+,\Fura^-)\mapsto ( \nu(\F^+) , \langle F_1^+\cup (n-F_1^-), \ldots, F_k^+\cup (n-F_k^-) \rangle)
$$
from which we infer that \eqref{A9}, \eqref{A10} are surjective.
\end{proof}

\begin{figure} 
\begin{align*}
(\ell,\F) &= \left( 4, \left\langle \{0,3,8\},\{1,5\},\{2,6,7\}, \{4\} \right\rangle] \right)\\[0.5cm]
&= \quad \begin{xy} 
(0,0)*{\circ}="0"; (10,0)*{\circ}="1"; (20,0)*{\circ}="2"; (30,0)*{\circ}="3"; (40,0)*{\bullet}="4";(50,0)*{\circ}="5"; (60,0)*{\circ} = "6"; (70,0)*{\circ} ="7"; (80,0)*{\circ} ="8";
(30,-9.5)*{\smallfrown}; (30,-4.5)*{\smallfrown};(50,-4.5)*{\smallfrown};
"0"; (0,-15)**\dir{-}; "3"; (30,-15)**\dir{-}; "8";(80,-15)**\dir{-}; (0,-15);(80,-15)**\dir{-};
"1";(10,-10)**\dir{-}; "5";(50,-10)**\dir{-}; (10,-10);(28.5,-10)**\dir{-};(31.5,-10);(50,-10)**\dir{-};
"2";(20,-5)**\dir{-};"6";(60,-5)**\dir{-}; "7";(70,-5)**\dir{-};
(20,-5);(28.5,-5)**\dir{-};(31.5,-5);(48.5,-5)**\dir{-};(51.5,-5);(70,-5)**\dir{-};
"4";(40,-2.5)**\dir{-}; (39,-2.5);(41,-2.5)**\dir{-};
\end{xy} \\[0.5cm]
\underrightarrow{F}^- &= \langle\{0,3\},\{1\},\{2\},\{4\} \rangle = \quad \begin{xy} 
(0,0)*{\circ}="0"; (10,0)*{\circ}="1"; (20,0)*{\circ}="2"; (30,0)*{\circ}="3"; (40,0)*{\circ}="4";
"0"; (0,-15)**\dir{-}; "3"; (30,-15)**\dir{-}; (0,-15);(30,-15)**\dir{-};
"1";(10,-10)**\dir{-};(9,-10);(11,-10)**\dir{-};
"2";(20,-5)**\dir{-};(19,-5);(21,-5)**\dir{-};
"4";(40,-2.5)**\dir{-}; (39,-2.5);(40,-2.5)**\dir{-};
\end{xy} \\
\underrightarrow{F}^+ &= \langle\{0\},\{3\},\{1,2\},\{4\} \rangle = \quad \begin{xy} 
(0,0)*{\circ}="0"; (10,0)*{\circ}="1"; (20,0)*{\circ}="2"; (30,0)*{\circ}="3"; (40,0)*{\circ}="4";
"0"; (0,-15)**\dir{-}; (0,-15);(1,-15)**\dir{-};
"3";(30,-10)**\dir{-};(29,-10);(31,-10)**\dir{-};
"1";(10,-5)**\dir{-};"2";(20,-5)**\dir{-}; (10,-5);(20,-5)**\dir{-};
"4";(40,-2.5)**\dir{-}; (39,-2.5);(40,-2.5)**\dir{-};
\end{xy}
\end{align*}
\caption{Decomposition of an ordered marked partition in $\protect\underrightarrow{\widehat\Omega}_\off(8,4)$ into ordered partitions $\protect\underrightarrow{F}^- \in \protect\underrightarrow{\Pi}_\baro(4,4)$ and $\protect\underrightarrow{F}^+ \in \protect\underrightarrow{\Pi}_\baro(4,4)$.}
\label{fig:decomposition}
\end{figure}

\begin{figure} 
\begin{align*}
(\ell,\F) &= \left( 4, \left\langle \{0,4,6,8\},\{1,3,5\},\{2,7\} \right\rangle] \right)\\[0.5cm]
&= \quad \begin{xy} 
(0,0)*{\circ}="0"; (10,0)*{\circ}="1"; (20,0)*{\circ}="2"; (30,0)*{\circ}="3"; (40,0)*{\bullet}="4";(50,0)*{\circ}="5"; (60,0)*{\circ} = "6"; (70,0)*{\circ} ="7"; (80,0)*{\circ} ="8";
(30,-4.5)*{\smallfrown}; (40,-9.5)*{\smallfrown}; (40,-4.5)*{\smallfrown};(50,-4.5)*{\smallfrown};(60,-4.5)*{\smallfrown};
"0"; (0,-15)**\dir{-}; "4";(40,-15)**\dir{-}; "6";(60,-15)**\dir{-};  "8";(80,-15)**\dir{-}; (0,-15);(80,-15)**\dir{-};
"1";(10,-10)**\dir{-}; "3"; (30,-10)**\dir{-}; "5";(50,-10)**\dir{-}; 
(10,-10);(38.5,-10)**\dir{-};(41.5,-10);(50,-10)**\dir{-};
"2";(20,-5)**\dir{-};"7";(70,-5)**\dir{-};
(20,-5);(28.5,-5)**\dir{-};(31.5,-5);(38.5,-5)**\dir{-};(41.5,-5);(48.5,-5)**\dir{-};(51.5,-5);(58.5,-5)**\dir{-};(61.5,-5);(70,-5)**\dir{-};
\end{xy} \\[0.5cm]
\underrightarrow{F}^- &= \langle\{0,4\},\{1,3\},\{2\} \rangle = \quad \begin{xy} 
(0,0)*{\circ}="0"; (10,0)*{\circ}="1"; (20,0)*{\circ}="2"; (30,0)*{\circ}="3"; (40,0)*{\circ}="4";
"0"; (0,-15)**\dir{-}; "4";(40,-15)**\dir{-};
(0,-15);(40,-15)**\dir{-};
"1";(10,-10)**\dir{-};"3"; (30,-10)**\dir{-}; 
(10,-10);(30,-10)**\dir{-};
"2";(20,-5)**\dir{-}; 
(19,-5);(21,-5)**\dir{-};
\end{xy} \\
\underrightarrow{F}^+ &= \langle\{0,2,4\},\{3\},\{1\} \rangle = \quad \begin{xy} 
(0,0)*{\circ}="0"; (10,0)*{\circ}="1"; (20,0)*{\circ}="2"; (30,0)*{\circ}="3"; (40,0)*{\circ}="4";
"0"; (0,-15)**\dir{-}; "2";(20,-15)**\dir{-}; "4";(40,-15)**\dir{-};
(0,-15);(40,-15)**\dir{-};
"3";(30,-10)**\dir{-};
(29,-10);(31,-10)**\dir{-};
"1";(10,-5)**\dir{-};
(9,-5);(11,-5)**\dir{-};
\end{xy}
\end{align*}
\caption{Decomposition of an ordered marked partition in $\protect\underrightarrow{\widehat\Omega}_\d(8,3)$ into ordered partitions $\protect\underrightarrow{F}^- \in \protect\underrightarrow{\Pi}_\circ(4,3)$ and $\protect\underrightarrow{F}^+ \in \protect\underrightarrow{\Pi}_\circ(4,3)$.}
\label{fig:decomposition23456789}
\end{figure}
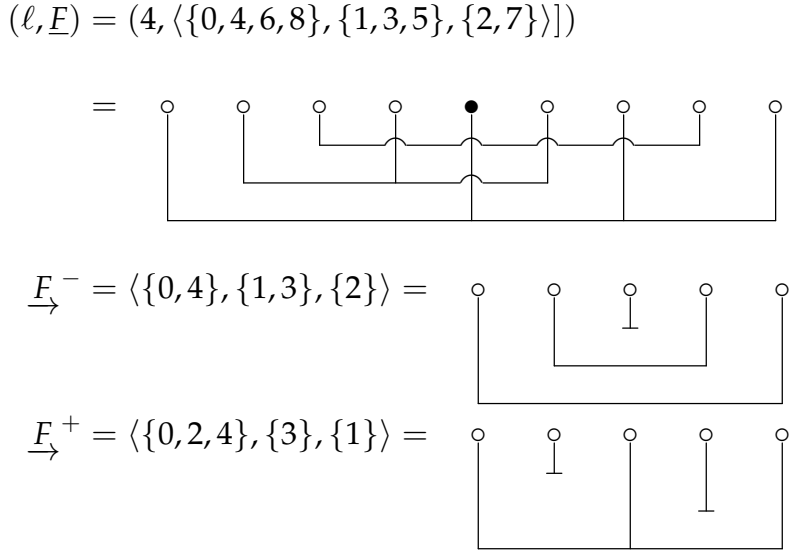

\subsection{Graphs and paths}\label{subsec:GaP}

Let $\KK_k$ be the complete graph with $k\geq 2$ vertices which we label as $1,2,\ldots,k$, and denote by $P=i_0 i_1\cdots i_n$ (with $i_{s}\neq i_{s+1}$) the path of length $n$ which visits the listed vertices in the given order. This means that consecutive indices $i_{s}i_{s+1}$ correspond to an edge. Backtracking is allowed, e.g. $P=121$ is an admissible path. We denote by $\Sigma(n,k)$ the set of all paths of length $n$, and by $\overline\Sigma(n,k)$ the set of such paths that visit every vertex at least once. $\overline\Sigma(n,k)$ is non-empty for $n\geq k-1$. We denote by $\Sigma_{ij}(n,k)\subset \Sigma(n,k)$ resp.\ $\overline\Sigma_{ij}(n,k)\subset \overline\Sigma(n,k)$ the subset of paths that start at vertex $i$ and end at vertex $j$. For $i\neq j$, $\overline\Sigma_{ij}(n,k)$ is non-empty if $n\geq k-1$, and $\overline\Sigma_{ii}(n,k)$ is non-empty if $n\geq k$.

With an ordered partition $\Fura\in \underrightarrow\Pi^\nc(n,k)$ we associate a path $P(\Fura)=i_0 i_1\cdots i_n\in\overline\Sigma(n,k)$ by identifying each block $F_i$ with the $i$th vertex of $K_n$. That is, $i_s=i$ if and only if $s\in F_i$. This yields the following.

\begin{lem}\label{lemB1}
The maps 
\begin{equation}\label{eqB1}
\underrightarrow\Pi^\nc(n,k) \to \overline\Sigma(n,k),
\quad \Fura \mapsto P(\Fura) ,
\end{equation}
\begin{equation}\label{eqB2}
\underrightarrow\Pi_\circ^\nc(n,k) \to \overline\Sigma_{11}(n,k),
\quad \Fura \mapsto P(\Fura) ,
\end{equation}
\begin{equation}\label{eqB3}
\underrightarrow\Pi_\baro^\nc(n,k) \to \overline\Sigma_{1k}(n,k),
\quad \Fura \mapsto P(\Fura) ,
\end{equation}
are bijective. 
\end{lem}

We assign a $(k\times k)$ edge matrix $\WW=(w_{ij})$ to the graph $\KK_k$, by assigning a {\em weight} $w_{ij}$ to each edge $ij$, and $w_{ii}=0$ to the coefficients on the diagonal. We also assign the diagonal vertex matrix $\DD(\vecu)=\diag(u_1,\ldots,u_k)$ which assigns weight $u_i$ to vertex $i$. The {\em edge weight} of a path $P=i_0\ldots i_n$ is then defined as $w_P=w_{i_0i_1} w_{i_1i_2}\cdots w_{i_{n-1}i_n}$, and the {\em vertex weight} as $u_P=u_{i_0} u_{i_1} \cdots u_{i_{n-1}} u_{i_n}$; the {\em total weight} is thus $$u_P w_P=u_{i_0} w_{i_0i_1} u_{i_1} w_{i_1i_2}\cdots u_{i_{n-1}} w_{i_{n-1}i_n} u_{i_n}.$$

The combinatorics of a path of length $n$ can be understood by means of the matrix $[\DD(\vecu) \WW]^n \DD(\vecu)$, as we will explain now. We define the linear operator $\scrL$ acting on functions $F$ of $\CC^k$ by
\begin{equation}\label{eqLF}
\scrL F (\vecu) = \frac{1}{(2\pi\i)^k} \ointccw\cdots \ointccw F(\vecz) \, \exp\bigg(\frac{u_1}{z_1}+\cdots+\frac{u_k}{z_k}\bigg)\, \frac{\d z_1}{z_1^2} \cdots \frac{\d z_k}{z_k^2} .
\end{equation}
More explicitly, if $F$ has the Taylor series
\begin{equation}
F(\vecu)=\sum_{\nu_1,\ldots,\nu_k=0}^\infty C_{\nu_1,\ldots,\nu_k} u_1^{\nu_1}\cdots u_k^{\nu_k},
\end{equation}
then
\begin{equation}\label{TaylorF}
\scrL F(\vecu)= \sum_{\nu_1,\ldots,\nu_k=1}^\infty \frac{C_{\nu_1,\ldots,\nu_k}}{(\nu_1-1)!\cdots(\nu_k-1)!} u_1^{\nu_1-1}\cdots u_k^{\nu_k-1} .
\end{equation}
Thus $\scrL F(\vecu)$ is the derivative $\partial_{u_1}\cdots\partial_{u_k}$ of the Borel transform of $F(\vecu)$.

\begin{lem} For $n\geq 1$,
\begin{equation}\label{dudu}
\scrL \sum_{P\in\overline\Sigma_{ij}(n,k)}  u_P w_P = 
\scrL \;  \big( [\DD(\vecu) \WW]^n \DD(\vecu)\big)_{ij} ,
\end{equation}
\end{lem}

\begin{proof}
Note that
\begin{equation}
\scrL \sum_{P\in\overline\Sigma_{ij}(n,k)}  u_P w_P = 
\scrL \sum_{P\in\Sigma_{ij}(n,k)}  u_P w_P  ,
\end{equation}
since the terms in $\sum_{P}  u_P w_P$ that are constant in $u_l$ correspond exactly to those paths that do not visit vertex $l$. Eq.~\eqref{dudu} then follows from simple matrix algebra.
\end{proof}

Consider the $(k\times k)$ matrix valued function 
\begin{equation}
F(\vecu)=\big( 1- [\DD(\vecu) \WW])^{-1} \DD(\vecu) =(\DD(\vecu)^{-1} - \WW)^{-1}  .
\end{equation}
We have the series expansion
\begin{equation}\label{serexp}
F(\vecu) = \sum_{n=0}^\infty [\DD(\vecu) \WW]^n \DD(\vecu)= \sum_{\nu_1,\ldots,\nu_k=0}^\infty C_{\nu_1,\ldots,\nu_k} u_1^{\nu_1}\cdots u_k^{\nu_k},
\end{equation}
with the Taylor coefficients given by $(k\times k)$ matrices $C_{\nu_1,\ldots,\nu_k}$.
The first series in \eqref{serexp} (and hence also the above Taylor series) converges absolutely for $|u_i|<r_0^{-1}$ with $r_0= k \max |w_{ij}|$.
Summing over $n$ and using the geometric series (note that for the $n=0$ term $\scrL \DD(\vecu)=0$), we then have
\begin{equation}
\scrL \sum_{n=1}^\infty \sum_{P\in\overline\Sigma_{ij}(n,k)}  u_P w_P = 
[\scrL F_{ij}](\vecu) = [\scrL F]_{ij}(\vecu) .
\end{equation}
After the variable substitution $z_i\mapsto 1/z_i$, \eqref{eqLF} becomes
\begin{equation}\label{eqLF2}
\GG^{(k)}(\vecu):=\scrL F (\vecu) = \frac{1}{(2\pi\i)^k} \ointccw\cdots \ointccw (\DD(\vecz)-\WW)^{-1}  \, \exp(\vecu\cdot\vecz)\, \d z_1\cdots \d z_k .
\end{equation}

Let us work out the example $k=2$ in some more detail. In this case
\begin{equation}
(\DD(\vecz)-\WW)^{-1} = \begin{pmatrix} z_1 & -w_{12} \\  -w_{21} & z_2 \end{pmatrix}^{-1}
= \frac{1}{z_1 z_2 - w_{12}w_{21}} \begin{pmatrix} z_2 & w_{12} \\  w_{21} & z_1 \end{pmatrix}
\end{equation}
and hence
\begin{equation}\label{eqLF2222}
\GG^{(2)}(\vecu) = \frac{1}{(2\pi\i)^2} \ointccw\ointccw \frac{1}{z_2 - z_1^{-1} w_{12}w_{21}} \begin{pmatrix} z_2 & w_{12} \\  w_{21} & z_1 \end{pmatrix} \exp(u_1z_1+u_2z_2)\, \d z_2\, \frac{\d z_1}{z_1} .
\end{equation}
We contract the path of integration in $z_2$ to zero, and pick up the residue at $z_2 = z_1^{-1} w_{12}w_{21}$. This yields
\begin{equation}
\GG^{(2)}(\vecu) = \frac{1}{2\pi\i} \ointccw \begin{pmatrix} z_1^{-1} w_{12}w_{21} & w_{12} \\  w_{21} & z_1 \end{pmatrix} \exp(u_1z_1+u_2 z_1^{-1} w_{12}w_{21})\, \frac{\d z_1}{z_1} .
\end{equation}

For $u_1,u_2\neq 0$ set $z_1 = \i u_1^{-1/2} u_2^{1/2} (-w_{12} w_{21})^{1/2} \e^{\i\theta}$ with any fixed choice for the branch of the square-root so that $\d \vecz_1/\vecz_1 = \i \d \theta$ . Then
\begin{multline}
\GG^{(2)}(\vecu) = \frac{1}{2\pi} \int_0^{2\pi}  
\begin{pmatrix} \i u_1^{1/2} u_2^{-1/2} (-w_{12} w_{21})^{1/2} \e^{-\i\theta} 
& w_{12} \\ w_{21}  &\i u_1^{-1/2} u_2^{1/2} (-w_{12} w_{21})^{1/2} \e^{\i\theta}  \end{pmatrix}
\\
\times \exp\bigg(\i u_1^{1/2} u_2^{1/2}(-w_{12}w_{21})^{1/2}  (\e^{\i\theta}+\e^{-\i\theta})\bigg)\, \d\theta  .
\end{multline}
Compare this with the classical integral representation for the $J$-Bessel function,
\begin{equation} \label{JB-def}
J_n(z)=\frac{1}{2\pi} \int_0^{2\pi} \e^{\i z \sin\theta - \i n\theta} d\theta=\frac{i^{-n}}{2\pi} \int_0^{2\pi} \e^{\i z \cos\theta -\i n\theta} d\theta.
\end{equation}
Thus
\begin{equation}\label{eqk2A1}
g_{11}^{(2)}(\vecu) = - u_1^{1/2} u_2^{-1/2} (-w_{12}w_{21})^{1/2} \,J_1(2 u_1^{1/2} u_2^{1/2}(-w_{12}w_{21})^{1/2} ) ,
\end{equation}
\begin{equation}\label{eqk2A2}
g_{12}^{(2)}(\vecu) = w_{12} \, J_0(2 u_1^{1/2} u_2^{1/2}(-w_{12}w_{21})^{1/2} ) ,
\end{equation}
\begin{equation}
g_{21}^{(2)}(\vecu) = w_{21} \, J_0(2 u_1^{1/2} u_2^{1/2}(-w_{12}w_{21})^{1/2} )  ,
\end{equation}
\begin{equation}\label{eqk2A122}
g_{22}^{(2)}(\vecu) =  -u_1^{-1/2} u_2^{1/2} (-w_{12}w_{21})^{1/2} \,J_1(2 u_1^{1/2} u_2^{1/2}(-w_{12}w_{21})^{1/2} ) .
\end{equation}

\end{appendix}

\end{document}